\documentclass[sigplan,10pt]{acmart}

\renewcommand\paragraph[1]{\noindent{\bf #1} \hspace{0.1in}} 
\usepackage{booktabs} 

\usepackage{lscape}
\usepackage{color,colortbl}
\usepackage{todonotes}
\usepackage{bcprules}
\usepackage{program}
\usepackage{dirtree}
\usepackage[normalem]{ulem}
\usepackage{algorithmicx}
\newcommand\Loadedframemethod{default}
\usepackage[framemethod=\Loadedframemethod]{mdframed}
\usepackage{enumitem}
\usepackage[nodisplayskipstretch]{setspace}
\usepackage[belowskip=-10pt,aboveskip=3pt]{caption}

\usepackage{listings}
\usepackage{lstlinebgrd}
\let\origthelstnumber\thelstnumber
\makeatletter
\newcommand*\Suppressnumber{%
  \lst@AddToHook{OnNewLine}{%
    \let\thelstnumber\relax%
     \advance\c@lstnumber-\@ne\relax%
    }%
}

\newcommand*\Reactivatenumber{%
  \lst@AddToHook{OnNewLine}{%
   \let\thelstnumber\origthelstnumber%
   \advance\c@lstnumber\@ne\relax}%
}

\lstdefinestyle{nonumbers}
               {numbers=none}
               
\newcommand{\knotical}{\textsc{Knotical}}
\newcommand{\ocaml}{OCaml}
\newcommand{\symkat}{\textsc{Symkat}}
\newcommand{\symdiff}{\textsc{SymDIFF}}
\newcommand{\interproc}{\textsc{Interproc}}

\newcommand{\proj}[2]{\mathsf{proj}_{#1}(#2)}

\newcommand{\states}{\mathcal{S}}

\newcommand{\powset}[1]{\mathcal{P}(#1)}

\newcommand{\true}{\mathsf{true}}
\newcommand{\false}{\mathsf{false}}

\usepackage{ifthen}

\newcommand{\preformula}[3]{
	\ifthenelse{\equal{#1}{}}
	{
		\ifthenelse{\equal{#2}{}}
		{\varphi_{\textsf{pre}}^{#3}}
		{\varphi_{\textsf{pre},#2}^{#3}}
	}
	{
		\ifthenelse{\equal{#2}{}}
		{#1_{\textsf{pre}}^{#3}}
		{#1_{\textsf{pre},#2}^{#3}}
	}
}

\newcommand{\tlformula}[3]{
	\ifthenelse{\equal{#1}{}}
	{
		\ifthenelse{\equal{#2}{}}
		{\varphi_{\alpha}^{#3}}
		{\varphi_{\alpha,#2}^{#3}}
	}
	{
		\ifthenelse{\equal{#2}{}}
		{#1_{\alpha}^{#3}}
		{#1_{\alpha,#2}^{#3}}
	}
}

\newcommand{\postformula}[3]{
	\ifthenelse{\equal{#1}{}}
	{
		\ifthenelse{\equal{#2}{}}
		{\varphi_{\textsf{post}}^{#3}}
		{\varphi_{\textsf{post},#2}^{#3}}
	}
	{
		\ifthenelse{\equal{#2}{}}
		{#1_{\textsf{post}}^{#3}}
		{#1_{\textsf{post},#2}^{#3}}
	}
}

\newcommand{\progs}{\mathsf{Prog}}
\newcommand{\configs}{\mathsf{Config}}
\newcommand{\fault}{\mathsf{fault}}
\newcommand{\config}[2]{\langle{#1},{#2}\rangle}
\newcommand{\hypotheses}{\mathcal{A}}
\newcommand{\hypothesesB}{\mathcal{B}}
\newcommand{\hypothesesD}{\mathcal{D}}
\newcommand{\hypothesesclass}{\mathfrak{A}}
\newcommand{\refines}[2]{{#1}\preceq{#2}}
\newcommand{\absrefines}[4]{{#1}\preceq^{#3}_{#4}{#2}}

\newcommand{\abskrefines}[3]{{#1}\preceq^{#3}{#2}}
\newcommand{\abskhyporefines}[4]{{#1}\preceq^{#3}_{#4}{#2}}

\newcommand{\absstatekrefines}[5]{
	\ifthenelse{\equal{#5}{}}
	{
		{#1}\,\,{}_{\textsf{st}}\!\!\sqsubseteq^{#3}_{#4}{#2}
	}
	{
		{#1}\,\,{}^{#5}_{\textsf{st}}\!\!\sqsubseteq^{#3}_{#4}{#2}
	}
}

\newcommand{\interfaceFunc}{\textsf{RefRelation}}
\newcommand{\interfaceT}{\mathbb{T}}
\newcommand{\interfaceRel}{\interfaceT}

\newcommand{\interfaceOpStar}[1]{#1^{\star}}
\newcommand{\interfaceOpComp}{\odot}
\newcommand{\interfaceOpDisj}{\oplus}
\newcommand{\interfaceOpTrans}{\otimes}

\newcommand{\kat}{\mathcal{K}}
\newcommand{\bkat}{\mathcal{B}}

\newcommand{\kstar}[1]{{#1}^*}
\newcommand{\kneg}[1]{\bar#1}
\newcommand{\actions}{\textsf{P}}
\newcommand{\tests}{\textsf{B}}
\newcommand{\boolgr}{\textsf{BExp}}
\newcommand{\katgr}{\textsf{KExp}}
\newcommand{\grdef}{::=}

\newcommand{\katdiff}[3]{#1 \Delta_{#3} #2}
\newcommand{\katminus}[3]{#1 \setminus_{#3} #2}

\newcommand{\kM}[2]{\textsf{#1}_{\bf \textsf{#2}}}

\newcommand{\event}[1]{\texttt{\ttfamily\fontseries{b}\selectfont#1}}
\newcommand{\FALSE}{\textsf{false}}
\newcommand{\TRUE}{\textsf{true}}

\newcommand{\statetobool}[1]{\alpha(#1)}
\newcommand{\translate}[2]{\atranslate(#1,#2)}
\newcommand{\abstractionRefinement}{\sqsubseteq}
\newcommand{\abstractionCombinedRefinement}{\sqcup}

\newcommand\defeq{\stackrel{\triangle}{=}}

\newcommand{\asearch}{\textsc{Synth}}
\newcommand{\atranslate}{\textsc{Translate}}
\newcommand{\akatdiff}{\textsc{KATdiff}}
\newcommand{\asolvediff}{\textsc{SolveDiff}}
\newcommand{\arestrict}{\textsc{Restrict}}
\newcommand{\acexscore}{\textsc{Distance}}
\newcommand{\acexs}{\textsf{cexs}}

\newenvironment{theoremApp}[1]{{\noindent\bf Theorem \ref{#1} (restated).} \it}{}

\newenvironment{corollaryApp}[1]{{\noindent\bf Corollary \ref{#1} (restated).} \it}{}

\newtheorem{remark}{Remark}[section]

\newenvironment{itemize*}%
  {\begin{itemize}[leftmargin=10pt]%
    \setlength{\itemsep}{0.0in}%
    \setlength{\topsep}{0.0in}%
    \setlength{\parskip}{0.0in}}%
  {\end{itemize}}
\newenvironment{enumerate*}%
  {\begin{enumerate}%
    \setlength{\itemsep}{0.0in}%
    \setlength{\topsep}{0.0in}%
    \setlength{\parskip}{0.0in}}%
  {\end{enumerate}}

\definecolor{light-gray}{gray}{0.85}
\newcommand{\timos}[1]{\noindent{\color{magenta} Timos: #1}}

\newcommand\red[1]{{\color{red} #1}}

\newcommand\Eremoved[1]{{\color{red}\sout{#1}}}
\newcommand\Tadded[1]{{\color{olive} #1}}
\newcommand\Tremoved[1]{{\color{orange}\sout{#1}}}

\newcommand\Cremoved[1]{{\color{magenta}\sout{#1}}}

\newcommand\ignore[1]{}

\renewcommand\Eremoved[1]{}
\renewcommand\Tremoved[1]{}
\renewcommand\Cremoved[1]{}
\renewcommand\timos[1]{}
\renewcommand\todo[1]{}

\newcommand\ie{{\it i.e.}}
\newcommand\eg{{\it e.g.}}

\newcommand\FULL[1]{\red{(omitted content)}}

\newcommand\omittedproofref[1]{The proof can be found in Appendix #1 in the supplemental materials.}
\newcommand\omittedproofreflc[1]{the proof can be found in Appendix #1 in the supplemental materials.}

\renewcommand\omittedproofref[1]{(See Apx.~#1)}
\renewcommand\omittedproofreflc[1]{the proof is in Apx.~#1.}

\newcommand\numbenchmarks{37}

\usepackage[ruled]{algorithm2e} 

\SetAlFnt{\small}
\SetAlCapFnt{\small}
\SetAlCapNameFnt{\small}
\SetAlCapHSkip{0pt}
\IncMargin{-\parindent}

\acmConference[PL'18]{ACM SIGPLAN Conference on Programming Languages}{January 01--03, 2018}{New York, NY, USA}
\acmYear{2018}
\acmISBN{} 
\acmDOI{} 
\startPage{1}

\setcopyright{none}
\acmDOI{0000001.0000001}

\received{July 2018}

\begin{document}
\title[Trace Refinement Relations]{Specification and Inference of \\ Trace Refinement Relations}

\thanks{This work was supported in part by the National Science Foundation
(NSF) award \#1618542 and the Office of Naval Research (ONR) award \#N000141712787.}

\author{Timos Antonopoulos}
\affiliation{
  \institution{Yale University}            
}
\email{timos.antonopoulos@yale.edu}          

\author{Eric Koskinen}
\affiliation{
  \institution{Stevens Institute of Technology}
}
\email{eric.koskinen@stevens.edu}

\author{Ton-Chanh Le}
\affiliation{
  \institution{Stevens Institute of Technology}
}
\email{tchanle@stevens.edu}

\begin{abstract}
Modern software is constantly changing. Researchers and practitioners are
increasingly aware that verification tools can be impactful if they embrace
change through analyses that are compositional and span program versions.
Reasoning about similarities and differences between programs goes back to
Benton~\cite{benton}, who introduced state-based \emph{refinement relations},
which were extended by Yang~\cite{DBLP:journals/tcs/Yang07} and
others~\cite{GLP2017,UTS17}. However, to our knowledge, refinement relations
have not been explored for \emph{traces}: existing techniques, including
bisimulation, cannot capture similarities/differences between how two programs
behave over time.

We present a novel theory that allows one to perform compositional reasoning
about the similarities/differences between how fragments of two different
programs behave over time through the use of what we call \emph{trace-refinement
relations}. We take a reactive view of programs and found Kleene Algebra with
Tests (KAT)~\cite{kozen} to be a natural choice to describe traces since it
permits algebraic reasoning and has built-in composition. Our theory involves a
two-step semantic abstraction from programs to KAT, and then our trace
refinement relations correlate behaviors by (i) categorizing program behaviors
into \emph{trace classes} through KAT intersection and (ii) correlating atomic
events/conditions across programs with KAT hypotheses. We next describe a
synthesis algorithm that iteratively constructs trace-refinement relations
between two programs by exploring sub-partitions of their traces, iteratively
abstracting them as KAT expressions, discovering relationships through a custom
edit-distance algorithm, and applying strategies (i) and (ii) above. We have
implemented this algorithm as {\knotical}, the first tool capable of
synthesizing trace-refinement relations. It built from the ground up in
{\ocaml}, using {\interproc}~\cite{interproctool} and
{\symkat}~\cite{symkattool}. We have demonstrated that useful relations can be
efficiently generated across a suite of \numbenchmarks\ benchmarks that include
changing fragments of array programs, systems code, and web servers.

\end{abstract}

%
%


%
%


\maketitle

\renewcommand{\shortauthors}{Antonopoulos et al.}



\section{Introduction}
\label{sec:intro}

Modern software changes at a rapid pace. Software engineering
  practices, \eg{}~ Agile, advocate a view that software is an evolutionary
  process, a series of source code edits that lead, slowly but surely, toward an
  improved system. Meanwhile, as these software systems grow, fragments of
  code are reused in increasingly many different contexts. To make matters
  worse, these contexts themselves may be changing, and code written under some
  assumptions today may be used under different ones tomorrow. With so many
  moving parts and adoption of formal methods an uphill battle, now more than
  ever, researchers and practitioners have found  compositionality and
  reasoning across versions~\cite{continuousreasoning,vmv}.
  to be indispensable.

Changes can be exploited for good purposes: they offer a sort of informal specification,
where programmers often view their new code in terms of how it has
deviated from the existing code (\ie\ a commit or patch), including
the removal of bugs, addition of new features, performance improvements, etc.
With compositional theories and tools, one can reuse previous analysis results for
unchanged code, and combine them with new analyses of only the changing code fragment.


It is therefore a natural question to ask: how does a
given program $C_1$ compare to $C_2$, a modified version of $C_1$?
If one is merely interested in knowing whether they are strictly
equivalent (or whether $C_2$ is contained within $C_1$),
this is a classical notion of \emph{concrete program refinement}~\cite{morgan1994}
and includes compiler correctness and
translation validation~\cite{transvalid}.
Intuitively, $C_1$ concretely refines $C_2$ provided that, when executed from the same initial state, they both reach the same final state.
Researchers have developed
  algorithms and tools (\eg~\cite{LahiriHKR2012,LahiriMSH2013,WDLE2017}) to check whether, say, two versions of a function return the same results.
Similarly, bisimulation provides \emph{equivalence} between how
  programs behave over time, perhaps accounting for different implementations.


Concrete refinement and bisimulation are not typically focused on how the programs differ, but simply whether or not they are equivalent.
In his canonical 2004 work~\cite{benton}, Benton weakened classical
refinement, allowing one to define equivalence relations over the state
space, so that the two programs reach the same output equivalence
relation when executed from states in a particular input
equivalence relation. Such equivalences allow one to
describe what differences over the states one does or does not care about, for
example, focusing on important variables or ignoring scratch data.
This strategy is compositional because one can
correlate the output relation of one code region with the input relation of
the next. Benton's work was
later extended by Yang~\cite{DBLP:journals/tcs/Yang07}
and others~\cite{GLP2017,UTS17}.

Benton's work focuses on \emph{state} 
relations. To the best of
our knowledge, no prior work has explored the question 
\emph{trace}-oriented relations to express
similarities/differences between the execution behaviors of $C_1$ and
$C_2$ over time. 
Examples include whether two programs send/receive messages in the
same order, follow the same allocation/release orders, or have certain
I/O patterns.
Encoding these kinds of properties in state relations with, \eg, history variables seems taxing, yet
 a trace approach would be appealing because it would be more granular and hopefully lead to more flexibility.

\paragraph{Toward trace refinement relations.}
We take a reactive
view of programs, treating their execution in terms of \emph{events},
which can be suitably defined in terms of statements, function calls,
I/O, etc.~as needed. 
We considered a few options for characterizing traces.
One choice would be a
temporal logic such as LTL or CTL, perhaps
using the LTL ``chop'' operator~\cite{DBLP:conf/stoc/BarringerKP84}
to support composition.
%
We found Kleene Algebra with Tests (KAT)~\cite{kozen} to be
a natural choice for a few reasons.
Briefly, KAT is an amalgamation of Kleene Algebra which (like regular expressions) has constructors for union +, concatenation $\cdot$, and star-iteration *, and boolean algebra which has boolean predicates and operations. KAT expressions consist of a combination of event symbols (herein denoted in uppercase: \textsf{A}, \textsf{B}, \textsf{C}) and boolean ``test'' symbols (herein denoted in lowercase: \textsf{a}, \textsf{b}, \textsf{c}).
One can write KAT expressions that mix event symbols with boolean test symbols, \eg, $(\textsf{a}\!\cdot\!\textsf{A})^*\!\cdot\!\textsf{C}\!\cdot\!\textsf{b}$.
A KAT can be used as a helpful abstraction of simple abstract syntax trees. For example, the KAT expression $(\textsf{b}\!\cdot\!\textsf{C}\;+\;\overline{\textsf{b}}\!\cdot\!\textsf{D})^*$ models a program that is a multi-path loop, branching on \textsf{b}.
%
In this way, programs can be abstracted into KAT expressions using
a syntactic translation~\cite{kozen,Kozen06} or, more generally using semantic translations, as we discuss later in this paper.
%
Also, KAT is appealing because it supports algebraic reasoning, has a natural composition
operator, allows hypothesis introduction,
and there is existing tool support such as {\symkat}~\cite{symkat}.
For KAT representations of programs, we can define
an analog of concrete program refinement, called
\emph{concrete KAT refinement} (see Def.~\ref{def:concretekrefine}), but
this is not enough to capture programs that have differences.

\paragraph{Contributions.} 
We present a novel theory of \emph{trace-refinement relations}, which weaken both the notion of concrete KAT refinement as well as state refinement relations~\cite{benton}.
Intuitively, the idea is to reason piece-wise and relate
classes of program $C_2$ behaviors (traces) to classes of $C_1$ behaviors. We identify
a class of $C_2$ traces by applying a trace \emph{restriction} and, for that class, we identify
an appropriate separate restriction that can be applied to $C_1$. We also 
provide abstractions over individual events so that we can identify
which atomic events in $C_2$ correspond to which atomic
events in $C_1$ as well as which atomic events are unimportant.
We treat a concrete program $C$ in terms of its traces by abstracting it---via an intermediate abstract program---to a KAT expression, denoted $k_C$.
  A \emph{trace-refinement relation} $\interfaceRel$ then characterizes the relationship between $C_2$ and $C_1$ in terms of their respective trace abstractions $k_{2}$ and $k_{1}$. Specifically, each element of $\interfaceRel$ is a tuple $(r_2,r_1,\hypotheses)$ embodying a trace class relationship:  a restriction $r_2$ for $C_2$, a restriction $r_1$ for $C_1$, and atomic event/condition abstractions as a set of KAT hypotheses $\hypotheses$.
The overall refinement condition is that, for each such
triple, KAT inclusion holds over the programs restricted and abstracted
trace-wise. Technically, we write 
$\abskrefines{C_2}{C_1}{\interfaceRel}{}{}$
to mean:
\[\begin{array}{c}
    \forall (r_2,r_1,\hypotheses) \in \interfaceRel.
    \;\; k_{2} \cap r_2 \leq_{\hypotheses} k_{1} \cap r_1 
     \textrm{ and }     \bigcup_{\proj{1}{\interfaceRel}} \supseteq k_{2}.
\end{array}\]
Here $\leq_\hypotheses$ denotes KAT inclusion (up to $\hypotheses$),
  restriction is achieved through intersection, and
the additional condition requires that $\interfaceRel$ accounts for all behaviors of $C_2$.
The latter requirement is not always enforced, as partial solutions are also useful.
Our treatment based on restrictions, if applied to state-based relations,
would be akin to taking the pre-relation to be
disjunctive and considering each disjunct one at a time.
Furthermore, notice there is no particular requirement on the relationship
between $r_2$ and $r_1$ within a tuple, affording flexibility in how to
relate the corresponding classes.
Finally, we have shown that our trace-refinement relations are 
composable (Thm.~\ref{thm:refinement-all}), permitting an analysis of
one pair of program fragments to be reused in many contexts and when the program is further changed.

The second part of our work introduces a novel algorithm that is able
to synthesize a trace-refinement relation $\interfaceRel$, given input programs
$C_1$ and $C_2$, such that $\absrefines{C_2}{C_1}{\interfaceRel}{}{}$. Overall, our search algorithm constructs $(r_2,r_1,\hypotheses)$
triples, by looking for restrictions that can be placed by $r_2,r_1$, on $C_2,C_1$ and relaxing their behavior
with new hypotheses $\hypotheses$.
Each stage of our algorithm is a candidate tuple of trace classes---that may need further restriction---and we proceed as follows.
First, we iteratively synthesize a two-step semantic abstraction $\alpha$ from a program
$C_1$ (and $C_2$) to a KAT expression $k_1$ (and $k_2$).
This semantic translation goes beyond
the syntactic $C$-to-$k$ translation of Kozen~\cite{kozen}, expressing
restrictions (via \texttt{assume}) and accounting for paths becoming infeasible.
Second, we check whether $k_2$ is included in $k_1$ using KAT reasoning, returning a counterexample string if not.
  Third, we use a custom edit-distance algorithm on such counterexamples to find relationships between cross-program trace classes, with a scoring scheme to correlate program behaviors.
Fourth, we employ case-analysis on branching in $C_1$ and $C_2$, in circumstances where the branching in one program prevents immediate inclusion in the other. This leads us to introduce more restrictions which---unlike~\cite{benton,DBLP:journals/tcs/Yang07,GLP2017,UTS17}---are not required to be of the initial state, but rather may appear
at program locations anywhere in the program.
When we are unable to continue refinement through case-analysis,
we introduce KAT hypotheses that either
treat unimportant events (\eg\ logging) as \textsf{skip}, or else equate
similar events.
Finally, we construct increasingly restricted trace classes used in subsequent iterations of the algorithm.


The goal of our work is to build foundations for compositional
  trace-based reasoning. To this end, we generate trace-refinement relations $\interfaceRel$ to capture the multitude of conditions and ways in which
  one code region $C_1$ can refine another region $C_2$, in order to produce results
  that are reusable.
  Additionally, synthesized trace refinement relations can also be used
  by experts. Like a syntactic diffing tool, the relations can guide them to
  understand how code has changed.


We have implemented our algorithm in a new tool called {\knotical}, that operates on
an input pair of C-like programs and
synthesizes trace refinement relations. {\knotical} is 
built from the ground up, in {\ocaml} using
{\interproc}~\cite{interproctool} for abstract interpretation,
{\symkat} to generate KAT counter examples~\cite{symkattool, symkat},
and our own edit distance algorithm.

We have evaluated our tool on a collection of \numbenchmarks\ benchmark
examples that we have built 
(Apx.~\ref{apx:fullresults}). Almost all examples
necessitate trace refinement relations that cannot be expressed using concrete
refinement or other prior techniques. The examples range from those
designed to exercise the various aspects of our approach (restriction,
hypotheses, edit distance, etc.), to broader examples including user
I/O, array access patterns, and reactive web servers
(\eg\ thttpd~\cite{thttpd} and merecat~\cite{merecat}).
Our evaluation
demonstrates that interesting and precise trace refinement relations
can be discovered.

\paragraph{Summary.} We make the following contributions:
\begin{itemize*}
\item A theory of \emph{trace}-refinement relations, going beyond prior \emph{state} refinement relations. (Sec.~\ref{sec:refinement})
\item A proof of composition. (Thm.~\ref{thm:refinement-all})
\item A novel synthesis algorithm that iteratively constructs trace-refinement relations. (Sec.~\ref{sec:algorithm})
  \item A proof of soundness. (Thm.~\ref{thm:sound})
\item The first tool for trace refinement relations. 
  (Sec.~\ref{sec:algorithm}-\ref{sec:eval})
\item A customized edit-distance algorithm for scoring and
  finding alignments between programs. (Sec.~\ref{sec:edit-distance})
\item A collection of benchmarks and experimental validation, demonstrating viability. 
  (Sec.~\ref{sec:eval})
\end{itemize*}

\paragraph{Related work.} To our knowledge, we are the
first to generalize Benton-style refinement~\cite{benton,DBLP:journals/tcs/Yang07} to
trace relationships.  Bouajjani \emph{et al.~}\cite{BEL2017}
have focused on concurrent loop-free programs. Their notion of
  refinement is not quite based on ``traces'' in the sense that we describe
  herein, but rather on graphs over the reads-from relation and
  program order.
More distantly related are bisimulations, hyper temporal
logics~\cite{ClarksonFKMRS14}
and 
self-composition~\cite{barthe2004secure,terauchi2005secure}.
(See Sec.~\ref{sec:relwork})

\paragraph{Limitations.}
We developed a theory for trace refinement
relations and, while KAT has worked well, it has also meant that we
were restricted to terminating programs. We leave possibly
non-terminating programs to future work. Our implementation was
also limited in the number of symbols due to 
{\symkat}~\cite{symkattool}'s use of \texttt{char} to represent symbols.

\vspace{-5pt}
\newcommand\evrecv{$\event{recv}$}
\newcommand\evlog{$\event{log}$}
\newcommand\evsend{$\event{send}$}
\newcommand\evcheck{$\event{check}$}
\newcommand\evwarn{$\event{warn}$}
\newcommand\evRep{$\event{Rep}$}

\begin{figure}
  \begin{tabular}{|l|l|}
    \hline
        {\bf Program $C_1$ } & {\bf Program } $C_2$ \\
        \hline
    \begin{minipage}{1.5in}
		\footnotesize
		\begin{program}[style=tt,number=true]
		  wh\tab ile(x $>$ 0) \{
		     m = \evrecv();
		     if (l) \evlog(m); \label{ln:c1log}
		     if\tab (m $>$ 0)  \{ \label{ln:c1m}
		        n = construct\evRep ly();
		        \evsend(n);
		        if (l) \evlog(n); \untab
		      \}
		      x--; \untab
		   \} 
		\end{program}
	\end{minipage}
	&
	\begin{minipage}{1.6in}
		\footnotesize
		\begin{program}[style=tt,number=true]
		  wh\tab ile(x $>$ 0) \{
		     m = \evrecv();
		     if\tab\  (m $>$ 0) \{
		        auth = \evcheck(m);  \label{ln:auth}
		        if\tab (auth $>$ 0) \{
		          n = construct\evRep ly();
		          \evsend(n); \untab
		        \} \untab
		     \} else \{ \evlog(m); \}
		     x--;  \untab
		  \}
		\end{program}
	\end{minipage}\\
	\hline
  \end{tabular}
  \caption{\label{fig:example} (Left) A simple reactive program $C_1$ that receives messages and sends replies. (Right) A modified version of $C_2$ with changes including the addition of authentication.}
\end{figure}


\newcommand\kkttK{\textsf{K}_{\event{check}}} 
\newcommand\kkttS{\textsf{S}_{\event{send}}} 
\newcommand\kkttO{\textsf{O}_\event{log}} 
\newcommand\kkttD{\textsf{D}_\event{nondet()}}
\newcommand\kkttR{\textsf{R}_\event{[-inf; inf]}}
\newcommand\kkttX{\textsf{X}_\event{x--}}
\newcommand\kkttG{\textsf{G}_\event{nondet()}}
\newcommand\kkttB{\textsf{B}_\event{C1(x1, c1)}}
\newcommand\kkttL{\textsf{L}_\event{log}} 
\newcommand\kkttN{\textsf{N}_\event{nondet()}}
\newcommand\kkttC{\textsf{C}_\event{Rep}}
\newcommand\kkttE{\textsf{E}_\event{recv}}
\newcommand\kkttT{\textsf{T}_\event{nondet()}}
\newcommand\kkttAny{\textsf{Any}}
\newcommand\kkttc{\textsf{c}_{\texttt{m}\!>\!0}} 
\newcommand\kkttd{\textsf{d}_{\texttt{auth}\!>\!0}}
\newcommand\kkttb{\textsf{b}_{\texttt{l}=\TRUE}} 
\newcommand\kktta{\textsf{a}_{\texttt{x}\!>\!0}}

\newcommand\intermatrix[7]{%
  \left[\begin{array}{cc} #1 \\ #4 \end{array}\right]%
  \begin{array}{cc} #2 \\ #5 \end{array}%
  \left[\begin{array}{cc} #3 \\ #6 \end{array}\right],%
  #7}

\section{Overview}\label{sec:overview}

Consider the two programs in Fig.~\ref{fig:example},
inspired by the Merecat project~\cite{merecat} which
enhanced the thttpd web server~\cite{thttpd} to support SSL connections.
We are interested in knowing how the the new program compares to the previous,
both of which involve typical web server behavior: alternately receives a request and sends a response.
This is illustrated by the two program fragments $C_1$
and $C_2$. The programs involve some differences, perhaps arising
from changes/edits that were made to $C_1$. 
There are still similarities: both programs involve a loop that iterates over
\texttt{x}, \evrecv{}ing messages and possibly \evsend{}ing responses.
On the other hand, $C_2$ only performs a \evlog{} when it \evrecv{}s an
\texttt{m} such that $\texttt{m}\leq 0$, and it additionally performs an \texttt{auth}orization \evcheck{}
on \texttt{m}.
In addition, $C_1$ only performs \evlog{}s when the flag \texttt{l} is enabled.



\subsection{Relating the programs' behavior over time}
\label{subsec:overview-partA}

%
We take a reactive view of programs, considering not only the
programs' local stack/heap state, but also the programs' I/O
side-effects. For simplicity in this paper we will work with 
stack variables and events shown as function calls (denoted \evrecv,
\evlog, etc.) but our work generalizes to heap structures.

We would like to express \emph{similarities} in how the
programs behave over time, such as alternation between \evsend{}
and \evrecv{}. We would like a
  theory to also tolerate the \emph{differences} between how the programs behave over time,
such as the \evrecv{}/\evsend{} behavior in $C_1$ versus the
\evrecv{}/\evcheck{}/\evsend{} behavior in $C_2$.
%
Intuitively, the theory we develop will need some way of
expressing \emph{restrictions} that can be placed on one program
(\eg\ \texttt{auth} is always greater than 0 in $C_2$) so that its traces
are included in the other (\eg\ \evlog{}-free traces of $C_1$), as
well as to provide abstractions that relate an event in one program
(\eg\ the \evsend{} event in $C_1$) to an analogous event in
the other (\eg\ \evsend{} in $C_2$).

\newcommand\ktx{\texttt{x}\red{fix}}
\newcommand\ktR{\texttt{R}\red{fix}}
\newcommand\ktl{\texttt{l}\red{fix}}
\newcommand\ktL{\texttt{L}\red{fix}}
\newcommand\ktm{\texttt{m}\red{fix}}
\newcommand\ktC{\texttt{C}\red{fix}}
\newcommand\ktS{\texttt{S}\red{fix}}
\newcommand\ktD{\texttt{D}\red{fix}}
Expressing properties of the way a program behaves over time motivates the need
for a suitable trace-oriented relational logic (as opposed to state relations~\cite{benton,DBLP:journals/tcs/Yang07,GLP2017,UTS17}). As discussed in Sec.~\ref{sec:intro}, we \ignore{In our work, we first considered temporal logics such as LTL and CTL. We}found that Kleene Algebra with Tests (KAT) was a natural choice~\cite{kozen}
for a few reasons, including that
the KAT constructors naturally abstract program entities. 
Our theory involves an abstraction from concrete programs $C_1,C_2$, via abstract programs to a representation as KAT expressions $k_1,k_2$, respectively (Sec.~\ref{subsec:ctokat}). 

The terms of a KAT expression are \emph{event} symbols or boolean
\emph{test} symbols. We introduce (uppercase) event symbols for program
statements such as ``$\kkttE$'' for \evrecv\ and (lowercase) boolean test
symbols for the integer expressions above such as ``$\kktta$'' for \texttt{x>0}.
(For ease of reading, we use subscripts to indicate which program expressions
correspond to the symbol.) Thus, the behaviors of programs $C_1$ and $C_2$ in
Fig.~\ref{fig:example} can be represented, respectively, as:
\newcommand\cdott{\!\cdot\!}
\begin{displaymath}
	\begin{array}{ll}
		k_1 \defeq& (\kktta(\kkttE(\kkttb\cdott\kkttO\;+\;\overline{\kkttb}\cdott1)\cdott(\kkttc\cdott\kkttC\cdott\kkttS \\
              & (\kkttb\cdott\kkttL\;+\;\overline{\kkttb}\cdott1)\;+\;\overline{\kkttc}\cdott1)\kkttX))^*\cdott\overline{\kktta}\\
		k_2 \defeq& (\kktta(\kkttE(\kkttc\cdott\kkttK\cdott(\kkttd\cdott\kkttC\cdott\kkttS\;\\
		&+\;\overline{\kkttd}\cdott1)\;+\; \overline{\kkttc}\cdott\kkttO)\cdott\kkttX)^*\cdott\overline{\kktta}
	\end{array}
\end{displaymath}
where ``1'' is the identity symbol in KAT, akin to \texttt{skip} in programs.
Note that composition $\cdot$ binds tighter than union $+$, and we use
overline (\eg\  $\overline{\kkttc}$) to indicate negation.
The above KAT symbols represent program statements, but we use KAT symbols more
generally as semantic entities.
%



\paragraph{Trace-refinement relations.} With KAT representations of programs, it is straight-forward to
  define \emph{concrete} KAT refinement (Def.~\ref{def:concretekrefine}
  for programs that are equivalent. However, not all behaviors of $C_1$ are
contained within $C_2$ (such as some logging in $C_1$) and vice-versa
(authorization failures in $C_2$) and so KAT refinement does not hold for Fig.~\ref{fig:example}. Nonetheless, we may still be interested in \emph{which} behaviors of $C_1$ are in $C_2$ and how one might correlate events in $C_1$ with $C_2$. We may want to describe how both programs have a substantially similar \evrecv/\evsend\ relationship. Imagine that we could somehow focus on the executions of $C_2$ in which \texttt{auth} was always greater than 0, somehow focus on the executions of $C_1$ that had no \evlog\ events (when \texttt{l} was always $\FALSE$), and finally on the executions of both programs where they \evrecv\ valid messages and thus $\texttt{m}>0$.  In that case, the programs would have the following \emph{more restricted} behaviors, represented as the following restricted KAT expressions:
\begin{align}
(\kktta(\kkttE(\kkttc\cdott\kkttC\cdott\kkttS)\kkttX))^*\overline{\kktta} &\leq k_1 & \label{eqn:sketchC1}\\
(\kktta(\kkttE(\kkttc\cdott\kkttK\cdott\kkttC\cdott\kkttS)\kkttX))^*\overline{\kktta} &\leq k_2 & \label{eqn:sketchC2}
\end{align}
The above equations are just \emph{classes} of trace behaviors of
$C_1$ and $C_2$, respectively, with $\leq$ denoting KAT inclusion.
If we could now further somehow \emph{ignore} the $\kkttK$ event, the above KAT expressions are equivalent. (In this case they are syntactically equivalent, but they could also be semantically equivalent.) Finding this correlation takes care of some behaviors of $C_1$.  Doing this for all behaviors of $C_1$ leads us to our trace-oriented notion of refinement relations.

We formalize this kind of reasoning into a weak (as
opposed to concrete) and compositional notion of refinement based on
what we call \emph{trace-refinement relations}. 
We consider one class of
traces of $C_2$ at a time like we did above in Eqns.~\ref{eqn:sketchC1} and~\ref{eqn:sketchC2}. We
translate programs into KAT expressions and then reason abstractly
about traces of $C_2$ by considering its corresponding KAT expression
$k_2$ and focus on particular behaviors by \emph{restricting}
behaviors---also described as another KAT expression $r_2$---with
intersection: $k_2 \cap r_2$. For this restricted behavior of $k_2$, it
is then often helpful to restrict $k_1$ (which corresponds to $C_1$)
with a perhaps rather unrelated $r_1$. Then we can ask whether
equivalence holds between $k_2\cap r_2$ and $k_1 \cap r_1$. Returning to the
running example, we can consider the class of traces of $C_2$ in which
\texttt{auth} is always greater than 0 by letting
\begin{equation}\label{eqn:lauth}
r_2= (\kktta(\kkttAny\cdott\boxed{\kkttc}\cdott\boxed{\kkttd}))^*\cdott\overline{\kktta}
\end{equation}
where $\kkttAny$ is shorthand for the disjunction of all event symbols in the KAT at hand. This restriction allows behaviors of the program where after any event, both $\texttt{m}>0$ and $\texttt{auth}>0$.
We can use
this restriction to focus on $k_2\cap r_2$.
Similarly we can restrict $C_1$ to the classes of traces that do not
involve logging by letting
\begin{equation}
r_1=(\kktta(\kkttAny\cdott\boxed{\overline{\kkttb}}))^*\cdot\overline{\kktta}, 
\end{equation}
requiring that $\overline{\kkttb}$ holds on every iteration of the loop.
With these restrictions in place, we get Eqns~\ref{eqn:sketchC1} and~\ref{eqn:sketchC2} above.

In some cases, we can witness classes of traces in $C_2$ that are in $C_1$
simply with a pair of restrictions.  However,
restrictions are not the only way that we relate $k_2$ to
$k_1$. Looking at Eqns~\ref{eqn:sketchC1} and~\ref{eqn:sketchC2}, there is still the discrepancy that
the $\kkttK$ event occurs in  $C_2$ but not $C_1$. Since we are already focused on a case where \texttt{auth} is always greater than 0, the $\kkttK$ event is not so important. We can ignore such unimportant events by introducing additional \emph{hypotheses} into the KAT. In this case, we can introduce the hypothesis ``$\kkttK=1$,'' and we finally have the KAT relationship $(k_2\cap r_2) \equiv_{\{\kkttK=1\}} (k_1\cap r_1)$.
Our choice of working with KAT enables us to exploit algebraic reasoning and so we can introduce hypotheses for other purposes too. It is often convenient to let syntactically identical statements between $C_2$ and $C_1$ use the same KAT symbol. In other cases, we may prefer not to, but we can introduce KAT hypotheses to instead selectively relate statements.

\paragraph{Putting it all together.}
As discussed so far, we have only considered one class of $C_2$ traces
and there are of course many others. Ultimately, we will collect a set
$\interfaceRel=\{(r_1,r_2,\hypotheses),\allowbreak(r'_1,r'_2,\hypotheses'),\ldots\}$, each triple considering
different cases. Overall, we use the notation
$\absrefines{k_2}{k_1}{\interfaceRel}{}{}$ to mean that $(k_2\cap r_2)
\leq_{\hypotheses} (k_1\cap r_1)$ holds of each triple and that if we union over
the first projection of $\interfaceRel$, we have taken care of every
possible behavior of $k_2$.
%
Notice that for ``weak'' completeness
we could always add a triple
$(1,1,\hypotheses^\top)$ where $\hypotheses^\top$ maps every single symbol to 1 (\texttt{skip}). The goal is instead to generate \emph{useful and precise} trace-refinement relations. To this end, our algorithm and implementation favor searching for complex restrictions and resort to hypotheses only when necessary.


%
%
To the best of our knowledge, one cannot describe these kinds of trace-based
refinement relations in prior works such as Benton~\cite{benton} and
  bisimulation (at least not without extraordinary effort to represent trace behavior with complicated ghost variables). Nonetheless, there are many intuitive trace-based properties that one would like to express along the lines of \emph{the new program alternately receives and sends messages like the original},  \emph{the new program additionally performs an authorization check after each receive}, and so on.
%


\emph{Remark: Path Sensitivity.} KAT  has facilities to express many well-known tricks for
increasing path sensitivity such as loop unrolling, trace
partitioning~\cite{MR2005}, control-flow refinement~\cite{GJK2009}.
For lack of space, we omit examples such as  alignments  between different loop iterations.   

\subsection{Composition, Contexts, Spanning Versions}

  So far we have discussed reasoning about a change from 
  $C_1$ to $C_2$, while the context remains fixed.
  But what about the context of $C_1$?
  A single fragment $C_1$ can be used in many different contexts within a large program. Thus, when $C_1$ is changed to $C_2$, there is benefit to performing a single analysis that considers all possible contexts, rather than considering how $C_1$ and $C_2$ relate in each context~\cite{continuousreasoning}.
  This approach also allows us to cope with the fact that the context \emph{itself} may change.
  Consider, for example, programs:
  \[\begin{array}{ccc}
      P_1 = \textsf{A}\!\cdot\!\textsf{b}\!\cdot\! C_1 \!\cdot\!
              (\textsf{g}\!\cdot\!\textsf{G}+\textsf{!g}\!\cdot\! 1)&\textrm{and}&
      P_2 =  \textsf{A}\!\cdot\!\textsf{b}\!\cdot\! C_2 \!\cdot\!
              (\textsf{g}\!\cdot\!\textsf{G}+\textsf{!g}\!\cdot\! 1)\\
      \end{array}\]
    In these programs, we need to know how $C_1$ and $C_2$ relate
    in a context of boolean symbol \textsf{b} and how they
    may impact boolean symbol \textsf{g}. If this were the only
    context we were concerned about, we could focus on search for refinement
    relations to those that assume \textsf{b}. However, what if the
    program
    is then changed to
$      P_3 =  \textsf{A}\!\cdot\!\textsf{m}\!\cdot\! C_2 \!\cdot\!
              (\textsf{g}\!\cdot\!\textsf{G}+\textsf{!g}\!\cdot\! 1)$,
    where \textsf{b} is no longer in the context, but \textsf{m}
    is? If we are context agnostic in our analysis in our $P_1$-vs-$P_2$ analysis
    then we can reuse refinement relations for reasoning
    about $P_2$-vs-$P_3$.

    Returning to
    Fig.~\ref{fig:example}, fragments $C_1$
    and $C_2$ may be used in different contexts. Perhaps in one context it is important that
    \emph{all failed connections are logged} and we want to ensure a
    change from $C_1$ to $C_2$ preserves this property. In that case 
    we need a refinement relation that does not ignore
    \evlog\ events, and assume that the context of $C_1$ ensures $\textsf{l}=\textsf{true}$.
    Formally, we would have the tuple
    $$
    ((\kkttb \cdot\overline{\kkttc} \cdot \kkttAny)^*,
    (\overline{\kkttc} \cdot \kkttAny)^*, \hypotheses_{log})
    \in \interfaceRel{}
    $$
    This  restricts to traces of $C_1$ where
    logging is enabled and all connections fail and restricts
    traces of $C_2$ to those where all connections fail. Moreover, we
    require a set of hypotheses $\hypotheses_{log}$ which does not imply that $\kkttO=1$.

    In a different context, other relations would be useful. As noted earlier,
    this example comes from a change that added SSL support to thttpd~\cite{thttpd}.
    Therefore,
    we may wish to have a refinement relation, specifying that \emph{as long as
      all messages are authenticated in $C_2$, then it behaves the same as $C_1$}.
    Our theory allows one to express this relationship written, formally as:
    \[
         (\kkttAny^*, (\kkttd \cdot \kkttAny)^*, \{
         \kkttO=1, \kkttL=1, \kkttK=1\}) \in \interfaceRel{}
    \]
    Here, $C_1$ is unrestricted, $C_2$ is focused only on executions that are authorized,
    and we use a set of hypotheses that ignores all {\evlog} events and ignores the \evcheck\ event in $k_2$.


It is not hard to see that our formalism can capture other more complicated contexts, such
as an outer loop. More broadly, encompasing all of KAT, we have proved that our trace-refinement relations
are \emph{compositional} (Thm.~\ref{thm:refinement-all}), allowing us to reason about the overall trace-refinement, by considering pairs of program segments at a time. 

\subsection{Automation}\label{subsec:overview-automation}




\paragraph{From $C$ to $k$ and back.} Before we present
our main algorithm,
we need to translate back and forth between a program $C$ and its
corresponding KAT expression $k$. The former lets us learn fine-grained details
about the behavior of the program, while the latter lets us perform
coarse-grained cross-program comparisons. To get $k$ from $C$, we
exploit program semantics to obtain precise
KAT expressions, \eg\ excluding infeasible paths.
%
Technically, we work with an (iteratively refined)
two-step abstraction function $\alpha$ that takes
%
concrete states, via abstract states, to symbols in the KAT boolean subalgebra,
using a procedure called $\atranslate : (C,\alpha) \rightarrow \kat$ (see Sec.~\ref{subsec:ctokat}).
%
Our abstraction is not a one-way process:
our refinement search (discussed below) involves discovering various
classes of traces of $k_1$, each expressed as a restriction $k_1\cap r_1$.
For example, we may consider traces of $C_2$ in which \texttt{auth} in
Line~\ref{ln:auth} is always positive. Our algorithm
discovers the restriction $r_2$ in Eqn.~\ref{eqn:lauth},
%
and uses a subprocedure
  $\arestrict(C_1,r_1,C_2,r_2,\hypotheses,\alpha)$
  to instrument these restrictions via
a source code transformation.
We now need to translate this $r_2$ back into the program $C_2$ so
that we can explore more fine-grained behaviors of $C_2$ with help from
abstract interpretation.
This is denoted
$\arestrict(C_1,r_1,C_2,r_2,\hypotheses,\alpha)$
and our implementation represents these restrictions via
a source code transformation, using a form of a product program in which we instrument \texttt{assume} statements.
As our algorithm
continues to search for more fine-grained classes of traces where
refinement holds, it may iteratively instrument more assumptions and
continue to refine the abstraction $\alpha$ (maintaining a
monotonicity constraint).

\paragraph{Overall algorithm:} 

\noindent
\scalebox{0.62}{
  \includegraphics[width=5.5in]{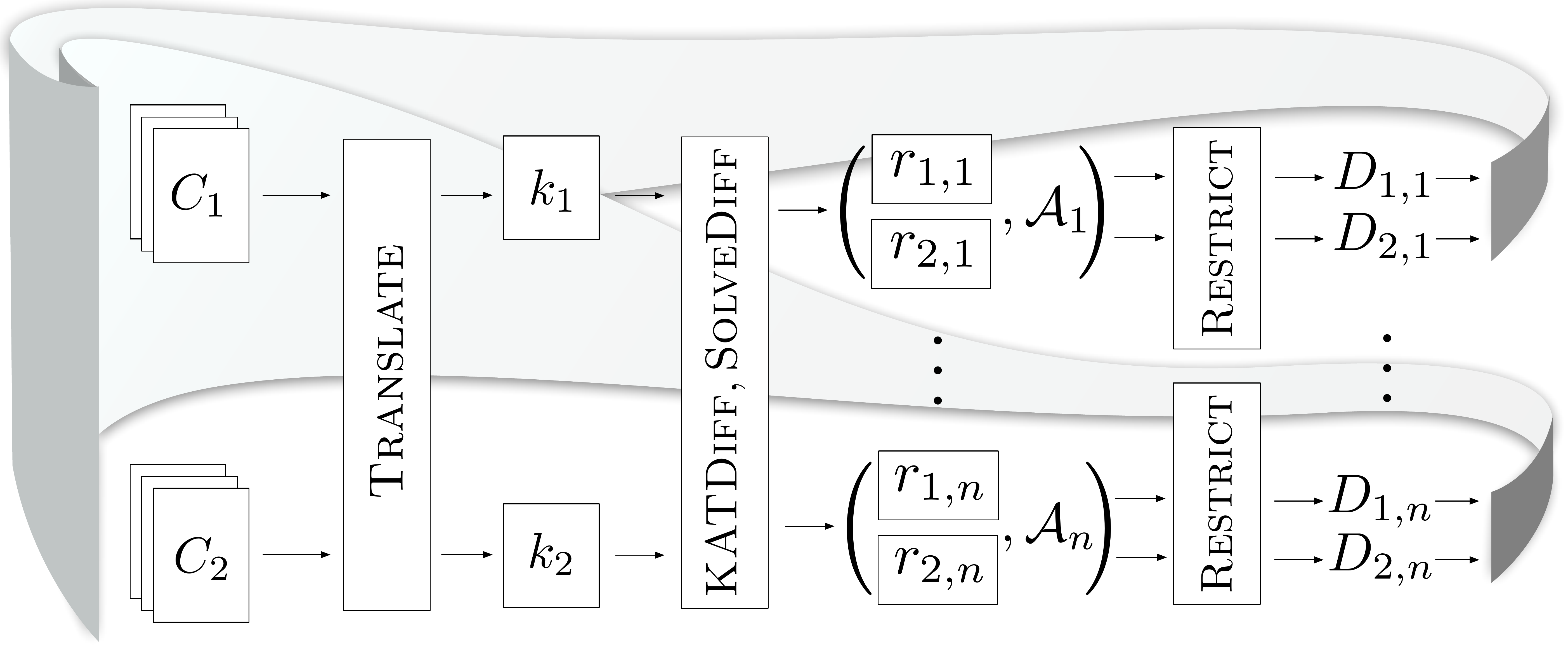}
}
The above is a depiction of our overall algorithm (Sec.~\ref{sec:algorithm}) that synthesizes trace-refinement relations by attempting to discover increasingly granular trace classes of $C_1$ that are included in trace classes of $C_2$.
Each such partial solution is a
triple $(r_1,r_2,\hypotheses)$ such that ${k_1\cap r_1}\leq_{\hypotheses}{k_2\cap r_2}$ (Def.~\ref{def:kat-refinement}), where $k_1$ corresponds to $C_1$ and $k_2$ to $C_2$.



At the high level, each iteration of the algorithm is a recursive call
({\asearch}), where we are exploring a
region of the solution space where $C_1$ has possibly been restricted,
$C_2$ has possibly been restricted and a collection of KAT hypotheses
$\hypotheses$ are in use. Moreover, we have a current abstraction
$\alpha$ from programs to $\kat$.
Each iteration of our algorithm proceeds by calculating
the aforementioned KAT abstractions $k_1$ and $k_2$ of the (possibly already restricted) programs ({\atranslate}).
Next, we consider whether $k_1$ is included in or is equivalent to $k_2$, under the
current set $\hypotheses$ of hypotheses (\akatdiff).
To check this refinement, we build on the recent work of
Pous~\cite{symkat}, using his tool {\symkat}~\cite{symkattool}.  If
this refinement holds, then the algorithm returns this triple
$(k_1,k_2,\hypotheses)$ as a solution that may be
assembled with others into a complete solution by previous calls to
{\asearch}.
Alternatively, if the inclusion/equivalence does not hold, we employ a
sub-procedure ({\asolvediff}) to decide, based on the counterexamples
and the overall KAT expressions, whether to (i) introduce restrictions $(r_{1,i},r_{2,i})$ and/or (ii) introduce hypotheses $\hypotheses_i$. In Secs.~\ref{sec:algorithm} and~\ref{sec:edit-distance} we discuss how the sub-procedure employs a custom edit distance algorithm for this purpose. 
Finally, the restrictions $r_{1,i}$ and $r_{2,i}$ are instrumented back into the programs ({\arestrict}), to produce new programs $D_{1,i}$ and $D_{2,i}$ that are considered recursively.
We have proved that our algorithm for generating trace-refinement relations is sound (Thm.~\ref{thm:sound}). Weak completeness is less interesting because any partial solution can easily be made a complete solution through aggressive use of hypotheses. We note that we do not expect from our algorithm to be able to generate every possible solution. Hence, we are more interested in generating precise and useful trace-refinement relations.

\newcommand\inCCauth[1]{\textsf{asm}(#1)@C_2\!\!:\!\!\ell_{\ref{ln:auth}}}
\newcommand\inClog[1]{\textsf{asm}(#1)@C_1\!\!:\!\!\ell_{\ref{ln:c1log}}}
\newcommand\inCm[1]{\textsf{asm}(#1)@C_1\!\!:\!\!\ell_{\ref{ln:c1m}}}
\newcommand\evlogC{\evlog_{C_1}}
\newcommand\evlogCC{\evlog_{C_2}}

\begin{figure}
  \input{solution}
  \caption{\label{fig:solution} Output of {\knotical}: A trace-refinement relation $\interfaceRel$ such that $\absrefines{C_2}{C_1}{\interfaceRel}{}{}$ for the example in Fig.~\ref{fig:example}.}
\end{figure}

\subsection{The {\knotical} Tool}

We have developed a prototype tool {\knotical} that implements our
algorithms and is the first tool capable of synthesizing trace-refinement relations. Our tool is built from the
ground up written in {\ocaml} and uses {\interproc}\cite{interproctool}
for abstract interpretation and {\symkat}~\cite{symkat} for
symbolically checking KAT expression equalities and inclusion.

Fig.~\ref{fig:solution} illustrates one of the 75 solutions found by {\knotical} when run on the example in Fig.~\ref{fig:example}. The synthesized trace refinement relation $\interfaceRel$ has five tuples. This output illustrates the restrictions using the notation ``$\inCCauth{\kkttd}$'' meaning, for example, that we instrument an \texttt{assume(auth>0)} on line~\ref{ln:auth} of $C_2$. Notice that {\knotical} has considered various case splits, based on these three boolean conditions. It begins with the conditions in the left-hand side ($C_2$) and needs to discover at least one solution in $C_1$ for each case. When \texttt{auth>0}, {\knotical} introduces hypotheses to ignore \evlog\ events in either program and the \evcheck\ event in $C_2$. Otherwise the \evlog\ event does not occur in $C_2$ so it needn't be ignored.

%


\paragraph{Edit-distance for refinement.}
During the algorithm, when considering whether
the current $k_1$ refines $k_2$, {\symkat} may find that it doesn't
and return a counterexample of a string $w_1$ that is in $k_1$ but not
$k_2$ (and $w_2$, vice-versa). Returning to the running example, such
a pair might be
$w_1=\kktta\cdott\kkttE\cdott\overline{\kkttb}\cdott\overline{\kkttc}\cdott\kkttX\cdott\overline{\kktta}$
and
$w_2=\kktta\cdott\kkttE\cdott\overline{\kkttb}\cdott\overline{\kkttc}\cdott\kkttO\cdott\kkttX\cdott\overline{\kktta}$. These
counterexamples give us information as to how $k_1$ and $k_2$
diverge.
%
Our algorithm departs from a traditional counterexample-guided
approach and instead is able to
consider not only the entirety of counterexample strings $w_1$
and $w_2$, but also the  KAT expressions
$k_1$ and $k_2$, in order to find a better correlation between the
two.  
It is easy for a human reader to see that the relationship between
$k_1$ and $k_2$ fits better, when the * expression in $k_1$ is correlated with the * expression in $k_2$.
%
%
To this end,
we developed a custom edit-distance algorithm~\cite{treeedit} (see Sec.~\ref{sec:edit-distance}). 

\paragraph{Evaluation.}
We created a series of \numbenchmarks\ benchmarks for most of
which, trace-based refinement relations cannot be expressed
in prior formalisms (Sec.~\ref{sec:eval}). 
On most benchmarks, our tool was able to generate a non-trivial trace-refinement relation in seconds or fractions of a second.


\vspace{-5pt}
\vspace{-7pt}
\section{Preliminaries}\label{sec:preliminaries}

\paragraph{Strings, Sets, Composition, Programs}
A string $s$ over an alphabet $\Sigma$ is a sequence $s_1\cdot s_2\cdots s_n$ of symbols $s_i\in\Sigma$, for $i\in[1,n]$.
Given sets $S_1,\ldots,S_n$, a set $S\subseteq S_1\times S_2\times\ldots S_n$, and an element $s=(s_1,\ldots,s_n)\in S$ we denote with $\proj{i}{s}$ the projection of $s$ to its $i$-th element $s_i$ in $S_i$. We abuse notation, denoting as $\proj{i}{S}$ the set $\{s_i\in S_i\mid s_i\in\proj{i}{s}, s\in S\}$.
%

We assume a set $\progs$ of (essentially imperative) \emph{programs} operating on a set $\states$ of \emph{states}.  We assume a distinguished ``error state'' $\fault \in \states$. A \emph{configuration} is a pair $\config{C}{\sigma}$, where $C$ is a program and $\sigma$ a state; we write $\configs$ for the set of all configurations. We assume a binary relation $\mathord{\leadsto} \subseteq \configs \times \states$ capturing the ``big step'', nondeterministic operational semantics of our programs; $\config{C}{\sigma} \leadsto \rho$ means that executing program $C$ in initial state $\sigma$ can result in the final state $\rho$.

\paragraph{Kleene Algebra with Tests.}
We use KAT~\cite{kozen} to represent classes of traces
within a program. A \emph{Kleene Algebra with Tests} $\kat$ is a two-sorted
structure $(\Sigma,\bkat,+,\cdot,\kstar{},\kneg{\;},0,1)$, where
$(\Sigma,+,\cdot,\kstar{},0,1)$ is a Kleene algebra,
$(\bkat,+,\cdot,\kneg{\;},0,1)$ is a Boolean algebra, and $(\bkat,+,\cdot,0,1)$
is a sub-algebra of $(\Sigma,+,\cdot,0,1)$. We distinguish between two sets of
symbols: set $\actions$ for primitive actions, and set $\tests$ for primitive tests. The grammar of boolean test expressions is:
$\boolgr\grdef b\in\tests\mid b_1\cdot b_2\mid b_1+b_2\mid \kneg{b}\mid 0\mid 1$
and we define the grammar $\katgr$ of KAT expressions as:
$$\katgr\grdef p\in\actions\mid b\in\boolgr\mid k_1\cdot k_2\mid k_1+k_2\mid \kstar{k}\mid 0\mid 1$$
The free Kleene algebra with tests over $\actions\cup\tests$, is obtained by
quotienting $\boolgr$ with the axioms of Boolean algebras, and $\katgr$ with the
axioms of Kleene Algebra. For $e,f\in \kat$, we write $e\leq f$ if $e+f= f$, and
all Kleene Algebras with Tests $\kat$ we consider here are $*$-continuous, where
any elements $a,b,c$ in $\kat$, satisfy the axiom $a\cdot b^*\cdot
c=\sum_{n\in\mathbb{N}}a\cdot b^n\cdot c$~(\cite{Kozen90}). By convention we use
lower case letters for test symbols and upper case letters for actions. We may
also abuse notation, writing program conditions and statements rather than
boolean symbols and action symbols (in which case we implicity create symbols
for each). For Fig.~\ref{fig:example} booleans include
$\tests = \{\kktta, \kkttb\}$,
actions include 
$\actions=\{ \kkttO, \kkttE \}$,
and
$
k=...(\kkttb\!\cdot\!\kkttO + \overline{\kkttb}\!\cdot\!1)... \in \kat
$.
\begin{definition}[Intersection]
	Given a KAT $\kat$ and two of its elements $k_1$ and $k_2$ we define $k_1\cap k_2$ to be equal to $l_1+\ldots +l_n+\ldots$, where $\{l_i\}_{i\in\mathbb{N}}$ is the set of all elements $l_i$ in $\kat$ such that $l_i\leq k_1$ and $l_i\leq k_2$.\footnote{Notice that for any two KAT expressions $k_1+\ldots+k_n$ and $l_1+\ldots+l_m$ over a KAT $\kat$, for $n,m\in\mathbb{N}$, there is a finite number of elements $h_1+\ldots+h_r$, for $r\in\mathbb{N}$ such that $(k_1+\ldots+k_n)\cap(l_1+\ldots+l_m)\equiv h_1+\ldots+h_r$.
Since we never start with a KAT expression as an infinite disjunction in what follows, any time we talk about the intersection of two KAT expressions as a disjunction of KAT elements, we will refer to such a finite disjunction.
}
\end{definition}


For KAT expressions $k_1,k_2$ and $l$, and a set of hypotheses $\hypotheses$, we write $l\in \katminus{k_1}{k_2}{\hypotheses}$ if $l\leq_{\hypotheses} k_1$ and $l\not\leq_{\hypotheses} k_2$. Similarly, we write $l\in\katdiff{k_1}{k_2}{\hypotheses}$ if $l\in \katminus{k_1}{k_2}{\hypotheses}$ or $l\in \katminus{k_2}{k_1}{\hypotheses}$. Finally, for two KATs $\kat_1$ and $\kat_2$, we denote with $\kat_1\cup\kat_2$ the smallest KAT that contains both $\kat_1$ and $\kat_2$. Finally, when we refer to strings we mean KAT strings, which are KAT expressions where only the concatenation operation is used.


\paragraph{Program Refinement.}
Program refinement is a classical concept~\cite{morgan1994} and can be formulated in different ways, depending on the context. 
Often, the usual notion of refinement is too concrete because it does not consider the context in which $C_1$ and $C_2$ are used. Benton~\cite{benton} introduced a weaker notion of refinement, parameterized by an input relation between the states of the two programs as well as an output relation. We call this an \emph{interface}, which is an equivalence relation on the set of states $\states$ and defined as follows:
\begin{definition}
\label{def:absrefines}
For interfaces $I,O$ and programs $C,C'$,  we say \emph{$C'$ refines $C$ w.r.t.\ $(I,O)$}, written $\absrefines{C'}{C}{I}{O}$, if the following two conditions are met, for all states $\sigma,\sigma'$ such that $I(\sigma,\sigma')$:
\begin{enumerate}
\item if $\config{C'}{\sigma'}\leadsto \fault$, then $\config{C}{\sigma}\leadsto \fault$;
\item if $\config{C'}{\sigma'} \leadsto \rho'$, then either there exists $\rho$ such that $\config{C}{\sigma} \leadsto \rho$ and $O(\rho,\rho')$, or else $\config{C}{\sigma}\leadsto \fault$.
\end{enumerate}
We say that \emph{$C'$ (concretely) refines $C$}, written $\refines{C'}{C}$, when $\absrefines{C'}{C}{\textsf{id}}{\textsf{id}}$ where $\textsf{id}$ is the \emph{identity relation}.\footnote{Benton used the notation $\vdash C \stackrel{}{\thicksim} C' : I \Rightarrow O$ whereas we use notation by James Brotherston (personal communication)}
\end{definition}
\noindent
Yang~\cite{DBLP:journals/tcs/Yang07} extended Benton's work to express relational heap properties using a variant of Separation Logic~\cite{seplogic}.


\vspace{-5pt}
\section{KAT Representations and Refinements}
\label{sec:refinement}
 In this section we discuss a two-step semantic abstraction (Sec.~\ref{subsec:ctokat}), trace refinement and trace-refinement relations (Sec.~\ref{subsec:kat-refinements}), and composition results (Sec.~\ref{subsec:composition}).


\subsection{Abstracting programs into KAT expressions}\label{subsec:translation}
\label{subsec:ctokat}

We describe how to abstract a while-style program $C$ to
  a KAT expression $k$ over a KAT $\kat$.
We parameterize such a translation by an abstraction $\alpha$ used for both abstracting concrete states of the program to abstract states, as well as the latter to elements of the boolean subalgebra of $\kat$. More concretely, given a program $C$ over a set of states $\states$, we define $\alpha$ to be a tuple $(\kat, A_S, \alpha_S,\alpha_B)$, where $\kat$ is a KAT, $A_S$ is a set of abstract states, $\alpha_S$ is a mapping from $\states$ to $A_S$ corresponding to the program abstraction given by the abstract interpretation, and $\alpha_B$ is a mapping from $A_S$ to $\bkat$, the boolean subalgebra of $\kat$. Additionally, we require that for any $b\in\bkat$, there is a set of states $\{a_1,\ldots,a_n\}\in A_S$ such that $b\equiv\alpha_B(a_1)+\ldots+\alpha_B(a_n)$. When $\kat$ and $A_S$ are clear from the context, we write $\alpha=\alpha_B\circ\alpha_S$.

With such an abstraction $\alpha=(\kat, A_S, \alpha_S, \alpha_B)$ as a parameter, we say a translation from $C$ to a KAT expression $k\in\kat$ is \emph{valid} (resp. \emph{strongly valid}), if for any states $\sigma,\rho\in\states$, $\config{C}{\sigma}\leadsto\rho$ only if (resp. if and only if) $\alpha_B(\alpha_S(\sigma))\cdot k\cdot \alpha_B(\alpha_S(\rho))\not\equiv 0$. We assume a procedure $\translate{C}{\alpha}$ that returns $k\in\kat$ and the translation from $C$ to $k$ is valid (Sec.~\ref{sec:algorithm} for an implementation). Finally, we will (Sec.~\ref{sec:algorithm}) iteratively construct abstractions and thus need the following notion of refinement over abstractions:

\begin{definition}[Refining abstractions]
	For two abstractions $\alpha=(\kat,A_S,\alpha_S,\alpha_B)$ and $\alpha'=(\kat',A_S',\alpha_S',\alpha_B')$ over the same set of concrete states $\states$, we say that $\alpha'$ \emph{refines} $\alpha$, and write it as $\alpha'\abstractionRefinement\alpha$, if $\kat$ is a subalgebra of $\kat'$ and for any state $\sigma\in\states$, $\alpha_B'(\alpha_S'(\sigma))\leq \alpha_B(\alpha_S(\sigma))$.
\end{definition}


\noindent
Let $\alpha_1=(\kat_1,A^1_S,\alpha^1_S,\alpha^1_B)$ and $\alpha_2=(\kat_2,A^2_S,\alpha^2_S,\alpha^2_B)$ be two abstractions, both refining an abstraction $\alpha$ with Boolean algebra $\bkat$.
By $\alpha^1_S\times\alpha^2_S$ we denote the function from $\states$ to $A^1_S\times A^2_S$, that maps a state $\sigma\in\states$ to $(\alpha^1_S(\sigma),\alpha^2_S(\sigma))$. Further, we define $\alpha^1_B\cdot\alpha^2_B$ to be the function from $A^1_S\times A^2_S$ to $\bkat$ that maps a tuple $(a_1,a_2)\in A^1_S\times A^2_S$ to $\alpha^1_S(a_1)\cdot \alpha^2_S(a_2)$ in $\bkat$. The \emph{combined abstraction} of $\alpha_1$ and $\alpha_2$, written $\alpha_1\abstractionCombinedRefinement\alpha_2$, is defined to be the abstraction $(\kat_1\cup\kat_2,A^1_S\times A^2_S,\alpha^1_S\times\alpha^2_S,\alpha^1_B\cdot\alpha^2_B)$.

\subsection{KAT refinements}
\label{subsec:kat-refinements}

With abstractions from programs to KAT expressions in hand, we now first define \emph{concrete} KAT refinement, and then our notion of trace-refinement \emph{relations} (Def.~\ref{def:kat-refinement}). 

\begin{definition}[Concrete KAT refinement]\label{def:concretekrefine}
	Let $k_1$ and $k_2$ be two KAT expressions over $\kat$. We say that $k_1$ \emph{concretely refines} $k_2$, and denote it by $\refines{k_1}{k_2}$, if for any $b,d\in\bkat$:
	\begin{enumerate}
		\item $b\cdot k_1\equiv 0$ implies $b\cdot k_2\equiv 0$,
		\item $b\cdot k_1\cdot d\not\equiv 0$ implies $b\cdot k_2\cdot d\not\equiv 0$, or $b\cdot k_2\equiv 0$.
	\end{enumerate}
\end{definition}



The following  relates concrete trace refinement, via abstraction, back to concrete program refinement \omittedproofref{\ref{apx:omitted}}.

\begin{theorem}\label{thm:concrete-and-kat-concrete-refinement}
	Let $C_1$ and $C_2$ be two programs, and let $k_1$ and $k_2$ be the two KAT expressions obtained from a strongly valid translation of the two programs respectively, under some abstraction $\alpha$. Then it holds that $\refines{C_1}{C_2}$ if and only if $\refines{k_1}{k_2}$.
\end{theorem}

\noindent
We now weaken concrete KAT refinement, presenting trace-refinement relations.
Intuitively, the idea is to reason piece-wise, considering classes of traces within $k_1$ and, for each, correlating them with a corresponding trace class in $k_2$, with the help of KAT hypotheses.
Note that, for some element $k$ of a KAT $\kat$, we say a set $S=\{s_1,\ldots,s_n\}$ of $\kat$ elements \emph{partitions} $k$, if $k=s_1+\ldots +s_n$.

\begin{definition}[Trace Refinement Relations]\label{def:kat-refinement}
	Let $\kat$ be a KAT, let $\hypothesesclass$ be a class of hypotheses over $\kat$, and let $\interfaceRel$ be a relation over $\kat\times\kat\times\powset{\hypothesesclass}$.
	Given two KAT elements $k_1$ and $k_2$ of $\kat$, we say that $k_1$ \emph{refines} $k_2$, with respect to $\interfaceRel$, denoted by $\abskrefines{k_1}{k_2}{\interfaceRel}{}{}$, if $\proj{1}{\interfaceRel} \text{ partitions }k_1$ and, $$\text{for any $(l_1,l_2,\hypotheses)\in \interfaceRel$},\;\;\;\;\;l_1\cap k_1\leq_{\hypotheses} l_2\cap k_2.$$
\end{definition}

We also consider \emph{trace equivalence relations}, slightly adapting Def.~\ref{def:kat-refinement} to use equivalence ($\equiv$), rather than inclusion ($\leq$), as well as requiring that both $\proj{1}{\interfaceRel}$ partitions $k_1$ and $\proj{2}{\interfaceRel}$ partitions $k_2$.

As discussed in Sec.~\ref{sec:overview}, intuitively each $(l_1,l_2,\hypotheses)$ triple in a trace-refinement relation $\interfaceRel$ identifies restrictions on $k_1$ and $k_2$, as well as KAT hypotheses $\hypotheses$ that allow us to align the $k_1\cap l_1$ trace classes with ones in $k_2\cap l_2$. In the example from Sec.~\ref{subsec:overview-partA}, we gave examples of an $l_1$ that excluded logging by forcing $\overline{\kkttb}$ to hold at each iteration of the loop.

\begin{remark}\label{rem:trivial-solution}
	As trace-refinement is a weakening of concrete refinement, it is natural that
  two KAT expressions $k_1$ and $k_2$ may be such that $k_1$ refines
  $k_2$, but does not concretely refine it. For any two expressions $k_1$ and $k_2$, the singleton set containing only the tuple $(k_1,k_2,\hypotheses)$, where $\hypotheses$ is a set of hypotheses that equates all actions to $1$ and all boolean variables to $0$\todo{check this} is a trivial solution to trace-refinement between $k_1$ and $k_2$.
\end{remark}

Finally, we overload the KAT refinement definition to be used on programs themselves, when the abstraction $\alpha$ is clear from the context. Thus, for two programs $C_1$ and $C_2$, and a trace-refinement relation $\interfaceRel$, we may write $\abskrefines{C_1}{C_2}{\interfaceRel}{}{}$ to mean that $\abskrefines{\atranslate(C_1,\alpha)}{\atranslate(C_2,\alpha)}{\interfaceRel}{}{}$.



\paragraph{Classes of hypotheses.} For this work, we will explore the effect of just a few types of classes of hypotheses. In general, checking equality of KAT expressions under arbitrary additional hypotheses, can be undecidable~(\cite{Kozen96}). Because of that, and guided by the limitations imposed by certain libraries we use in our implementation ({\symkat}), we focus on the following types of hypotheses when $\kM{A}{},\kM{B}{}\in\actions$ and $\kM{a}{},\kM{b}{}\in\tests$: (i) to ignore certain actions: $\kM{A}{}\equiv 1$, (ii) to fix the valuation of certain booleans: $\kM{b}{}\equiv 1$ or $\kM{b}{}\equiv 0$, (iii) to express commutativity of actions against tests: $\kM{A}{}\cdot \kM{b}{}\equiv \kM{b}{}\cdot \kM{A}{}$ (currently not used in our implementation) and (iv) to relate single elements: $\kM{A}{}=\kM{B}{}$ or $\kM{a}{}=\kM{b}{}$.


\subsection{Composition}
\label{subsec:composition}

Given trace-refinement relations $\interfaceRel_1$ and $\interfaceRel_2$, we define their \emph{composition} $\interfaceRel_1\interfaceOpComp\interfaceRel_2$ to be the trace-refinement relation $\interfaceRel=\{(l_1\cdot m_1,l_2\cdot m_2,\hypotheses_1\cup \hypotheses_2)\mid (l_1,l_2,\hypotheses_1)\in\interfaceRel_1, (m_1,m_2,\hypotheses_2)\in\interfaceRel_2\}$.
Similarly, we define \emph{disjunction} $\interfaceRel_1\interfaceOpDisj\interfaceRel_2$ to be the trace-refinement relation $\interfaceRel=\{(l_1+m_1,l_2+m_2,\hypotheses_1\cup \hypotheses_2)\mid (l_1,l_2,\hypotheses_1)\in\interfaceRel_1, (m_1,m_2,\hypotheses_2)\in\interfaceRel_2\}$. Finally, for $\interfaceRel$, we define $\interfaceOpStar{\interfaceRel}$ to be $\{(o^*,p^*,\hypotheses)\mid (o,p,\hypotheses)\in\interfaceRel\}$.

Thm.~\ref{thm:refinement-all} below allows us to reason about individual
fragments of KAT expressions, and combine the analyses into a result that holds
overall. We can do so by building trace-refinement relations in a bottom-up
fashion, capturing larger and larger fragments of those KAT expressions, guided
by their structure. \omittedproofref{\ref{apx:omitted}}

\begin{theorem}\label{thm:refinement-all}
	Suppose $k_1,k_2,l_1$ and $l_2$ are KAT expressions. Let $\interfaceRel_k$ and $\interfaceRel_l$ be trace-refinement relations, such that $\abskrefines{k_1}{k_2}{\interfaceRel_k}$ and $\abskrefines{l_1}{l_2}{\interfaceRel_l}$. Then $\abskrefines{k_1\cdot l_1}{k_2\cdot l_2}{\interfaceRel_k\interfaceOpComp \interfaceRel_l}$, $\abskrefines{k_1 + l_1}{k_2+ l_2}{\interfaceRel_k\interfaceOpDisj \interfaceRel_l}$, $\abskrefines{k_1 + l_1}{k_2+ l_2}{\interfaceRel_k\cup \interfaceRel_l}$, and $\abskrefines{k_1^*}{k_2^*}{\interfaceOpStar{\interfaceRel_k}}$.
\end{theorem}

As a simple corollary we can always extend a trace-refinement relation corresponding to a pair of KAT expressions, to one corresponding to a pair of KAT expressions obtained from the former by enclosing them into any common context.

\begin{corollary}\label{cor:same-context}
	Given any KAT expressions $m,l,k_1$ and $k_2$, and trace-refinement relation $\interfaceRel$ such that $\abskrefines{k_1}{k_2}{\interfaceRel}$, it holds that $\abskrefines{m\cdot k_1\cdot l}{m\cdot k_2\cdot l}{\interfaceRel'}$, where $\interfaceRel'$ is the set $\{(m\cdot r_1\cdot l,m\cdot r_2\cdot l,\hypotheses)\mid (r_1,r_2,\hypotheses)\in\interfaceRel\}$.
\end{corollary}




Finally, we present a transitivity result, stating how we can extend two trace-refinement relations to achieve it.
Let $\interfaceRel_1$ and $\interfaceRel_2$ be two trace-refinement relations, such that for any tuple $(o_1,p_1,\hypotheses_1)$ in $\interfaceRel_1$, there is a tuple $(o_2,p_2,\hypotheses_2)$ in $\interfaceRel_2$, such that $p_1\leq o_2$. For such trace-refinement relations, we define their \emph{transitive trace-refinement relation} to be the one containing the tuples $(o_1,p_2,\hypotheses_1\cup \hypotheses_2)$. We denote such a trace-refinement relation by $\interfaceRel_1\interfaceOpTrans\interfaceRel_2$.

\begin{theorem}\label{thm:transitivity}
	For any elements $k,l$ and $m$ in a KAT $\kat$, and any trace-refinement relations $\interfaceRel_1$, $\interfaceRel_2$, if $\abskrefines{k}{l}{\interfaceRel_1}$ and $\abskrefines{l}{m}{\interfaceRel_2}$, and $\interfaceRel_1\interfaceOpTrans\interfaceRel_2$ is defined, then $\abskrefines{k}{m}{\interfaceRel_1\interfaceOpTrans\interfaceRel_2}$.
\end{theorem}

\vspace{-5pt}
\section{Automation}
\label{sec:algorithm}

\newcommand\BD{{\mathcal B}}
\newcommand\Inv{{\mathcal I}}

Our overall algorithm is given in Fig.~\ref{fig:algorithm}. The input to our
algorithm are programs $C_1,C_2$ provided, for example, in a C-like source
format and parsed into ASTs.  Our algorithm returns trace-refinement
relations for $C_1,C_2$, and is parametric as to whether the relations
are for equivalence versus inclusion. Technically, it returns a finite
set $O = \{(l_1^1,l_2^1,\hypotheses^1,\alpha^1),\ldots\}$ from which
the trace refinement relation $\interfaceFunc(O)$ can be constructed
by unifying to a common abstraction $\alpha=\alpha^1\abstractionCombinedRefinement\ldots\abstractionCombinedRefinement\alpha^n$.

\begin{figure}
\scalebox{0.9}{
  \fbox{
\begin{program}[style=sf]
{\bf Input:} Two programs $C_1,C_2$ and an abstraction $\alpha$.\\
{\bf Output:} \tab A set $O=\{(l^1_1,l^1_2,\hypotheses^1,\alpha^1),\ldots\}$ such that\\
$\abskrefines{\translate{C_1}{\alpha'}}{\translate{C_2}{\alpha'}}{\interfaceFunc(O)}$\\
where $\alpha'$ is the common abstraction of $O$.\untab \\
{\bf Al}\tab{\bf gorithm:} {\asearch}($C_1$,$C_2$,$\hypotheses,\alpha$) // Initially let $\hypotheses=\emptyset$ \\
  $k_1$ := {\atranslate}($C_1,\alpha$)\\
  $k_2$ := {\atranslate}($C_2,\alpha$)\\
  {\acexs} = {\akatdiff}($k_1$,$k_2$,$\hypotheses$)\\
  if no {\acexs} return \{($k_1$,$k_2$,$\hypotheses, \alpha$)\}\\
  els\tab e let R = {\asolvediff}($k_1$,$k_2$,$\hypotheses$,{\acexs}) in\\
     fla\tab tmap ($\lambda$ ( $r_1 \;\; r_2 \;\; \hypotheses'$ ). \\
       let $(D_1,D_2,\alpha') = \arestrict(C_1,r_1,C_2,r_2,\hypotheses\cup\hypotheses',\alpha)$ in\\
	 \asearch($D_1$, $D_2$, $\hypotheses\cup\hypotheses', \alpha'$) ) R \untab \untab \untab\\
\end{program}
}
}
   \caption{The skeleton of {\asearch}, which synthesizes trace-refinement relations for input programs $C_1,C_2$.}\label{fig:algorithm}
\end{figure}

Our main function {\asearch} uses several sub-components discussed below. At the
high level, it begins by using {\atranslate}, analyzing $C_1$ and using an iteratively constructed abstraction $\alpha$ to obtain the KAT expression $k_1$ (similar for $C_2,k_2$), per Sec.~\ref{subsec:ctokat}. The algorithm then checks for KAT equivalence or inclusion between $k_1$ and $k_2$ with {\akatdiff}. If no counterexamples are found, {\akatdiff} returns $k_1$ and $k_2$, together with the current set of hypotheses $\hypotheses$ as a solution. On the other hand, if {\akatdiff} does find counterexamples, 
they are fed into {\asolvediff}, which examines them along with the KAT expressions to determine what restrictions and/or hypotheses could be employed to subdivide the search space into trace classes for which we hope further refinements can be discovered.
{\asolvediff} returns this decision, given as a list of $(r_1,r_2,\hypotheses)$ triples. Then {\arestrict} is used to construct increasingly restricted versions of the input programs $C_1$ and $C_2$ and new abstractions $\alpha'$. These are then are considered recursively by {\asearch}.

\subsection{Sub-procedures}
\label{subsec:subprocedures}

We now define and discuss the sub-procedures used by {\asearch}.
We also discuss the implementation (and limitations) of these subcomponents.
As noted, the overall algorithm is parameterized by whether we are looking for solutions to equivalence ($\equiv_{\hypotheses}$), or simply to inclusion ($\leq_{\hypotheses}$). The functionality of the sub-procedures is largely the same for the two cases.

\noindent$\bullet$\;{\bf ${\atranslate}(C,\alpha)$:} As described in
Sec.~\ref{subsec:translation}, this sub-procedure takes as input a program
$C$ and abstraction $\alpha=(\kat,A_S,\alpha_S,\alpha_B)$, and returns a KAT
expression $k$ in $\kat$ such that $\config{C}{\sigma}\leadsto\rho$ implies that
$\alpha_B(\alpha_S(\sigma))\cdot k\cdot\alpha_B(\alpha_S(\rho))\not\equiv 0$.
In the \emph{implementation}, {\interproc} populates the locations of
  the program with invariants according to abstraction $\alpha$. This
  is then used in converting the abstract program into a KAT expression $k$ in $\kat$ that covers its behavior. Semantic information is exploited where paths of the program are determined to be infeasible. For example, consider the simple \emph{instrumented} program\\
\texttt{asm(d==0); c=d; if (c==0) execB() else execD();}
The standard syntactic translation~\cite{kozen} alone, would produce the
  expression
  $\kM{c}{d==0}\cdot\kM{E}{c=d}\cdot(\kM{c}{c==0}\cdot\kM{B}{execB}+\overline{\kM{c}{c==0}}\cdot\kM{D}{execD}).$
  In our case, {\interproc} determines that $\textsf{c==0}$ is
  always true under the instrumented $\texttt{asm(d==0)}$ and the
  program is instead converted to the simpler expression $\kM{B}{execB}$.

\vskip3pt\noindent
$\bullet$\;{\bf ${\akatdiff}(k_1,k_2,\hypotheses)$:} Given two KAT expressions
$k_1,k_2$ and hypotheses $\hypotheses$, {\akatdiff} returns {\acexs}, which is a
set of KAT expressions $k$ with $k\in \katminus{k_1}{k_2}{\hypotheses}$ and
possibly $k\in \katminus{k_2}{k_1}{\hypotheses}$ (depending on whether we seek
equivalence or inclusion). We assume this sub-procedure to be sound and
complete. If {\acexs} is empty, then the two input
KAT expressions $k_1$ and $k_2$ are such that $k_1\leq_{\hypotheses} k_2$ (or
$k_1\equiv_{\hypotheses} k_2$). Our \emph{implementation} uses
  {\symkat}~\cite{symkattool}, which only obtains a \emph{singleton} set
  {\acexs}, and is thus either (i) a single string $c$ in
  $\katminus{k_1}{k_2}{\hypotheses}$ when we seek inclusion  or (ii) a pair of strings $(c_1,c_2)$, with $c_1$ in $\katminus{k_1}{k_2}{\hypotheses}$ and $c_2$ in $\katminus{k_2}{k_1}{\hypotheses}$, when we seek equivalence. If $k_1=a\cdot M\cdot (b\cdot F+\overline{b}\cdot G)$ and $k_2=a\cdot M\cdot\overline{b}\cdot G$, then the string $a\cdot M\cdot b\cdot F$ is included in $k_1$ but not in $k_2$, and thus in $\katminus{k_1}{k_2}{\hypotheses}$.

\vskip3pt\noindent
$\bullet$\;{\bf ${\asolvediff}(k_1,k_2,\hypotheses,{\acexs})$:} This procedure
takes KAT expressions $k_1$ and $k_2$, a set of hypotheses, and the set
  of counterexamples $\acexs$ above. It returns a set $R$ of tuples $(r_1,r_2,\hypotheses_r)$, each called a \emph{restriction}. Restrictions $R$ has the property that $\proj{1}{R}$ partitions $k_1$, ensuring that we have completely covered all traces. Furthermore, in the interest of \emph{progress}, we also assume that each counterexample in $\acexs$ is not a counterexample for $\katminus{k_1\cap r_1}{k_2\cap r_2}{\hypotheses_r}$, or even for $\katminus{k_2\cap r_2}{k_1\cap r_1}{\hypotheses_r}$ depending on whether equivalence is considered instead of inclusion. In our \emph{implementation}, we apply a customized edit-distance algorithm discussed in Sec.~\ref{sec:edit-distance}, which returns a set of \emph{transformations} that can be applied to two KAT strings $c_1$ and $c_2$ to make them equivalent. These transformations are in the form of removing alphabet symbols from the strings at particular locations, or replacing some symbol with another. From these transformations, {\asolvediff} constructs a list of restrictions to be applied on the input programs of the form $(r_1,r_2,\hypotheses_r)$, where $r_1$ and $r_2$ are KAT expressions, and $\hypotheses_r$ is a set of hypotheses.

When the edit-distance algorithm asks for the removal of an alphabet symbol from, say, string $c_1$, we consider two cases, depending on whether the symbol corresponds to a boolean condition or not. If so, the KAT expression $r_1$ corresponding to this transformation is essentially obtained by adding a hypothesis inserting the valuation of the boolean variable, in the given KAT expression $k_1$. Since we want these restrictions to cover all behaviors of the input programs, we also consider the negation of that valuation. As such, at least two restrictions are considered, namely $(r_1,r_2,\hypotheses_r)$ and $(r_1',r_2',\hypotheses_r')$, such that $r_1$ and $r_1'$ cover $k_1$. On the other hand, when the removal of an event $\kM{M}{}$ is required, a hypothesis of the form $\kM{M}{}\equiv 1$ is added to the set of hypotheses.

\vskip3pt\noindent
$\bullet$\;{\bf $\arestrict(C_1,r_1,C_2,r_2,\hypotheses,\alpha)$} obtains new
programs from previous ones, using restrictions from {\asolvediff}. Given
programs $C_1,C_2$, KAT restrictions $r_1,r_2$, a set of hypotheses
$\hypotheses$ and current abstraction $\alpha$, this sub-procedure returns a
tuple $(D_1,D_2,\alpha')$, where $D_1,D_2$ are the new programs and $\alpha'$ is
a new abstraction that refines $\alpha$, such that, for $i\in\{1,2\}$,
$\atranslate(D_i,\alpha')\leq_{\hypotheses} \atranslate(C_i,\alpha)\cap r_i$,
but $\atranslate(D_i,\alpha')\equiv_{\hypotheses} \atranslate(C_i,\alpha')\cap
r_i$. In other words, the KAT expression obtained
from the new program $D_i$ under the new refined abstraction $\alpha'$, is
included in the KAT expression from the original program $C_i$ under the old
abstraction $\alpha$ using the restriction $r_i$, but at the same time, if we
used the new abstraction $\alpha'$ to translate the program $C_i$ under the
restriction $r_i$, into a KAT expression we would obtain the same as by just
translating $D_i$ under the new abstraction. Our \emph{implementation}
  restricts programs by instrumenting \texttt{assume} statements on appropriate lines of code. 
For example, for a program
$(\kkttb\cdott\kkttO+\overline{\kkttb}\cdott1)^*$
we can implement restriction
$r=(\kkttb\cdot (\kkttb\cdott\kkttO+\overline{\kkttb}\cdott1))^*=(\kkttb\cdott\kkttO)^*$
with an \texttt{assume(l==true)} instrumented immediately inside the body of the
corresponding while loop. This can be seen in the output of our tool shown in Fig.~\ref{fig:solution}.
\subsection{Formal Guarantees}

The key challenge is soundness, even under the sub-procedure assumptions noted above, and \omittedproofreflc{\ref{apx:omitted-automation}}.

\begin{theorem}[Soundness]\label{thm:sound}
	For all $C_1,C_2$, and abstractions $\alpha$, let $O=\asearch(C_1,C_2,\emptyset,\alpha)$, let $\alpha'$ be the common abstraction of $O$ and let $k_1=\translate{C_1}{\alpha'}$ and $k_2=\translate{C_2}{\alpha'}$. Then $\absrefines{k_1}{k_2}{\interfaceFunc(O)}{}{}$.


\end{theorem}

Weak completeness is easier because, as per Remark~\ref{rem:trivial-solution}, trivial solutions can be constructed. So we are more interested in generating increasingly precise solutions. For progress, as long as the sub-procedure {\asolvediff} returns restrictions that handle the counterexamples returned by {\akatdiff}, then these counterexamples will not be seen again in the recursive steps that follow.


\vspace{-5pt}
\section{Edit-distance on expressions and strings}
\label{sec:edit-distance}

Our main algorithm depends on
{\asolvediff} to examine a pair of KAT expressions $k_1,k_2$, a set of hypotheses $\hypotheses$, as well
as counterexamples to their equivalence, and determine appropriate
restrictions $r_1,r_2$ and additional hypotheses $\hypotheses'$ that could be used to
further search for trace classes of $k_1$ that are contained in $k_2$,
up to hypotheses $\hypotheses\cup\hypotheses'$. To achieve this, {\asolvediff} tries to identify the differences between the KAT expressions $k_1$ and $k_2$, or between their string-based counterexamples, and attempts to find useful restrictions of least impact, to apply to the two input programs. As such, we implemented a sub-procedure {\acexscore}, that takes as inputs two KAT strings $c_1$ and $c_2$, or two KAT expressions $k_1$ and $k_2$, and returns a list of \emph{scored} transformations to be applied on the two strings (or KAT expressions) in order to make them equivalent. In our implementation we use the custom edit-distance algorithm only on counterexample strings, and in Apx~\ref{sec:global-edit-distance}, we discuss how the global edit-distance for \emph{general} KAT expressions can help in conjunction with the composition results of Section~\ref{subsec:composition}. The edit-distance on such KAT expressions has to handle the structure of the expressions, and is naturally more involved than the linear one on strings. (The former is more similar to trees~\cite{treeedit}.) The idea behind the sub-procedure {\acexscore} is similar to edit-distance algorithms in the literature for comparing two strings/trees/graphs~\cite{treeedit}. These edit-distance algorithms, return a sequence of usually simple single-symbol transformations that are classified as symbol removals, insertions, and replacements, that equate the two input strings when they are applied on them.
%
%

We had to customize edit distance for our purposes of cross-program correlation. We thought that inserting a symbol in one of the two strings or KAT expressions, is less natural than removing another one from the other string or expression, and encode such insertions in one string as removals from the other. Therefore we employ just removal and replacement transformations on the two inputs:
\begin{itemize}
  \item $\textsf{Remove}(c,s)$: Returns a new string obtained from $s$ with the symbol $c$ removed,
  \item $\textsf{Replace}(c_1,c_2,s)$: Returns a new string obtained from $s$ with the symbol $c_1$ replaced with the symbol $c_2$.
\end{itemize}

Note that each copy of each symbol in the string is uniquely labeled, and the transformations above speak about these labeled symbols, making the order in which the transformations are applied irrelevant. Moreover, in our experience, certain transformations have more impact, or are in some way {\it heavier} than others. As such, we attempt to score them, and use the score of each individual transformation to ultimately score the whole sequence of transformations. For example, replacing an event symbol $\kM{M}{}$ in some string with another symbol $\kM{N}{}$, is certainly a transformation that is semantically more involved than simply setting a boolean symbol $\kM{c}{}$ to $\true$.
The full algorithm for edit-distance can be found in Appendix~\ref{apx:edit-distance-alg} of the supplemental material.


\begin{example}
	Consider the two input strings $s_1=\kM{a}{}\cdot\kM{A}{}\cdot\kM{B}{}$ and $s_2=\kM{d}{}\cdot\kM{e}{}\cdot\kM{B}{}$. Running the procedure {\acexscore} on $s_1$ and $s_2$, will return the pair $(T,S)$, where $T$ is the sequence of transformations $[\mathsf{Replace}(\kM{a}{},\kM{d}{},s_1),\mathsf{Remove}(\kM{e}{},s_2),\mathsf{Remove}(\kM{A}{},s_1)]$ and the score $S=\mathsf{replace\_scr}+2*\mathsf{remove\_scr}$, where $\mathsf{replace\_scr}$ is the cost of replacing one symbol with another of same type, and $\mathsf{remove\_scr}$ is the cost of removing a symbol from one of the input strings. With this sequence of transformations, both strings become equal to $\kM{d}{}\cdot\kM{B}{}$.
\end{example}

The sequence of transformations returned by {\acexscore} is converted into the one {\asolvediff} returns, as follows, for $\kM{A}{}, \kM{B}{}$ action symbols and $\kM{a}{},\kM{b}{}$ boolean symbols.
\begin{itemize}[leftmargin=10pt]
	\item $\mathsf{Remove(\kM{A}{},s)}$: add a new hypothesis $\kM{A}{}\equiv 1$
	\item $\mathsf{Remove(\kM{a}{},s)}$: perform case analysis and include two tuples of restrictions, resp. corresponding to setting $\kM{a}{}$ to $\true$ and setting $\kM{a}{}$ to $\false$
	\item $\mathsf{Replace(\kM{A}{},\kM{B}{},s)}$: add a new hypothesis $\kM{A}{}\equiv \kM{B}{}$
	\item $\mathsf{Replace(\kM{a}{},\kM{b}{},s)}$: add a new hypothesis $\kM{a}{}\equiv \kM{b}{}$
\end{itemize}

\vspace{-5pt}
\section{Evaluation}
\label{sec:eval}

\newcommand\EVsome{ {\checkmark} } 
\newcommand\EVnone{ $\emptyset$ }
\newcommand\EVcmp{$=^\interfaceRel$}
\newcommand\EVclt{$\leq^\interfaceRel$}

\paragraph{Implementation}
We have realized our algorithm in a prototype tool called {\knotical}.
Our tool is written in {\ocaml}, using {\interproc} as an
abstract interpreter~\cite{interproctool}, and {\symkat} as a symbolic solver for
KAT equalities~\cite{symkat,symkattool}.
We have described implementation choices made for the tool's subcomponents in Section~\ref{subsec:subprocedures}. {\knotical} generates multiple solutions, internally represented in the form of trees. During the {\akatdiff} and {\asolvediff} steps of the algorithm, multiple choices can be made and each solution tree corresponds to a particular set of choices.
%
Branching in a solution tree corresponds to the different restrictions applied and their complement, as a result of performing case analysis on a particular condition.
%
Often the solution trees (or subtrees) are partial, in the sense that the different restrictions applied to the programs, when taken together do not cover all behaviors of the input programs. Partial solutions can readily be converted to complete ones (See Remark~\ref{rem:trivial-solution})

\begin{figure}
  \centering
  \scalebox{0.95}{
		\setlength{\tabcolsep}{1pt}
\begin{tabular}{|r|l|r|r|c||r|r||r|r|r|r|}
\hline
\rowcolor{lightgray}
 & & & & & {\bf Time} & & \multicolumn{2}{|c|}{{\bf Tuples}} & \multicolumn{2}{|c|}{{\bf Hypos}} \\
\rowcolor{lightgray}
{\bf \#} & {\bf Benchmark} & {\bf loc} & {\bf $f$s} & {\bf Dir} & {\bf (s)} & {\bf Sols} & \emph{min} & \emph{max} & \emph{min} & \emph{max} \\
\hline
  1 & \texttt{0arith.c} &   28 &    4 & \EVcmp & 0.03 &     1 &    1 &    1 &    2 &    2 \\
  2 & \texttt{0complete.c} &   22 &    5 & \EVcmp & 0.02 &     1 &    2 &    2 &    2 &    2 \\
  3 & \texttt{0complete1.c} &   28 &    6 & \EVcmp & 0.09 &     2 &    1 &    2 &    3 &    4 \\
  4 & \texttt{0false.c} &   15 &    3 & \EVcmp & 0.01 &     1 &    1 &    1 &    1 &    1 \\
  5 & \texttt{0if.c} &   25 &    5 & \EVcmp & 0.02 &     1 &    1 &    1 &    2 &    2 \\
  6 & \texttt{0ifarecv.c} &   27 &    5 & \EVcmp & 0.02 &     1 &    1 &    1 &    2 &    2 \\
  7 & \texttt{0impos.c}$^\bullet$ &   19 &    4 & \EVclt & 0.01 &     0 &    0 &    0 &    0 &    0 \\
  8 & \texttt{0medstrai.c} &   46 &   22 & \EVcmp & 6.18 &     3 &    1 &    1 &    2 &    2 \\
  9 & \texttt{0needax.c} &   24 &    4 & \EVclt & 0.01 &     1 &    1 &    1 &    1 &    1 \\
 10 & \texttt{0nohyp.c} &   21 &    4 & \EVcmp & 0.04 &     2 &    1 &    1 &    2 &    2 \\
 11 & \texttt{0noloop.c} &   31 &    3 & \EVcmp & 0.25 &     2 &    1 &    1 &    0 &    1 \\
 12 & \texttt{0nondet.c} &   48 &    7 & \EVcmp & 0.41 &     2 &    2 &    2 &    4 &    4 \\
 13 & \texttt{0pos.c} &   22 &    4 & \EVclt & 0.12 &     1 &    1 &    1 &    0 &    0 \\
 14 & \texttt{0rename.c} &   13 &    4 & \EVcmp & 0.05 &     1 &    1 &    1 &    1 &    1 \\
 15 & \texttt{0rename1.c} &   14 &    4 & \EVcmp & 0.01 &     1 &    1 &    1 &    1 &    1 \\
 16 & \texttt{0sanity.c} &    8 &    3 & \EVcmp & 0.00 &     1 &    1 &    1 &    0 &    0 \\
 17 & \texttt{0sanity1.c} &    8 &    3 & \EVcmp & 0.01 &     2 &    1 &    1 &    1 &    1 \\
 18 & \texttt{0smstrai.c} &   45 &   22 & \EVcmp & 0.66 &     5 &    1 &    1 &    1 &    1 \\
 19 & \texttt{1acqrel.c} &   38 &    2 & \EVcmp & 0.02 &     1 &    1 &    1 &    0 &    0 \\
 20 & \texttt{1asendrecv.c} &   47 &    8 & \EVcmp & 2.45 &    39 &    2 &    3 &    5 &   10 \\
 21 & \texttt{1assume.c} &   35 &    4 & \EVclt & 0.98 &    44 &    1 &    2 &    3 &   10 \\
 22 & \texttt{1concloop.c} &   38 &    5 & \EVcmp & 0.72 &    14 &    1 &    1 &    1 &    5 \\
 23 & \texttt{1concloop2.c} &   34 &    4 & \EVclt & 4.60 &   240 &    1 &    2 &    6 &   19 \\
 24 & \texttt{1concloop3.c} &   29 &    3 & \EVcmp & 0.69 &   127 &    1 &    4 &    3 &   12 \\
 25 & \texttt{1linarith.c} &   57 &    4 & \EVcmp & 1.20 &    12 &    1 &    1 &    4 &    4 \\
 26 & \texttt{1loopevent.c} &   36 &    3 & \EVcmp & 4.73 &    67 &    1 &    3 &    1 &    5 \\
 27 & \texttt{1loopprint.c} &   35 &    3 & \EVcmp & 0.36 &    12 &    1 &    2 &    3 &    5 \\
 28 & \texttt{1sendrecv.c} &   49 &    7 & \EVclt & 3.84 &    75 &    2 &    7 &    6 &   19 \\
 29 & \texttt{1toggle.c} &   42 &    2 & \EVcmp & 0.03 &     1 &    1 &    1 &    1 &    1 \\
 30 & \texttt{2altern.c} &   25 &    4 & \EVcmp & 0.03 &     2 &    1 &    1 &    2 &    2 \\
 31 & \texttt{2cdown.c} &   23 &    4 & \EVcmp & 0.02 &     1 &    1 &    1 &    1 &    1 \\
 32 & \texttt{2foil.c} &   20 &    4 & \EVcmp & 0.01 &     1 &    1 &    1 &    1 &    1 \\
 33 & \texttt{3buffer.c} &   63 &    7 & \EVcmp & 21.07 &   192 &    1 &    2 &    6 &   11 \\
 34 & \texttt{3syscalls.c} &   59 &    7 & \EVcmp & 17.77 &   156 &    1 &    2 &    7 &   16 \\
 35 & \texttt{4ident.c} &   69 &    6 & \EVcmp & 0.50 &     6 &    1 &    1 &    3 &    3 \\
 36 & \texttt{5thttpdEr.c} &   44 &   10 & \EVclt & 0.21 &     5 &    2 &    3 &    2 &    3 \\
 37 & \texttt{5thttpdWr.c} &   43 &   10 & \EVclt & 1.87 &    62 &    1 &    2 &    4 &    8 \\
\hline
\end{tabular}

  }
\caption{\label{fig:results} Results of applying {\knotical} to \numbenchmarks\ benchmarks. Those marked with $\bullet$ are expected to have no solutions.}
\end{figure}
\paragraph{Benchmarks.}
We have evaluated our approach by applying our tool to a collection of \numbenchmarks\ new benchmarks. Each benchmark includes two program fragments denoted $C_1$ and $C_2$. They can be found in Appendix~\ref{apx:fullresults} of the supplemental materials. Broadly speaking, our benchmarks categorized as:
  (\texttt{0*.c}) -- Program pairs that exercise various
    technical aspects of our algorithm, such as cases where refinement
    is trivial, concrete refinement holds, refinement can be achieved
    entirely from case-splits, and where refinement can only be
    achieved through aggressive introduction of hypotheses.
   (\texttt{1*.c}) -- Program pairs that involve user I/O, system calls, acquire/release, and reactive web servers.
   (\texttt{2*.c}) -- Program pairs that involve tricky patterns, requiring careful alignment between two fragments.
   (\texttt{[345]*.c}) -- These program pairs are more challenging: \texttt{3buffer.c} and \texttt{3syscalls.c} model array access patterns with complicated array iterations, and \texttt{4ident.c} involves a larger pair with identical code. Others model reactive web servers.

\noindent
Some of the more challenging examples include \texttt{5thttpdWr.c} and \texttt{5thttpdEr.c}, each containing a pair of fragments taken from the \texttt{thttpd}~\cite{thttpd} and Merecat~\cite{merecat} HTTP servers. These two servers are related because Merecat is an extension of \texttt{thttpd} that adds SSL support. These benchmarks contain distillations from the two servers, summarizing how they diverge in handling a request. Merecat, unlike \texttt{thttpd}, performs compression, uses SSL to write responses, and has a keep-alive option so that connections aren't closed when an error occurs. We have manually decomposed the two programs into two phases: writing a request (\texttt{5thttpdWr.c}) and the subsequent error handling (\texttt{5thttpdEr.c}), demonstrating the compositional nature of our relations.


\paragraph{Results.}
We ran {\knotical} on a MacBook Pro with a 3.1 GHz Intel Core i7 CPU and 16GB RAM, 
using the OCaml 4.06.1 compiler.
Some of the generated trace-refinement relations are shown in Appendix~\ref{apx:fullresults} of the supplemental materials.
The table in Figure~\ref{fig:results} summarizes these results, including the performance of
{\knotical}.  For each benchmark, we have
included the lines of code ({\bf loc}) and number of procedures ({\bf $f$s}). We also indicate ({\bf Dir}) whether the benchmark is
for refinement \EVclt\ or for equality \EVcmp. For some of the examples,
  we check only \EVclt\ because we wanted to ensure the tool was capable
of this antisymmetric reasoning. Some of the \EVclt\ examples were crafted for this purpose.

Next, we report the total time it took in seconds ({\bf Time}), as
well as the number of solutions discovered ({\bf Sols}).
%
We also report some basic statistics about the solutions generated for each benchmark. We report the number of {\bf Tuples} in the solution that has the fewest/most (\emph{min/max}) tuples. Similarly, we report the number of hypotheses ({\bf Hypos}) in the solution that has the fewest/most (\emph{min/max}) hypotheses. These statistics help show the quality of the solutions. Intuitively, fewer hypotheses means that the programs are more similar. We also evaluated the quality of the generated trace-refinement relations by inspecting many of them manually.\todo{IDE?}

\paragraph{Discussion.} In most cases, {\knotical} was able to generate expected solutions quickly, often in fractions of a second. For simpler benchmarks (\texttt{0*.c}), there were often concise solutions with either two tuples (due to a single case-split) or one tuple (due to hypotheses). \texttt{0nondet.c} is more complicated and both of its solutions had 4 tuples. More complicated benchmarks tended to have solutions with 3 to 7 tuples.  The largest number of tuples in a solution was 7 (\texttt{1sendrecv.c}) and the largest number of hypotheses in a solution was 19 (\texttt{01concloop2.c} and \texttt{1sendrecv.c}). Benchmark \texttt{1acqrel.c} had a solution with 0 hypotheses because it
contains non-terminating loops, which are translated to KAT expressions $0$.
Benchmark \texttt{0impos.c} is expected to have no solutions because its fragments contain
two different non-removable events, that cannot be made equivalent with axioms. ({\knotical} permits users to specify events that cannot be ignored.) 
Benchmarks \texttt{1concloop2.c}, \texttt{3buffer.c},
and \texttt{3syscalls.c} had hundreds of solutions because they have many complicated 
conditional branch and loop conditions. Case analysis on the permutations of these conditions leads to many solutions.
There is not much correlation between analysis running time (or number of solutions) and lines of code. There is a stronger correlation with code complexity: many events or conditions lead to longer analysis time. \texttt{3syscalls.c}, \eg, took longer and yielded more solutions.
In summary, our algorithm and tool {\knotical}, promptly generate concise trace-refinement relations.


\vspace{-5pt}
\section{Related Work}\label{sec:relwork}

To the best of
our knowledge, we are the first to define trace-refinement relations in terms
of programs' event behaviors over time, which we call \emph{trace}-oriented
refinement.  Prior works~\cite{benton,DBLP:journals/tcs/Yang07} view refinement in terms
of \emph{state} relations.

\paragraph{Bisimulation.}
A bisimilarity relation is over states and expresses that whenever one can perform an action from some state on one system, one can also perform the same action from any bisimilar state on the other system, and reach bisimilar states. While some weakenings of bisimulation have been shown\todo{CONCUR94,FDR,Bonchu,Pous}, they don't capture the types of equivalence discussed here. 
The very way in which we formulate program equivalences in this paper (expressing program behaviors over time as KAT expressions) is fundamentally different from bisimulation (which relies on step-by-step state relations).
Bisimulation is unable to capture that $\textsf{A}\cdot (\textsf{B}+\textsf{C})$ has the same behavior as $\textsf{A}\cdot \textsf{B}+\textsf{A}\cdot \textsf{C}$. It would also be tedious to use bisimulations to express that events commute ($A\cdot B\equiv B\cdot A$) or event inverses ($A\cdot B = 1$).


\paragraph{Concrete, state-based semantic differencing.}
Other recent works relate
a function's output to its input.
%
Lahiri \emph{et al.} describe {\symdiff}~\cite{LahiriHKR2012,LahiriMSH2013},
defining ``differential assertion checking,'' which says that from an
initial state that was non-failing on $C$, it becomes failing
on $C'$.
Their approach to assertion checking
bares some similarity to self-composition~\cite{barthe2004secure,terauchi2005secure}.
%
Godlin and Strichman~\cite{GS2009} offered support for mutual recursion.
Wood \emph{et al.}~\cite{WDLE2017} tackle program equivalence
in the presence of memory allocation and garbage collection. 
%
Unno \emph{et al.}~\cite{UTS17} describe a method of verifying relational
specifications based on Horn Clause solving. 
%
Jackson and Ladd~\cite{JL1994} describe an approach based on
dependencies between input and output variables, but do not offer
formal proofs.
%
Gyori \emph{et al.}~\cite{GLP2017} took steps
beyond concrete refinement, using equivalence
relations, similar to those of Benton, for dataflow-based change impact
analysis. 

\paragraph{Other works.}
Bouajjani \emph{et al.}~\cite{BEL2017} also eschew state refinement
relations in favor of a more abstract relationship between programs.
%
They focus on
concurrency questions that arise from reordering program statements
and/or re-orderings due to interleaving. 
The authors don't work with traces in the sense defined here;
rather, their traces are data-flow abstractions, represented as graphs.
There are some
analogies between $k$-safety of a single program, and reasoning about two programs.
Researchers have explored relational invariants (over multiple executions of a single program) via program transformations that ``glue''
copies of the program to itself, including self-composition~\cite{barthe2004secure,terauchi2005secure}, product programs~\cite{barthe2011product}, Cartesian Hoare logic~\cite{chl} and decomposition for $k$-safety~\cite{blazer}.
%
Logozzo \emph{et al.} describe \emph{verification modulo versions}~\cite{vmv}
and explore how necessary/sufficient environment conditions for a program $C$'s
safety can be used to determine whether program $C'$ introduced a regression or
is ``correct relative to $C$''.
The work does not involve refinement relations.
Composition for (non-relational) temporal logic was explored by
Barringer \emph{et al.}~\cite{DBLP:conf/stoc/BarringerKP84}, who introduced
the ``chop'' operator. 
Pous introduced a symbolic approach for determining language
equivalence between KAT expressions~\cite{Pous}
(see Sec.~\ref{sec:algorithm}).
%
Kumazawa and Tamai use edit distance to characterize the difference between
counterexamples within a single program (infinite vs lasso
traces)~\cite{kumazawa2011counterexample}.

\vspace{-15pt}
\section{Conclusion and Future Work}

We introduced \emph{trace refinement relations}, going beyond
the state refinement relations~\cite{benton,DBLP:journals/tcs/Yang07,GLP2017,UTS17}.
Our relations express trace-oriented restrictions on
a program behavior and case-wise correlate the
behaviors of another.
We have further provided a novel
synthesis algorithm, based on abstract
interpretation, KAT solving, restriction, and edit-distance. We have shown
with {\knotical}, the first tool capable of synthesizing trace-refinement relations,
that this approach is promising.
%
As discussed in Apx.~\ref{sec:global-edit-distance}, we plan to further
explore using edit-distance at both global and local levels.
Another avenue is to explore how temporal verification
can be adapted to trace-refinement relations.


\bibliographystyle{ACM-Reference-Format}
\bibliography{biblio}

\newpage
\appendix
\onecolumn
\begin{center}
  {\LARGE\bf Appendix}

\end{center}
\bigskip
\bigskip
\section{Omitted Lemmas and Proofs}
\label{apx:omitted}

\begin{theoremApp}{thm:concrete-and-kat-concrete-refinement}
	Let $C_1$ and $C_2$ be two programs, and let $k_1$ and $k_2$ be the two KAT expressions \Tremoved{that correspond to}\Tadded{obtained from a strongly valid translation of} the two programs respectively, under some abstraction $\alpha$. Then it holds that $\refines{C_1}{C_2}$ if and only if $\refines{k_1}{k_2}$.
\end{theoremApp}

\begin{proof}
	For what follows, we write $\alpha$ for $\alpha_B\circ \alpha_S$.
	For the {\it only if} direction, suppose that $\refines{C_1}{C_2}$ and pick any $b,d\in\bkat$. Suppose first that $b\cdot k_1\equiv 0$. Then let $\sigma_1,\ldots,\sigma_n$ be the set of states such that $\statetobool{\sigma_1}+\ldots+\statetobool{\sigma_n}\equiv b$ for $i\leq n$. Then, for all $i\leq n$, $\config{C_1}{\sigma_i}\leadsto \fault$ which implies that $\config{C_2}{\sigma_i}\leadsto \fault$ by assumption that $\refines{C_1}{C_2}$. This means that for all $i\leq n$, $\statetobool{\sigma_i}\cdot k_2\equiv 0$, and thus $b\cdot k_2\equiv\statetobool{\sigma_1}+\ldots+\statetobool{\sigma_n}\cdot k_2\equiv 0$.
	
	The second condition states that $b\cdot k_1\cdot d\not\equiv 0$ implies $b\cdot k_2\cdot d\not\equiv 0$, or $b\cdot k_2\equiv 0$.
	Assume that $b\cdot k_1\cdot d\not\equiv 0$. Let $\sigma_1,\ldots,\sigma_n$ be the set of states such that $b=\statetobool{\sigma_1}+\ldots+\statetobool{\sigma_n}$, and let $\rho_1,\ldots,\rho_m$ be the set of states such that $d=\statetobool{\rho_1}+\ldots+\statetobool{\rho_m}$. Therefore, $(\statetobool{\sigma_1}+\ldots+\statetobool{\sigma_n})\cdot k_1\cdot (\statetobool{\rho_1}+\statetobool{\rho_m})\not\equiv 0$. Let $(\sigma_i,\rho_j)$ be all the pairs, such that $\statetobool{\sigma_i}\cdot k_1\cdot\statetobool{\rho_j}\not\equiv 0$. It follows by definition that $\config{C_1}{\sigma_i}\leadsto\rho_j$. Therefore, either $\config{C_2}{\sigma_i}\leadsto\rho_j$ or $\config{C_2}{\sigma_i}\leadsto\fault$ by assumption that $\refines{C_1}{C_2}$. It follows that $b\cdot k_2\equiv(\statetobool{\sigma_1}+\ldots+\statetobool{\sigma_n})\cdot k_2\equiv 0$ or $b\cdot k_2\cdot d\equiv(\statetobool{\sigma_1}+\ldots+\statetobool{\sigma_n})\cdot k_2\cdot (\statetobool{\rho_1}+\ldots+\statetobool{\rho_m})\not\equiv 0$, as required.
	
	
	
	For the {\it if} direction, suppose that $\refines{k_1}{k_2}$, and let $\sigma,\rho$ be any two states in $\states$. If $\config{C_1}{\sigma}\leadsto\fault$, then $\statetobool{\sigma}\cdot k_1\equiv 0$. By the assumption that $\refines{k_1}{k_2}$, we have that $\statetobool{\sigma}\cdot k_2\equiv 0$. Therefore, $\config{C_2}{\sigma}\leadsto\fault$. On the other hand, if $\config{C_1}{\sigma}\leadsto\rho$, then $\statetobool{\sigma}\cdot k_1\cdot\statetobool{\rho}\not\equiv 0$. By the assumption that $\refines{k_1}{k_2}$, it follows that $\statetobool{\sigma}\cdot k_2\cdot\statetobool{\rho}\not\equiv 0$ or $\statetobool{\sigma}\cdot k_2\equiv 0$. Therefore, $\config{C_2}{\sigma}\leadsto\rho$ or $\config{C_2}{\sigma}\leadsto\fault$, as required.
\end{proof}

\begin{lemma}\label{lem:less-than-intersection}
	For any $k,l$ in a KAT $\kat$, if $k\leq l$ then $k\cap l=k$.
\end{lemma}

\begin{proof}
	By definition, $k\cap l$ is equal to $m_1+\ldots+m_n$, where $\{m_1,\ldots,m_n\}$ is the set of all elements $m$ in $\kat$ such that $m\leq k$ and $m\leq l$. By assumption, $k$ is equal to $m_i$ for some $i\leq n$. Therefore, $k=m_1+\ldots+m_n$ as required.
\end{proof}

\begin{lemma}\label{lem:superset-hypotheses}
	Suppose that $k_1,k_2\in\kat$ and $\hypotheses$ a set of hypotheses such that $k_1\leq_{\hypotheses} k_2$. Then for any $\hypotheses'$ with $\hypotheses\subseteq \hypotheses'$, $k_1\leq_{\hypotheses'} k_2$.
\end{lemma}

\begin{lemma}\label{lem:kat-less-product-disjunction}
	Let $k_1,k_2,l_1$ and $l_2$ be elements of a KAT $\kat$. If $k_1\leq l_1$ and $k_2\leq l_2$, then $k_1\cdot k_2\leq l_1\cdot l_2$ and $k_1+k_2\leq l_1+l_2$.
\end{lemma}

\begin{proof}
	Firstly notice that $k_1+l_1=l_1$ and $k_2+l_2=l_2$. For the first inequality, we want to show that $k_1\cdot k_2+l_1\cdot l_2 = l_1\cdot l_2$. Using the aforementioned equalities, $k_1\cdot k_2+l_1\cdot l_2=k_1\cdot k_2+(k_1+l_1)\cdot (k_2+l_2)=k_1\cdot k_2+k_1\cdot k_2+k_1\cdot l_2+l_1\cdot k_2+l_1\cdot l_2=k_1\cdot k_2+k_1\cdot l_2+l_1\cdot k_2+l_1\cdot l_2=(k_1+l_1)\cdot (k_2+l_2)=l_1\cdot l_2$ as required. For the second inequality, we have that $k_1+k_2+l_1+l_2=(k_1+l_1)+(k_2+l_2)=l_1+l_2$ as needed.
	
\end{proof}

\begin{lemma}\label{lem:kat-less-elements}
	Let $k_1,k_2\in\kat$. It holds that $k_1\leq k_2$ if and only if for all $m\in\kat$, $m\leq k_1$ implies $m\leq k_2$.
\end{lemma}

\begin{proof}
	For the \emph{if} direction, suppose that for all $m\in\kat$, $m\leq k_1$ implies $m\leq k_2$. Then in particular, for $m=k_1$, $k_1\leq k_1$ implies that $k_1\leq k_2$.
	
	For the \emph{only if} direction, suppose that $k_1\leq k_2$, and suppose for contradiction that there is $m\in\kat$ such that $m\leq k_1$ but $m\not\leq k_2$. Then $m+k_1=k_1$, but $m+k_2\neq k_2$. From the assumption that $k_1\leq k_2$, it follows that $k_1+k_2=k_2$. Thus $k_1+m+k_2=k_2$, which implies that $m+k_2=k_2$. It follows that $m\leq k_2$, which is a contradiction.
\end{proof}


\begin{lemma}\label{lem:intersection-distribution-right}
	Let $k,l,o,p$ be elements of some KAT $\kat$, and let $o\leq k$ and $p\leq l$. Then $(k\cdot l)\cap(o\cdot p)\leq (k\cap o)\cdot (l\cap p)$ and $(k+ l)\cap(o+ p)\leq (k\cap o)+ (l\cap p)$.
\end{lemma}

\begin{proof}
	We consider the first inequality first, namely, $(k\cdot l)\cap(o\cdot p)\leq (k\cap o)\cdot (l\cap p)$. By definition of intersection, we have that $(k\cdot l)\cap(o\cdot p)\leq o\cdot p$ and $(k+ l)\cap(o+ p)\leq o+ p$. Since $o\leq k$ and $p\leq l$, by Lemma~\ref{lem:less-than-intersection}, we have that $k\cap o=o$ and $l\cap p=p$. Therefore $(k\cdot l)\cap(o\cdot p)\leq (k\cap o)\cdot (l\cap p)$ and $(k+ l)\cap(o+ p)\leq (k\cap o)+ (l\cap p)$ as required.
\end{proof}

\begin{lemma}\label{lem:intersection-distribution-left}
	Let $k,l,o,p$ be elements of some KAT $\kat$. Then $(k\cap o)\cdot (l\cap p)\leq (k\cdot l)\cap(o\cdot p)$ and $(k\cap o)+ (l\cap p)\leq (k+ l)\cap(o+ p)$.
\end{lemma}

\begin{proof}
	Consider the expressions $(k\cap o)$ and $(l\cap p)$. By definiton, $k\cap o=x_1+\ldots+x_M$, where $\{x_1,\ldots,x_M\}$ is the set of all elements $x$ in $\kat$ such that $x\leq k$ and $x\leq o$. Similarly, $(l\cap p)=y_1+\ldots+y_N$ where $\{y_1,\ldots,y_N\}$ is the set of all $y$ in $\kat$ such that $y\leq l$ and $y\leq p$. Therefore, by Lemma~\ref{lem:kat-less-product-disjunction}, for any $x$ in the first set and any $y$ in the second set, $x+y\leq k+l$, $x+y\leq o+p$, $x\cdot y\leq k\cdot l$ and $x\cdot y\leq o\cdot p$.
	
	For the first inequality, namely, $(k\cap o)\cdot (l\cap p)\leq (k\cdot l)\cap(o\cdot p)$, notice that $(k\cap o)\cdot (l\cap p)$ is equal to $(x_1+\ldots+x_M)\cdot(y_1+\ldots+y_M)=(x_1\cdot y_1)+(x_1\cdot y_2)+\ldots+(x_i\cdot y_j)+\ldots+(x_M\cdot y_M)$. Therefore, $(k\cap o)\cdot(l\cap p)\leq k\cdot l$ and $(k\cap o)\cdot(l\cap p)\leq o\cdot p$, and thus $(k\cap o)\cdot(l\cap p)\leq (k\cdot l)\cap(o\cdot p)$, as required.
	Similarly, for the second inequality, notice that $(k\cap o)+ (l\cap p)$ is equal to $(x_1+\ldots+x_M)+(y_1+\ldots+y_M)=(x_1+ y_1)+(x_1+ y_2)+\ldots+(x_i+ y_j)+\ldots+(x_M+ y_M)$. Therefore, $(k\cap o)+(l\cap p)\leq k+ l$ and $(k\cap o)+(l\cap p)\leq o+ p$, and thus $(k\cap o)+(l\cap p)\leq (k+ l)\cap(o+ p)$.
\end{proof}

\begin{lemma}\label{lem:intersection-star-right}
	Let $k$ and $o$ be elements of some KAT $\kat$, and let $o\leq k$. Then $k^*\cap o^*\leq (k\cap o)^*$.
\end{lemma}

\begin{proof}
	It suffices to show that for all $n,m\in\mathbb{N}$, $k^n\cap o^m\leq (k\cap p)^*$. Let $m$ be any element of $\kat$, such that $m\leq k^n$ and $m\leq o^m$. Then, since $o\leq k$, by Lemma~\ref{lem:less-than-intersection} it holds that $o=k\cap o$, and therefore, $m\leq o^m$ implies that $m\leq (k\cap o)^m$, and thus $m\leq (k\cap o)^*$. Since $m$ was chosen arbitrarily among the elements $x$ in $\kat$ for which $x\leq k^n\cap o^m$, by Lemma~\ref{lem:kat-less-elements}, the result follows.
\end{proof}

\begin{lemma}\label{lem:intersection-star-left}
	Let $k$ and $o$ be elements of some KAT $\kat$. Then $(k\cap o)^*\leq k^*\cap o^*$.
\end{lemma}

\begin{proof}
	Let $M={m_1,\ldots,m_n}$ be the set of all elements $m$, such that $m\leq k$ and $m\leq o$. Therefore, $(k\cap o)^* = (m_1+\ldots+m_n)^*$, for $m_i\in M$. It suffices to show that for all $u\in\mathbb{N}$, $(m_1+\ldots+m_n)^u\leq k^*\cap o^*$. In particular, it is enough to show that $u\in\mathbb{N}$, $(m_1+\ldots+m_n)^u\leq k^u\cap o^u$. The latter is equal to $z_1+\ldots +z_s$, for $z_i\leq k^u$ and $z_i\leq o^u$. Notice that for any element $x\leq (m_1+\ldots+m_n)^u$, there is a function $f:[u]\rightarrow [n]$, such that $x\leq m_{f(1)}\cdot m_{f(1)}\cdots m_{f(u)}$. Since for all $j\leq u$, $m_{f(j)}\leq k$ and $m_{f(j)}\leq o$, it follows that $m_{f(1)}\cdot m_{f(1)}\cdots m_{f(u)}\leq k^u$ and $m_{f(1)}\cdot m_{f(1)}\cdots m_{f(u)}\leq o^u$. Hence, $m_{f(1)}\cdot m_{f(1)}\cdots m_{f(u)}\leq k^u\cap o^u$. Since $x$ was chosen arbitrarily, the result follows by Lemma~\ref{lem:kat-less-elements}.
\end{proof}

\begin{theorem}\label{thm:refinement-sequence}
	Suppose $k_1,k_2,l_1$ and $l_2$ are KAT expressions. Let $\interfaceRel_k$ and $\interfaceRel_l$ be trace-refinement relations, such that $\abskrefines{k_1}{k_2}{\interfaceRel_k}$ and $\abskrefines{l_1}{l_2}{\interfaceRel_l}$. Then $\abskrefines{k_1\cdot l_1}{k_2\cdot l_2}{\interfaceRel_k\interfaceOpComp \interfaceRel_l}$.
\end{theorem}

\begin{proof}
	We want to show that for any tuple $(x,y,\hypothesesD)\in\interfaceRel_k\interfaceOpComp\interfaceRel_l$, $(k_1\cdot k_2)\cap x\leq_{\hypothesesD}(l_1\cdot l_2)\cap y$. Choose such an arbitrary tuple $(x,y,\hypothesesD)\in\interfaceRel_k\interfaceOpComp\interfaceRel_l$, and let $(o,q,\hypotheses)\in\interfaceRel_k$ and $(p,r,\hypothesesB)\in\interfaceRel_l$ be the tuples that produced $(x,y,\hypothesesD)$. In other words, $x=o\cdot p$, $y=q\cdot r$ and $\hypothesesD=\hypotheses\cup \hypothesesB$.
	
	Since $\proj{1}{\interfaceRel_k}$ partitions $k_1$ and $\proj{1}{\interfaceRel_l}$ partitions $l_1$, we have that for any $o\in\proj{1}{\interfaceRel_k}$ and $p\in\proj{1}{\interfaceRel_l}$, $o\leq k$ and $p\leq l$, and thus by Lemma~\ref{lem:intersection-distribution-right}, we have that $(k_1\cdot l_1)\cap (o\cdot p)\leq (k_1\cap o)\cdot (l_1\cap p)$.
	By assumption, $k_1\cap o\leq_{\hypotheses} k_2\cap q$ and $l_1\cap p\leq_{\hypothesesB} l_2\cap r$. Since $\hypothesesD=\hypotheses\cup \hypothesesB$, by Lemma~\ref{lem:superset-hypotheses}, we have that $k_1\cap o\leq_{\hypothesesD} k_2\cap q$ and $l_1\cap p\leq_{\hypothesesD} l_2\cap r$. Therefore, by Lemma~\ref{lem:kat-less-product-disjunction}, we have that $(k_1\cap o)\cdot(l_1\cap p)\leq_{\hypothesesD}(k_2\cap q)\cdot(l_2\cap r)$. By Lemma~\ref{lem:intersection-distribution-left}, we have that $(k_2\cap q)\cdot(l_2\cap r)\leq_{\hypothesesD} (k_2\cdot l_2)\cap (q\cdot r)$, and hence $(k_1\cdot l_1)\cap (o\cdot p)\leq_{\hypothesesD} (k_2\cdot l_2)\cap (q\cdot r)$ as required.
\end{proof}

\begin{theorem}\label{thm:refinement-disjunction}
	Suppose $k_1,k_2,l_1$ and $l_2$ are KAT expressions. Let $\interfaceRel_k$ and $\interfaceRel_l$ be trace-refinement relations, such that $\abskrefines{k_1}{k_2}{\interfaceRel_k}$ and $\abskrefines{l_1}{l_2}{\interfaceRel_l}$. Then $\abskrefines{k_1 + l_1}{k_2+ l_2}{\interfaceRel_k\interfaceOpDisj \interfaceRel_l}$.
\end{theorem}

\begin{proof}
	We want to show that for any tuple $(x,y,\hypothesesD)\in\interfaceRel_k\interfaceOpDisj\interfaceRel_l$, $(k_1+k_2)\cap x\leq_{\hypothesesD}(l_1+l_2)\cap y$. Choose such an arbitrary tuple $(x,y,\hypothesesD)\in\interfaceRel_k\interfaceOpDisj\interfaceRel_l$, and let $(o,q,\hypotheses)\in\interfaceRel_k$ and $(p,r,\hypothesesB)\in\interfaceRel_l$ be the tuples that produced $(x,y,\hypothesesD)$. In other words, $x=o+ p$, $y=q+ r$ and $\hypothesesD=\hypotheses\cup \hypothesesB$.
	
	Since $\proj{1}{\interfaceRel_k}$ partitions $k_1$ and $\proj{1}{\interfaceRel_l}$ partitions $l_1$, we have that for any $o\in\proj{1}{\interfaceRel_k}$ and $p\in\proj{1}{\interfaceRel_l}$, $o\leq k$ and $p\leq l$, and thus by Lemma~\ref{lem:intersection-distribution-right}, we have that $(k_1+ l_1)\cap (o+ p)\leq (k_1\cap o)+ (l_1\cap p)$.
	By assumption, $k_1\cap o\leq_{\hypotheses} k_2\cap q$ and $l_1\cap p\leq_{\hypothesesB} l_2\cap r$. Since $\hypothesesD=\hypotheses\cup \hypothesesB$, by Lemma~\ref{lem:superset-hypotheses}, we have that $k_1\cap o\leq_{\hypothesesD} k_2\cap q$ and $l_1\cap p\leq_{\hypothesesD} l_2\cap r$. Therefore, by Lemma~\ref{lem:kat-less-product-disjunction}, we have that $(k_1\cap o)+(l_1\cap p)\leq_{\hypothesesD}(k_2\cap q)+(l_2\cap r)$. By Lemma~\ref{lem:intersection-distribution-left}, we have that $(k_2\cap q)+(l_2\cap r)\leq_{\hypothesesD} (k_2+ l_2)\cap (q+ r)$, and hence $(k_1+ l_1)\cap (o+ p)\leq_{\hypothesesD} (k_2+ l_2)\cap (q+ r)$ as required.
\end{proof}

\begin{theorem}\label{thm:refinement-disjunction-with-union}
	Suppose $k_1,k_2,l_1$ and $l_2$ are KAT expressions. Let $\interfaceRel_k$ and $\interfaceRel_l$ be trace-refinement relations, such that $\abskrefines{k_1}{k_2}{\interfaceRel_k}$ and $\abskrefines{l_1}{l_2}{\interfaceRel_l}$. Then $\abskrefines{k_1 + l_1}{k_2+ l_2}{\interfaceRel_k\cup \interfaceRel_l}$.
\end{theorem}

\begin{proof}
	Let $(x,y,\hypothesesD)$ be any tuple in $\interfaceRel_k\cup\interfaceRel_l$. Then $(x,y,\hypothesesD)\in\interfaceRel_k$ or $(x,y,\hypothesesD)\in\interfaceRel_l$. Since $\abskrefines{k_1}{k_2}{\interfaceRel_k}$ and $\abskrefines{l_1}{l_2}{\interfaceRel_l}$, it follows by definition that $k_1\cap x\leq_{\hypothesesD} k_2\cap y$ and $l_1\cap x\leq_{\hypothesesD} l_2\cap y$. Notice that if either $(x,y,\hypothesesD)\notin\interfaceRel_k$ or $(x,y,\hypothesesD)\notin\interfaceRel_l$, then, repsecitevly, either $k_1\cap x\equiv 0$ or $l_1\cap x\equiv 0$, and thus the above inequalities hold. Hence, by Lemmas~\ref{lem:intersection-distribution-right} and~\ref{lem:intersection-distribution-left}, $(k_1+l_1)\cap x\leq_{\hypothesesD}(k_2+l_2\cap y$ as required.
\end{proof}

\begin{theorem}\label{thm:refinement-star}
	Given any KAT expressions $k$ and $l$, and trace-refinement relation $\interfaceRel$ such that $\abskrefines{k}{l}{\interfaceRel}$, it holds that $\abskrefines{k^*}{l^*}{\interfaceOpStar{\interfaceRel}}$.
\end{theorem}

\begin{proof}
	We want to show that for any tuple $(x,y,\hypothesesD)\in\interfaceOpStar{\interfaceRel}$, $k^*\cap x\leq_{\hypothesesD} l^*\cap y$. Choose such an arbitrary tuple $(x,y,\hypothesesD)\in\interfaceOpStar{\interfaceRel}$, and let $(o,q,\hypotheses)\in\interfaceRel$ be the tuple that produced $(x,y,\hypothesesD)$. In other words, $x=o^*$, $y=q^*$ and $\hypothesesD=\hypotheses$.
	
	Since $\proj{1}{\interfaceRel}$ partitions $k_1$, we have that for any $o\in\proj{1}{\interfaceRel}$, $o\leq k$. By Lemma~\ref{lem:intersection-star-right} and the latter inequality, it follows that $k^*\cap o^*\leq_{\hypotheses} (k\cap o)^*$. Then, by the assumption that $\abskrefines{k}{l}{\interfaceRel}$, we have that $k\cap o\leq_{\hypotheses} l\cap q$, and thus $(k\cap o)^*\leq_{\hypotheses} (l\cap q)^*$. Furthermore, by Lemma~\ref{lem:intersection-star-left}, we have that $(l\cap q)^*\leq_{\hypotheses} l^*\cap q^*$. Together, these inequalities give us that $k^*\cap o^*\leq_{\hypotheses} l^*\cap q^*$, as required.	
\end{proof}

\begin{theoremApp}{thm:refinement-all}
	Suppose $k_1,k_2,l_1$ and $l_2$ are KAT expressions. Let $\interfaceRel_k$ and $\interfaceRel_l$ be trace-refinement relations, such that $\abskrefines{k_1}{k_2}{\interfaceRel_k}$ and $\abskrefines{l_1}{l_2}{\interfaceRel_l}$. Then
	\begin{itemize}
		\item $\abskrefines{k_1\cdot l_1}{k_2\cdot l_2}{\interfaceRel_k\interfaceOpComp \interfaceRel_l}$,
		\item $\abskrefines{k_1 + l_1}{k_2+ l_2}{\interfaceRel_k\interfaceOpDisj \interfaceRel_l}$,
		\item $\abskrefines{k_1 + l_1}{k_2+ l_2}{\interfaceRel_k\cup \interfaceRel_l}$, and
		\item $\abskrefines{k_1^*}{k_2^*}{\interfaceOpStar{\interfaceRel_k}}$.
	\end{itemize}
\end{theoremApp}

\begin{proof}
	It follows immediatelly from Theorems~\ref{thm:refinement-sequence}, \ref{thm:refinement-disjunction}, \ref{thm:refinement-disjunction-with-union} and~\ref{thm:refinement-star}.
\end{proof}

\begin{corollaryApp}{cor:same-context}
	Given any KAT expressions $m,l,k_1$ and $k_2$, and trace-refinement relation $\interfaceRel$ such that $\abskrefines{k_1}{k_2}{\interfaceRel}$, it holds that $\abskrefines{m\cdot k_1\cdot l}{m\cdot k_2\cdot l}{\interfaceRel'}$, where $\interfaceRel'$ is the set $\{(m\cdot r_1\cdot l,m\cdot r_2\cdot l,\hypotheses)\mid (r_1,r_2,\hypotheses)\in\interfaceRel\}$.
\end{corollaryApp}

\begin{proof}
	The result follows from Theorem~\ref{thm:refinement-all}, by noticing that $\abskrefines{m}{m}{\interfaceRel_{m}}$ and $\abskrefines{l}{l}{\interfaceRel_{l}}$, where $\interfaceRel_{m}$ and $\interfaceRel_{l}$ are the sets $\{(m,m,\emptyset)\}$ and $\{(l,l,\emptyset)\}$ respectively.
\end{proof}

\begin{theoremApp}{thm:transitivity}
	For any elements $k,l$ and $m$ in a KAT $\kat$, and any trace-refinement relations $\interfaceRel_1$, $\interfaceRel_2$, if $\abskrefines{k}{l}{\interfaceRel_1}$ and $\abskrefines{l}{m}{\interfaceRel_2}$, and $\interfaceRel_1\interfaceOpTrans\interfaceRel_2$ is defined, then $\abskrefines{k}{m}{\interfaceRel_1\interfaceOpTrans\interfaceRel_2}$.
\end{theoremApp}

\begin{proof}
	We want to show that for any $(o,p,\hypotheses)\in\interfaceRel_1\interfaceOpTrans\interfaceRel_2$, $k\cap o\leq_{\hypotheses} l\cap p$. Let $(o,r_1,\hypotheses_1)\in\interfaceRel_1$ and $(r_2,p,\hypotheses_2)\in\interfaceRel_2$, be the two tuples that produced the tuple $(o,p,\hypotheses)$ in their transitive trace-refinement relation. In other words, $r_1\leq r_2$ and $\hypotheses=\hypotheses_1\cup \hypotheses_2$. By assumption, we have that $k\cap o\leq_{\hypotheses_1} l\cap r_1$. Since $r_1\leq r_2$, $l\cap r_1\leq l\cap r_2$. Therefore, $k\cap o\leq_{\hypotheses_1} l\cap r_2$. Again by assumption, we have that $l\cap r_2\leq_{\hypotheses_2} m\cap p$. By Lemma~\ref{lem:superset-hypotheses}, we have that $k\cap o\leq_{\hypotheses} l\cap r_2$ and $l\cap r_2\leq_{\hypotheses} m\cap p$. Thus $k\cap o\leq_{\hypotheses} m\cap p$ as required.
\end{proof}

\subsection{Automation}\label{apx:omitted-automation}

\begin{theoremApp}{thm:sound}
  {\sc (Soundness).} For all $C_1,C_2$, and abstractions $\alpha$, let $O=\asearch(C_1,C_2,\emptyset,\alpha)$, let $\alpha'$ be the common abstraction of $O$ and let $k_1=\translate{C_1}{\alpha'}$ and $k_2=\translate{C_2}{\alpha'}$. Then $\absrefines{k_1}{k_2}{\interfaceFunc(O)}{}{}$.
\end{theoremApp}

\begin{proof}
  Let $\kat$ be a KAT. For a set of hypotheses $\hypotheses$ over $\kat$, two KAT expressions $k_1$ and $k_2$ and a trace-refinement relation $\interfaceRel$, we write $\abskhyporefines{k_1}{k_2}{\interfaceRel}{\hypotheses}$ to denote that $k_1$ refines $k_2$ with respect to $\interfaceRel$ by augmenting the set of hypotheses with $\hypotheses$. We proceed by induction on the number of recursive calls to show that for any abstraction $\alpha=(\kat,A_S,\alpha_S,\alpha_B)$, and any two programs $C_1$ and $C_2$, if $\interfaceRel=\interfaceFunc(\asearch(C_1,C_2,\hypotheses,\alpha))$, then $\abskhyporefines{\atranslate(C_1,\alpha)}{\atranslate(C_1,\alpha)}{\interfaceRel}{\hypotheses}$. Since the algorithm is initialised with $\hypotheses$ being the empty set, the trace-refinement relation $\interfaceRel$ returned will be such that $\abskrefines{k_1}{k_2}{\interfaceRel}$.

	  For the base case, suppose that the algorithm returns without any recursive calls. Then, for $k_1=\atranslate(C_1,\alpha)$ and $k_2=\atranslate(C_2,\alpha)$, the procedure $\akatdiff(k_1,k_2,\hypotheses)$ returns no counterexamples. By assumption, this means that $k_1\leq_{\hypotheses} k_2$, which implies that $\abskhyporefines{k_1}{k_2}{\interfaceRel}{\hypotheses}$, for $\interfaceRel=\{(k_1,k_2,\hypotheses,\alpha)\}$.
	  
	  For the inductive case, suppose that $\akatdiff(k_1,k_2,\hypotheses)$ returns a set of counterexamples $c=\{c_1,\ldots,c_m\}$. By assumption, the subprocedure $\asolvediff$, given $k_1, k_2, c$ and $\hypotheses$ as input, returns a set $R$ of restrictions, say of size $n\in\mathbb{N}$, such that $\proj{1}{R}$ partitions $k_1$.
	  Let $(r_1,r_2,\hypotheses')$ be a tuple in $R$, and let $(D_1,D_2,\alpha')$ be the output of $\arestrict(C_1,r_1,C_2,r_2,\hypotheses\cup\hypotheses',\alpha)$. By assumption, $\atranslate(D_1,\alpha')\equiv_{\hypotheses\cup\hypotheses'}\atranslate(C_1,\alpha')\cap r_1$, and the same holds for $D_2,C_2$ and $r_2$. By the inductive hypothesis, if $O$ is the output of $\asearch(D_1,D_2,\hypotheses\cup\hypotheses',\alpha')$, then $\abskhyporefines{\atranslate(D_1,\alpha')}{\atranslate(D_2,\alpha')}{\interfaceFunc(O)}{\hypotheses\cup\hypotheses'}$.
	  
	  For $i\leq n$, let $(r_{1,i},r_{2,i},\hypotheses_i)$ be the tuples in $R$ returned by the procedure $\asolvediff$. For each $i\leq n$, let $(D_{1,i},D_{2,i},\alpha'_i)$ be the result of $\arestrict(C_1,r_{1,i},C_2,r_{2,i},\hypotheses\cup\hypotheses_i,\alpha)$. Finally, let $O_i$ be the output of $\asearch(D_{1,i},D_{2,i},\hypotheses\cup\hypotheses_i)$ and $\interfaceRel_i'$ be equal to $\interfaceFunc(O_i)$. In other words, for each $i\leq n$, let $\interfaceRel_i'$ be the set $\{(k,l,\hypotheses)\mid (k,l,\hypotheses,\alpha'_i)\in O_i\}$, where $\alpha'_i$ is the common abstraction of $O_i$. Define $\beta$ to be the abstraction $\bigsqcup_{i\leq n}\alpha'_i$, and let $\interfaceRel_i$ be obtained from $\interfaceRel_i'$ by having all KAT expressions be over the common abstraction $\beta$. Then define $O$ to be equal to $O_1\cup\ldots\cup O_n$. Notice that the \textsf{flatmap} operator in the algorithm, simply returns $O$ from all the $O_i$, and notice that $\interfaceRel_1\cup\ldots\cup\interfaceRel_n=\interfaceFunc(O_1\cup\ldots\cup O_n)$. By the argument above, we have that for all $i\leq n$,
	  \begin{equation}\label{eq:indhyp}
		  \abskhyporefines{\atranslate(D_{1,i},\alpha'_i)}{\atranslate(D_{2,i},\alpha'_i)}{\interfaceFunc(O_i)}{\hypotheses\cup\hypotheses'_i},
	  \end{equation}
	  and $\atranslate(D_{1,i},\alpha'_i)\equiv_{\hypotheses\cup\hypotheses'}\atranslate(C_1,\alpha')\cap r_{1,i}$. Notice that since $\proj{1}{R}$ partitions $k_1$,
	  \begin{displaymath}
	  	\begin{array}{rl}
			\atranslate(C_1,\alpha')\cap r_{1,1}+\ldots+\atranslate(C_1,\alpha')\cap r_{1,n}&\equiv_{\hypotheses\cup\hypotheses'}\atranslate(C_1,\alpha')\\
			&\hspace{30pt}\cap(r_{1,1}+\ldots+r_{1,n})\\
			&\equiv_{\hypotheses\cup\hypotheses'}\atranslate(C_1,\alpha'),
	  	\end{array}
	  \end{displaymath}
	  and therefore
	  \begin{equation}\label{eq:union-1}
		  \atranslate(D_{1,1},\alpha'_1)+\ldots+\atranslate(D_{1,n},\alpha'_n)\equiv_{\hypotheses\cup\hypotheses'}\atranslate(C_1,\alpha').
	  \end{equation}
	  By a similar argument,
	  \begin{equation}\label{eq:union-2}
		  \atranslate(D_{2,1},\alpha'_1)+\ldots+\atranslate(D_{2,n},\alpha'_n)\equiv_{\hypotheses\cup\hypotheses'}\atranslate(C_2,\alpha').
	  \end{equation}
	  Therefore, by Theorem~\ref{thm:refinement-disjunction-with-union} and equations~(\ref{eq:indhyp}),~(\ref{eq:union-1}) and~(\ref{eq:union-2}),
	  $$\abskhyporefines{\atranslate(C_1,\alpha')}{\atranslate(C_2,\alpha')}{\interfaceRel_1\cup\ldots\cup\interfaceRel_n}{\hypotheses\cup\hypotheses_1\cup\ldots\hypotheses_n}{},$$
	  where, as was argued earlier, $\interfaceRel_1\cup\ldots\cup\interfaceRel_n=\interfaceFunc(O_1\cup\ldots\cup O_n)=\interfaceFunc(O)$.
%
%
\end{proof}

\vfill
\pagebreak
\section{Edit Distance Algorithm}
\label{apx:edit-distance-alg}

\begin{figure}[b]
\begin{program}[style=sf]
{\bf Input:} Two strings $s_1,s_2$.
{\bf Output:} A set $T$ of transformations and a total score for that set.
{\bf Algori}\tab{\bf thm:} {\sc Distance}($s_1$,$s_2$)
  if \tab($s_1=[]$ and $s_2=[]$)
  return $([],\mathsf{match\_scr})$
  \untab
  els\tab e if ($s_1=[]$)
  	return $(\mathsf{RemoveAll(s_2)},\mathsf{len}(s_2)*\mathsf{remove\_scr})$
  \untab
  els\tab e if ($s_2=[]$)
    return $(\mathsf{RemoveAll(s_1)},\mathsf{len}(s_1                                                                                                                                                   )*\mathsf{remove\_scr})$
  \untab
  els\tab e
	 $s_1=h_1:::t_1$ and $s_2=h_2:::t_2$
	 $(T_1,S_1)=${\sc Distance}$(t_1,s_2)$, $g_1=\mathsf{Remove(h_1,s_1)}$, $o_1=\mathsf{remove\_scr}$
	 $(T_2,S_2)=${\sc Distance}$(s_1,t_2)$, $g_2=\mathsf{Remove(h_2,s_2)}$, $o_2=\mathsf{remove\_scr}$
	 $(T_3,S_3)=${\sc Distance}$(t_1,t_2)$
	 if \tab($\mathsf{same\_symbol}(h_1,h_2)$)
	 	$g_3=\mathsf{Match(h_1,h_2,s_1,s_2)}$, $o_3=\mathsf{match\_scr}$
	\untab
	els\tab e
	  $g_3=\mathsf{Replace(h_1,h_2,s_1)}$
	  if \tab ($\mathsf{same\_type}(h_1,h_2)$)
	    $o_3=\mathsf{replace\_scr}$
	  \untab
	  els\tab e
		$o_3=\mathsf{penalty\_scr}$
	  \untab
	\untab
	 for \tab minimum $S_i$:
	  return $(g_i:::T_i,S_i+o_i)$
  \untab
\untab
   \end{program}
   \caption{\label{fig:cex-distance-alg} The distance algorithm for two counterexample strings.}
	 \label{fig:cex-distance-alg}
\end{figure}

The algorithm, shown in Figure~\ref{fig:cex-distance-alg}, traverses recursively the two iputs one symbol at a time, with the option of staying stationary on one of them at each iteration, and assigns a score on the association between the symbols at hand. For this, 4 different types of scores (in the form of rationals), are calcualted for any two strings, and are added to the total score at each iteration, depending on the action that is chosen. All possible cases are considered by the algorithm, and the association that leads to the smallest global score is finally chosen. The 4 different types of scores are as follows.
\begin{itemize}
	\item \textsf{remove\_scr}: Used when a symbol is removed from one of the two strings.
	\item \textsf{replace\_scr}: Used when a symbol is replaced with another symbol in one of the two strings.
	\item \textsf{match\_scr}: Used when a symbol in one string is matched with a symbol of the same type (boolean or event) in the other string.
	\item \textsf{penalty\_scr}: Used when a matching such as the one above is chosen, but where the matching is between symbols of different type.
\end{itemize}

The values for \textsf{remove\_scr} and \textsf{replace\_scr} are usually 1, whereas the \textsf{penalty\_scr} is higher that them, and correlated with the length of the input strings. The value of \textsf{match\_scr} on the other hand is negative, and used to counter-balance the effect of \textsf{penalty\_scr}. In the algorithm shown above, $\mathsf{RemoveAll(s_2)}$, for a string $s_2=a_1\cdots a_n$, is shorthand for the sequence: \\
$[\mathsf{Remove(a_1,s_2),\ldots,Remove(a_n,s_2)}]$.

\subsection{Global KAT expression edit-distance}\label{sec:global-edit-distance}

We have implemented a custom edit-distance algorithm that accepts general KAT expressions as inputs, instead of merely KAT strings. The edit-distance on such KAT expressions has to handle the structure of the expressions, and is naturally more involved than the linear one on strings. (The former is more similar to trees~\cite{treeedit}.) For example, the algorithm will attempt and match a subexpression under a star operation in one expression with a similar subexpression under a star operation in the other. In our experiments, using this distance algorithm on the whole KAT expressions, instead of the counterexamples to their equivalence or inclusion, would most of the time remove and replace many symbols. Our implementation mostly does not use this facility. However, searching for edit distance globally on the KAT expressions can be exploited in the beginning of the algorithm, in order to find natural alignments between two large programs, split them into subcomponents, apply the {\asearch} algorithm on each pair of such subcomponents, and finally use Theorem~\ref{thm:refinement-all} to combine the individual results into a solution that works over the whole programs. Our use of global KAT edit distance does not require further theoretical development and we plan to use our implementation of these ideas in future work.

\newpage
\section{Benchmarks and Full Results of {\knotical}}
\label{apx:fullresults}
\subsection{Example synthesized solution for benchmark \texttt{00arith.c}}
\begin{scriptsize}
\dirtree{%
.1 solution.
.2 AComplete.
.3 $\left\{\begin{array}{l}Axioms: \{E=1, b=!c\}\\ k_{1}=(a_{5}\!\cdot\!(b_{8}\!\cdot\!() = foo(); + \neg b_{8}\!\cdot\!() = bar();))*\!\cdot\!\neg a_{5}\\ k_{2}=() = foo();\!\cdot\!(a_{12}\!\cdot\!(c_{19}\!\cdot\!() = bar(); + \neg c_{19}\!\cdot\!() = foo();))*\!\cdot\!\neg a_{12}\end{array}\right.$.
}

\end{scriptsize}
\subsection{Example synthesized solution for benchmark \texttt{00complete.c}}
\begin{scriptsize}
\dirtree{%
.1 solution.
.2 (Complete), cond: $n > 0$.
.3 $\left\{\begin{array}{l}Cond: a_{10} \\ k_{1}=() = evA();\\ k_{2}=(a_{10}\!\cdot\!() = evB(); + \neg a_{10}\!\cdot\!() = evC();)\end{array}\right.$.
.4 AComplete.
.5 $\left\{\begin{array}{l}Axioms: \{E=V\}\\ k_{1}=() = evA();\\ k_{2}=1\!\cdot\!() = evB();\end{array}\right.$.
.3 $\left\{\begin{array}{l}Cond: \neg a_{10} \\ k_{1}=() = evA();\\ k_{2}=(a_{10}\!\cdot\!() = evB(); + \neg a_{10}\!\cdot\!() = evC();)\end{array}\right.$.
.4 AComplete.
.5 $\left\{\begin{array}{l}Axioms: \{E=A\}\\ k_{1}=() = evA();\\ k_{2}=1\!\cdot\!() = evC();\end{array}\right.$.
}

\end{scriptsize}
\subsection{Example synthesized solution for benchmark \texttt{00complete1.c}}
\begin{scriptsize}
\dirtree{%
.1 solution.
.2 AComplete.
.3 $\left\{\begin{array}{l}Axioms: \{!a=!b, V=D, E=A\}\\ k_{1}=(a_{7}\!\cdot\!() = evA(); + \neg a_{7}\!\cdot\!() = evB();)\\ k_{2}=(b_{13}\!\cdot\!() = evC(); + \neg b_{13}\!\cdot\!() = evD();)\end{array}\right.$.
}

\end{scriptsize}
\subsection{Example synthesized solution for benchmark \texttt{00false.c}}
\begin{scriptsize}
\dirtree{%
.1 solution.
.2 AComplete.
.3 $\left\{\begin{array}{l}Axioms: \{V=1\}\\ k_{1}=1\\ k_{2}=() = eventA();\end{array}\right.$.
}

\end{scriptsize}
\subsection{Example synthesized solution for benchmark \texttt{00if.c}}
\begin{scriptsize}
\dirtree{%
.1 solution.
.2 AComplete.
.3 $\left\{\begin{array}{l}Axioms: \{V=1, B=1\}\\ k_{1}=a = nondet();\!\cdot\!(a_{7}\!\cdot\!() = send();\!\cdot\!() = recv(); + \neg a_{7}\!\cdot\!1)\\ k_{2}=a = nondet();\!\cdot\!(a_{12}\!\cdot\!() = send(); + \neg a_{12}\!\cdot\!1)\!\cdot\!() = recv();\end{array}\right.$.
}

\end{scriptsize}
\subsection{Example synthesized solution for benchmark \texttt{00ifarecv.c}}
\begin{scriptsize}
\dirtree{%
.1 solution.
.2 AComplete.
.3 $\left\{\begin{array}{l}Axioms: \{N=1, D=1\}\\ k_{1}=() = init();\!\cdot\!a = recv();\!\cdot\!(a_{6}\!\cdot\!() = send(); + \neg a_{6}\!\cdot\!1)\\ k_{2}=() = init();\!\cdot\!a = recv();\!\cdot\!() = send();\end{array}\right.$.
}

\end{scriptsize}
\subsection{Example synthesized solution for benchmark \texttt{00impos.c}}
\begin{scriptsize}

\end{scriptsize}
\medskip
\emph{No solutions.}
\subsection{Example synthesized solution for benchmark \texttt{00medstrai.c}}
\begin{scriptsize}
\dirtree{%
.1 solution.
.2 (Partial), cond: $N > 0$.
.3 $\left\{\begin{array}{l}Cond: \neg b_{14} \\ k_{1}=(b_{14}\!\cdot\!() = m1();\!\cdot\!() = m2(); + \neg b_{14}\!\cdot\!1)\!\cdot\!() = m4();\!\cdot\!() = m5();\!\cdot\!(a_{9}\!\cdot\!() = m11();\!\cdot\!() = m12(); + \neg a_{9}\!\cdot\!1)\!\cdot\!() = m14();\!\cdot\!() = m15();\\ k_{2}=() = m1();\!\cdot\!() = m2();\!\cdot\!() = m4();\!\cdot\!() = m5();\!\cdot\!() = m11();\!\cdot\!() = m12();\!\cdot\!() = m14();\!\cdot\!() = m15();\end{array}\right.$.
.4 (Partial), cond: $N < 0$.
.5 $\left\{\begin{array}{l}Cond: a_{9} \\ k_{1}=1\!\cdot\!1\!\cdot\!() = m4();\!\cdot\!() = m5();\!\cdot\!(a_{9}\!\cdot\!() = m11();\!\cdot\!() = m12(); + \neg a_{9}\!\cdot\!1)\!\cdot\!() = m14();\!\cdot\!() = m15();\\ k_{2}=() = m1();\!\cdot\!() = m2();\!\cdot\!() = m4();\!\cdot\!() = m5();\!\cdot\!() = m11();\!\cdot\!() = m12();\!\cdot\!() = m14();\!\cdot\!() = m15();\end{array}\right.$.
.6 AComplete.
.7 $\left\{\begin{array}{l}Axioms: \{L=1, P=1\}\\ k_{1}=1\!\cdot\!1\!\cdot\!() = m4();\!\cdot\!() = m5();\!\cdot\!1\!\cdot\!() = m11();\!\cdot\!() = m12();\!\cdot\!() = m14();\!\cdot\!() = m15();\\ k_{2}=() = m1();\!\cdot\!() = m2();\!\cdot\!() = m4();\!\cdot\!() = m5();\!\cdot\!() = m11();\!\cdot\!() = m12();\!\cdot\!() = m14();\!\cdot\!() = m15();\end{array}\right.$.
}

\end{scriptsize}
\medskip
\emph{Remaining 1 solutions ommitted for brevity.}
\subsection{Example synthesized solution for benchmark \texttt{00needax.c}}
\begin{scriptsize}
\dirtree{%
.1 solution.
.2 AComplete.
.3 $\left\{\begin{array}{l}Axioms: \{V=E\}\\ k_{1}=(a_{7}\!\cdot\!() = evA(); + \neg a_{7}\!\cdot\!() = evB();)\\ k_{2}=() = evA();\end{array}\right.$.
}

\end{scriptsize}
\subsection{Example synthesized solution for benchmark \texttt{00nohyp.c}}
\begin{scriptsize}
\dirtree{%
.1 solution.
.2 AComplete.
.3 $\left\{\begin{array}{l}Axioms: \{D=1, !a=b\}\\ k_{1}=(a_{7}\!\cdot\!() = evA(); + \neg a_{7}\!\cdot\!() = evB();)\\ k_{2}=t = nondet();\!\cdot\!(b_{12}\!\cdot\!() = evB(); + \neg b_{12}\!\cdot\!() = evA();)\end{array}\right.$.
}

\end{scriptsize}
\subsection{Example synthesized solution for benchmark \texttt{00noloop.c}}
\begin{scriptsize}
\dirtree{%
.1 solution.
.2 (Partial), cond: $count <= 4$.
.3 $\left\{\begin{array}{l}Cond: \neg a_{9} \\ k_{1}=1\!\cdot\!(a_{5}\!\cdot\!0)*\!\cdot\!\neg a_{5}\\ k_{2}=(a_{9}\!\cdot\!() = evA();\!\cdot\!count = count + 1;)*\!\cdot\!\neg a_{9}\end{array}\right.$.
.4 AComplete.
.5 $\left\{\begin{array}{l}Axioms: \{\}\\ k_{1}=1\!\cdot\!(a_{5}\!\cdot\!0)*\!\cdot\!\neg a_{5}\\ k_{2}=1\!\cdot\!(a_{9}\!\cdot\!0)*\!\cdot\!\neg a_{9}\end{array}\right.$.
}

\end{scriptsize}
\subsection{Example synthesized solution for benchmark \texttt{00nondet.c}}
\begin{scriptsize}
\dirtree{%
.1 solution.
.2 (Complete), cond: $a > 0$.
.3 $\left\{\begin{array}{l}Cond: d_{15} \\ k_{1}=i = nondet();\!\cdot\!j = nondet();\!\cdot\!(c_{11}\!\cdot\!(d_{15}\!\cdot\!() = B(); + \neg d_{15}\!\cdot\!() = C();)\!\cdot\!i = nondet();)*\!\cdot\!\neg c_{11}\!\cdot\!() = D();\\ \;\;\;\!\cdot\!(a_{5}\!\cdot\!(b_{9}\!\cdot\!() = G(); + \neg b_{9}\!\cdot\!() = H();)\!\cdot\!j = nondet();)*\!\cdot\!\neg a_{5}\\ k_{2}=i = nondet();\!\cdot\!j = nondet();\!\cdot\!(c_{29}\!\cdot\!(f_{33}\!\cdot\!() = B(); + \neg f_{33}\!\cdot\!() = C();)\!\cdot\!i = nondet();)*\!\cdot\!\neg c_{29}\!\cdot\!() = D();\\ \;\;\;\!\cdot\!(a_{23}\!\cdot\!(e_{27}\!\cdot\!() = G(); + \neg e_{27}\!\cdot\!() = H();)\!\cdot\!j = nondet();)*\!\cdot\!\neg a_{23}\end{array}\right.$.
.4 AComplete.
.5 $\left\{\begin{array}{l}Axioms: \{b=e, B=C\}\\ k_{1}=i = nondet();\!\cdot\!j = nondet();\!\cdot\!(c_{11}\!\cdot\!1\!\cdot\!() = B();\!\cdot\!i = nondet();)*\!\cdot\!\neg c_{11}\!\cdot\!() = D();\\ \;\;\;\!\cdot\!(a_{5}\!\cdot\!(b_{9}\!\cdot\!() = G(); + \neg b_{9}\!\cdot\!() = H();)\!\cdot\!j = nondet();)*\!\cdot\!\neg a_{5}\\ k_{2}=i = nondet();\!\cdot\!j = nondet();\!\cdot\!(c_{29}\!\cdot\!(f_{33}\!\cdot\!() = B(); + \neg f_{33}\!\cdot\!() = C();)\!\cdot\!i = nondet();)*\!\cdot\!\neg c_{29}\!\cdot\!() = D();\\ \;\;\;\!\cdot\!(a_{23}\!\cdot\!(e_{27}\!\cdot\!() = G(); + \neg e_{27}\!\cdot\!() = H();)\!\cdot\!j = nondet();)*\!\cdot\!\neg a_{23}\end{array}\right.$.
.3 $\left\{\begin{array}{l}Cond: \neg d_{15} \\ k_{1}=i = nondet();\!\cdot\!j = nondet();\!\cdot\!(c_{11}\!\cdot\!(d_{15}\!\cdot\!() = B(); + \neg d_{15}\!\cdot\!() = C();)\!\cdot\!i = nondet();)*\!\cdot\!\neg c_{11}\!\cdot\!() = D();\\ \;\;\;\!\cdot\!(a_{5}\!\cdot\!(b_{9}\!\cdot\!() = G(); + \neg b_{9}\!\cdot\!() = H();)\!\cdot\!j = nondet();)*\!\cdot\!\neg a_{5}\\ k_{2}=i = nondet();\!\cdot\!j = nondet();\!\cdot\!(c_{29}\!\cdot\!(f_{33}\!\cdot\!() = B(); + \neg f_{33}\!\cdot\!() = C();)\!\cdot\!i = nondet();)*\!\cdot\!\neg c_{29}\!\cdot\!() = D();\\ \;\;\;\!\cdot\!(a_{23}\!\cdot\!(e_{27}\!\cdot\!() = G(); + \neg e_{27}\!\cdot\!() = H();)\!\cdot\!j = nondet();)*\!\cdot\!\neg a_{23}\end{array}\right.$.
.4 AComplete.
.5 $\left\{\begin{array}{l}Axioms: \{b=e, C=B\}\\ k_{1}=i = nondet();\!\cdot\!j = nondet();\!\cdot\!(c_{11}\!\cdot\!1\!\cdot\!() = C();\!\cdot\!i = nondet();)*\!\cdot\!\neg c_{11}\!\cdot\!() = D();\\ \;\;\;\!\cdot\!(a_{5}\!\cdot\!(b_{9}\!\cdot\!() = G(); + \neg b_{9}\!\cdot\!() = H();)\!\cdot\!j = nondet();)*\!\cdot\!\neg a_{5}\\ k_{2}=i = nondet();\!\cdot\!j = nondet();\!\cdot\!(c_{29}\!\cdot\!(f_{33}\!\cdot\!() = B(); + \neg f_{33}\!\cdot\!() = C();)\!\cdot\!i = nondet();)*\!\cdot\!\neg c_{29}\!\cdot\!() = D();\\ \;\;\;\!\cdot\!(a_{23}\!\cdot\!(e_{27}\!\cdot\!() = G(); + \neg e_{27}\!\cdot\!() = H();)\!\cdot\!j = nondet();)*\!\cdot\!\neg a_{23}\end{array}\right.$.
}

\end{scriptsize}
\subsection{Example synthesized solution for benchmark \texttt{00pos.c}}
\begin{scriptsize}
\dirtree{%
.1 solution.
.2 (Partial), cond: $x > 0$.
.3 $\left\{\begin{array}{l}Cond: a_{7} \\ k_{1}=(a_{7}\!\cdot\!() = evA(); + \neg a_{7}\!\cdot\!() = evB();)\\ k_{2}=() = evA();\end{array}\right.$.
.4 AComplete.
.5 $\left\{\begin{array}{l}Axioms: \{\}\\ k_{1}=1\!\cdot\!() = evA();\\ k_{2}=() = evA();\end{array}\right.$.
}

\end{scriptsize}
\subsection{Example synthesized solution for benchmark \texttt{00rename.c}}
\begin{scriptsize}
\dirtree{%
.1 solution.
.2 AComplete.
.3 $\left\{\begin{array}{l}Axioms: \{a=b\}\\ k_{1}=(a_{7}\!\cdot\!() = foo(); + \neg a_{7}\!\cdot\!() = bar();)\\ k_{2}=(b_{13}\!\cdot\!() = foo(); + \neg b_{13}\!\cdot\!() = bar();)\end{array}\right.$.
}

\end{scriptsize}
\subsection{Example synthesized solution for benchmark \texttt{00rename1.c}}
\begin{scriptsize}
\dirtree{%
.1 solution.
.2 AComplete.
.3 $\left\{\begin{array}{l}Axioms: \{a=!b\}\\ k_{1}=(a_{7}\!\cdot\!() = foo(); + \neg a_{7}\!\cdot\!() = bar();)\\ k_{2}=(b_{13}\!\cdot\!() = bar(); + \neg b_{13}\!\cdot\!() = foo();)\end{array}\right.$.
}

\end{scriptsize}
\subsection{Example synthesized solution for benchmark \texttt{00sanity.c}}
\begin{scriptsize}
\dirtree{%
.1 solution.
.2 AComplete.
.3 $\left\{\begin{array}{l}Axioms: \{\}\\ k_{1}=() = foo();\!\cdot\!() = bar();\\ k_{2}=() = foo();\!\cdot\!() = bar();\end{array}\right.$.
}

\end{scriptsize}
\subsection{Example synthesized solution for benchmark \texttt{00sanity1.c}}
\begin{scriptsize}
\dirtree{%
.1 solution.
.2 AComplete.
.3 $\left\{\begin{array}{l}Axioms: \{O=B\}\\ k_{1}=() = foo();\!\cdot\!() = bar();\\ k_{2}=() = bar();\!\cdot\!() = foo();\end{array}\right.$.
}

\end{scriptsize}
\subsection{Example synthesized solution for benchmark \texttt{00smstrai.c}}
\begin{scriptsize}
\dirtree{%
.1 solution.
.2 (Partial), cond: $N > 0$.
.3 $\left\{\begin{array}{l}Cond: b_{10} \\ k_{1}=(b_{10}\!\cdot\!() = m1(); + \neg b_{10}\!\cdot\!1)\!\cdot\!() = m4();\!\cdot\!(a_{7}\!\cdot\!() = m11(); + \neg a_{7}\!\cdot\!1)\!\cdot\!() = m14();\\ k_{2}=() = m1();\!\cdot\!() = m4();\!\cdot\!() = m11();\!\cdot\!() = m14();\end{array}\right.$.
.4 AComplete.
.5 $\left\{\begin{array}{l}Axioms: \{G=1\}\\ k_{1}=1\!\cdot\!() = m1();\!\cdot\!() = m4();\!\cdot\!1\!\cdot\!() = m14();\\ k_{2}=() = m1();\!\cdot\!() = m4();\!\cdot\!() = m11();\!\cdot\!() = m14();\end{array}\right.$.
}

\end{scriptsize}
\medskip
\emph{Remaining 3 solutions ommitted for brevity.}
\subsection{Example synthesized solution for benchmark \texttt{01acqrel.c}}
\begin{scriptsize}
\dirtree{%
.1 solution.
.2 AComplete.
.3 $\left\{\begin{array}{l}Axioms: \{\}\\ k_{1}=A = 0;\!\cdot\!R = 0;\!\cdot\!0\!\cdot\!0\\ k_{2}=AA = 0;\!\cdot\!RR = 0;\!\cdot\!(a_{8}\!\cdot\!AA = 1;\!\cdot\!AA = 0;\!\cdot\!(a_{11}\!\cdot\!n = n - 1;)*\!\cdot\!\neg a_{11}\!\cdot\!RR = 1;\!\cdot\!RR = 0;)*\!\cdot\!\neg a_{8}\!\cdot\!0\end{array}\right.$.
}

\end{scriptsize}
\subsection{Example synthesized solution for benchmark \texttt{01asendrecv.c}}
\begin{scriptsize}
\dirtree{%
.1 solution.
.2 (Partial), cond: $b > 0$.
.3 $\left\{\begin{array}{l}Cond: b_{11} \\ k_{1}=(a_{7}\!\cdot\!b = recv();\!\cdot\!1\!\cdot\!1\!\cdot\!(b_{11}\!\cdot\!n = constructReply();\!\cdot\!() = send(n);\\ \;\;\;\!\cdot\!1 + \neg b_{11}\!\cdot\!1)\!\cdot\!x = x - 1;)*\!\cdot\!\neg a_{7}\\ k_{2}=(a_{19}\!\cdot\!b = recv();\!\cdot\!(b_{26}\!\cdot\!auth = check(b);\!\cdot\!(c_{23}\!\cdot\!n = constructReply();\!\cdot\!() = sendA(n);\\ \;\;\; + \neg c_{23}\!\cdot\!1) + \neg b_{26}\!\cdot\!() = log(b);)\!\cdot\!x = x - 1;)*\!\cdot\!\neg a_{19}\end{array}\right.$.
.4 (Partial), cond: $b > 0$.
.5 $\left\{\begin{array}{l}Cond: b_{26} \\ k_{1}=(a_{7}\!\cdot\!b = recv();\!\cdot\!1\!\cdot\!1\!\cdot\!1\!\cdot\!n = constructReply();\!\cdot\!() = send(n);\\ \;\;\;\!\cdot\!1\!\cdot\!x = x - 1;)*\!\cdot\!\neg a_{7}\\ k_{2}=(a_{19}\!\cdot\!b = recv();\!\cdot\!(b_{26}\!\cdot\!auth = check(b);\!\cdot\!(c_{23}\!\cdot\!n = constructReply();\!\cdot\!() = sendA(n);\\ \;\;\; + \neg c_{23}\!\cdot\!1) + \neg b_{26}\!\cdot\!() = log(b);)\!\cdot\!x = x - 1;)*\!\cdot\!\neg a_{19}\end{array}\right.$.
.6 (Partial), cond: $auth > 0$.
.7 $\left\{\begin{array}{l}Cond: c_{23} \\ k_{1}=(a_{7}\!\cdot\!b = recv();\!\cdot\!1\!\cdot\!1\!\cdot\!1\!\cdot\!n = constructReply();\!\cdot\!() = send(n);\\ \;\;\;\!\cdot\!1\!\cdot\!x = x - 1;)*\!\cdot\!\neg a_{7}\\ k_{2}=(a_{19}\!\cdot\!b = recv();\!\cdot\!1\!\cdot\!auth = check(b);\!\cdot\!(c_{23}\!\cdot\!n = constructReply();\!\cdot\!() = sendA(n);\\ \;\;\; + \neg c_{23}\!\cdot\!1)\!\cdot\!x = x - 1;)*\!\cdot\!\neg a_{19}\end{array}\right.$.
.8 AComplete.
.9 $\left\{\begin{array}{l}Axioms: \{I=1, J=1, M=1, P=1\}\\ k_{1}=(a_{7}\!\cdot\!b = recv();\!\cdot\!1\!\cdot\!1\!\cdot\!1\!\cdot\!n = constructReply();\!\cdot\!() = send(n);\\ \;\;\;\!\cdot\!1\!\cdot\!x = x - 1;)*\!\cdot\!\neg a_{7}\\ k_{2}=(a_{19}\!\cdot\!b = recv();\!\cdot\!1\!\cdot\!auth = check(b);\!\cdot\!1\!\cdot\!n = constructReply();\!\cdot\!() = sendA(n);\\ \;\;\;\!\cdot\!x = x - 1;)*\!\cdot\!\neg a_{19}\end{array}\right.$.
.3 $\left\{\begin{array}{l}Cond: \neg b_{11} \\ k_{1}=(a_{7}\!\cdot\!b = recv();\!\cdot\!1\!\cdot\!1\!\cdot\!(b_{11}\!\cdot\!n = constructReply();\!\cdot\!() = send(n);\\ \;\;\;\!\cdot\!1 + \neg b_{11}\!\cdot\!1)\!\cdot\!x = x - 1;)*\!\cdot\!\neg a_{7}\\ k_{2}=(a_{19}\!\cdot\!b = recv();\!\cdot\!(b_{26}\!\cdot\!auth = check(b);\!\cdot\!(c_{23}\!\cdot\!n = constructReply();\!\cdot\!() = sendA(n);\\ \;\;\; + \neg c_{23}\!\cdot\!1) + \neg b_{26}\!\cdot\!() = log(b);)\!\cdot\!x = x - 1;)*\!\cdot\!\neg a_{19}\end{array}\right.$.
.4 (Partial), cond: $auth > 0$.
.5 $\left\{\begin{array}{l}Cond: c_{23} \\ k_{1}=(a_{7}\!\cdot\!b = recv();\!\cdot\!1\!\cdot\!1\!\cdot\!1\!\cdot\!1\!\cdot\!x = x - 1;)*\!\cdot\!\neg a_{7}\\ k_{2}=(a_{19}\!\cdot\!b = recv();\!\cdot\!(b_{26}\!\cdot\!auth = check(b);\!\cdot\!(c_{23}\!\cdot\!n = constructReply();\!\cdot\!() = sendA(n);\\ \;\;\; + \neg c_{23}\!\cdot\!1) + \neg b_{26}\!\cdot\!() = log(b);)\!\cdot\!x = x - 1;)*\!\cdot\!\neg a_{19}\end{array}\right.$.
.6 (Partial), cond: $b > 0$.
.7 $\left\{\begin{array}{l}Cond: \neg b_{26} \\ k_{1}=(a_{7}\!\cdot\!b = recv();\!\cdot\!1\!\cdot\!1\!\cdot\!1\!\cdot\!1\!\cdot\!x = x - 1;)*\!\cdot\!\neg a_{7}\\ k_{2}=(a_{19}\!\cdot\!b = recv();\!\cdot\!(b_{26}\!\cdot\!auth = check(b);\!\cdot\!1\!\cdot\!n = constructReply();\!\cdot\!() = sendA(n);\\ \;\;\; + \neg b_{26}\!\cdot\!() = log(b);)\!\cdot\!x = x - 1;)*\!\cdot\!\neg a_{19}\end{array}\right.$.
.8 AComplete.
.9 $\left\{\begin{array}{l}Axioms: \{I=1, J=1\}\\ k_{1}=(a_{7}\!\cdot\!b = recv();\!\cdot\!1\!\cdot\!1\!\cdot\!1\!\cdot\!1\!\cdot\!x = x - 1;)*\!\cdot\!\neg a_{7}\\ k_{2}=(a_{19}\!\cdot\!b = recv();\!\cdot\!1\!\cdot\!() = log(b);\!\cdot\!x = x - 1;)*\!\cdot\!\neg a_{19}\end{array}\right.$.
.5 $\left\{\begin{array}{l}Cond: \neg c_{23} \\ k_{1}=(a_{7}\!\cdot\!b = recv();\!\cdot\!1\!\cdot\!1\!\cdot\!1\!\cdot\!1\!\cdot\!x = x - 1;)*\!\cdot\!\neg a_{7}\\ k_{2}=(a_{19}\!\cdot\!b = recv();\!\cdot\!(b_{26}\!\cdot\!auth = check(b);\!\cdot\!(c_{23}\!\cdot\!n = constructReply();\!\cdot\!() = sendA(n);\\ \;\;\; + \neg c_{23}\!\cdot\!1) + \neg b_{26}\!\cdot\!() = log(b);)\!\cdot\!x = x - 1;)*\!\cdot\!\neg a_{19}\end{array}\right.$.
.6 AComplete.
.7 $\left\{\begin{array}{l}Axioms: \{I=1, J=1\}\\ k_{1}=(a_{7}\!\cdot\!b = recv();\!\cdot\!1\!\cdot\!1\!\cdot\!1\!\cdot\!1\!\cdot\!x = x - 1;)*\!\cdot\!\neg a_{7}\\ k_{2}=(a_{19}\!\cdot\!b = recv();\!\cdot\!(b_{26}\!\cdot\!auth = check(b);\!\cdot\!1\!\cdot\!1 + \neg b_{26}\!\cdot\!() = log(b);)\!\cdot\!x = x - 1;)*\!\cdot\!\neg a_{19}\end{array}\right.$.
}

\end{scriptsize}
\medskip
\emph{Remaining 37 solutions ommitted for brevity.}
\subsection{Example synthesized solution for benchmark \texttt{01assume.c}}
\begin{scriptsize}
\dirtree{%
.1 solution.
.2 (Complete), cond: $count <= 4$.
.3 $\left\{\begin{array}{l}Cond: a_{5} \\ k_{1}=count = nondet();\!\cdot\!(a_{5}\!\cdot\!() = printf(count);\\ \;\;\;\!\cdot\!count = count + 1;)*\!\cdot\!\neg a_{5}\\ k_{2}=count = nondet();\!\cdot\!1\!\cdot\!((a_{11} \wedge b_{12})\!\cdot\!() = printf(count);\\ \;\;\;\!\cdot\!count = count + 1;)*\!\cdot\!\neg a_{11}\end{array}\right.$.
.4 (Complete), cond: $number >= 0$.
.5 $\left\{\begin{array}{l}Cond: b_{12} \\ k_{1}=count = nondet();\!\cdot\!1\!\cdot\!(a_{5}\!\cdot\!() = printf(count);\\ \;\;\;\!\cdot\!count = count + 1;)*\!\cdot\!\neg a_{5}\\ k_{2}=count = nondet();\!\cdot\!1\!\cdot\!((a_{11} \wedge b_{12})\!\cdot\!() = printf(count);\\ \;\;\;\!\cdot\!count = count + 1;)*\!\cdot\!\neg a_{11}\end{array}\right.$.
.6 AComplete.
.7 $\left\{\begin{array}{l}Axioms: \{D=1, E=1, I=1, T=1, U=1\}\\ k_{1}=count = nondet();\!\cdot\!1\!\cdot\!(a_{5}\!\cdot\!() = printf(count);\\ \;\;\;\!\cdot\!count = count + 1;)*\!\cdot\!\neg a_{5}\\ k_{2}=count = nondet();\!\cdot\!1\!\cdot\!1\!\cdot\!((a_{11} \wedge b_{12})\!\cdot\!() = printf(count);\\ \;\;\;\!\cdot\!count = count + 1;)*\!\cdot\!\neg a_{11}\end{array}\right.$.
.3 $\left\{\begin{array}{l}Cond: \neg a_{5} \\ k_{1}=count = nondet();\!\cdot\!(a_{5}\!\cdot\!() = printf(count);\\ \;\;\;\!\cdot\!count = count + 1;)*\!\cdot\!\neg a_{5}\\ k_{2}=count = nondet();\!\cdot\!1\!\cdot\!((a_{11} \wedge b_{12})\!\cdot\!() = printf(count);\\ \;\;\;\!\cdot\!count = count + 1;)*\!\cdot\!\neg a_{11}\end{array}\right.$.
.4 AComplete.
.5 $\left\{\begin{array}{l}Axioms: \{D=1, E=1\}\\ k_{1}=count = nondet();\!\cdot\!1\!\cdot\!(a_{5}\!\cdot\!0)*\!\cdot\!\neg a_{5}\\ k_{2}=count = nondet();\!\cdot\!1\!\cdot\!((a_{11} \wedge b_{12})\!\cdot\!() = printf(count);\\ \;\;\;\!\cdot\!count = count + 1;)*\!\cdot\!\neg a_{11}\end{array}\right.$.
}

\end{scriptsize}
\medskip
\emph{Remaining 42 solutions ommitted for brevity.}
\subsection{Example synthesized solution for benchmark \texttt{01concloop.c}}
\begin{scriptsize}
\dirtree{%
.1 solution.
.2 (Partial), cond: $number >= 0$.
.3 $\left\{\begin{array}{l}Cond: b_{12} \\ k_{1}=count = 1;\!\cdot\!(a_{5}\!\cdot\!() = evA(count);\!\cdot\!count = count + 1;)*\!\cdot\!\neg a_{5}\\ k_{2}=count = 1;\!\cdot\!number = nondet();\!\cdot\!((a_{11} \wedge b_{12})\!\cdot\!() = evA(count);\!\cdot\!count = count + 1;)*\!\cdot\!(\neg a_{11} \vee \neg b_{12})\end{array}\right.$.
.4 (Partial), cond: $count <= 4$.
.5 $\left\{\begin{array}{l}Cond: \neg a_{5} \\ k_{1}=count = 1;\!\cdot\!(a_{5}\!\cdot\!() = evA(count);\!\cdot\!count = count + 1;)*\!\cdot\!\neg a_{5}\\ k_{2}=count = 1;\!\cdot\!number = nondet();\!\cdot\!1\!\cdot\!((a_{11} \wedge b_{12})\!\cdot\!() = evA(count);\!\cdot\!count = count + 1;)*\!\cdot\!\neg a_{11}\end{array}\right.$.
.6 (Partial), cond: $count <= 4$.
.7 $\left\{\begin{array}{l}Cond: \neg a_{11} \\ k_{1}=count = 1;\!\cdot\!0\!\cdot\!0\\ k_{2}=count = 1;\!\cdot\!number = nondet();\!\cdot\!1\!\cdot\!((a_{11} \wedge b_{12})\!\cdot\!() = evA(count);\!\cdot\!count = count + 1;)*\!\cdot\!\neg a_{11}\end{array}\right.$.
.8 AComplete.
.9 $\left\{\begin{array}{l}Axioms: \{D=1, U=1, T=1\}\\ k_{1}=count = 1;\!\cdot\!0\!\cdot\!0\\ k_{2}=count = 1;\!\cdot\!number = nondet();\!\cdot\!1\!\cdot\!0\!\cdot\!0\end{array}\right.$.
}

\end{scriptsize}
\medskip
\emph{Remaining 12 solutions ommitted for brevity.}
\subsection{Example synthesized solution for benchmark \texttt{01concloop2.c}}
\begin{scriptsize}
\dirtree{%
.1 solution.
.2 AComplete.
.3 $\left\{\begin{array}{l}Axioms: \{D=1, I=1, E=1, G=1, T=1, U=1\}\\ k_{1}=count = 1;\!\cdot\!(a_{4}\!\cdot\!() = printf(count);\\ \;\;\;\!\cdot\!count = count + 1;)*\!\cdot\!\neg a_{4}\\ k_{2}=count = 1;\!\cdot\!number = nondet();\!\cdot\!() = printf(count);\\ \;\;\;\!\cdot\!fv_{1} = 2;\!\cdot\!number = scanf(fv_{1});\!\cdot\!((b_{11} \wedge a_{12})\!\cdot\!count = count + 1;)*\!\cdot\!(\neg b_{11} \vee \neg a_{12})\end{array}\right.$.
}

\end{scriptsize}
\medskip
\emph{Remaining 238 solutions ommitted for brevity.}
\subsection{Example synthesized solution for benchmark \texttt{01concloop3.c}}
\begin{scriptsize}
\dirtree{%
.1 solution.
.2 (Partial), cond: $number >= 0$.
.3 $\left\{\begin{array}{l}Cond: b_{12} \\ k_{1}=(a_{5}\!\cdot\!() = evA();\!\cdot\!count = count + 1;)*\!\cdot\!\neg a_{5}\\ k_{2}=((a_{11} \wedge b_{12})\!\cdot\!() = evA();\!\cdot\!count = count + 1;)*\!\cdot\!(\neg a_{11} \vee \neg b_{12})\end{array}\right.$.
.4 (Complete), cond: $count <= 4$.
.5 $\left\{\begin{array}{l}Cond: a_{11} \\ k_{1}=(a_{5}\!\cdot\!() = evA();\!\cdot\!count = count + 1;)*\!\cdot\!\neg a_{5}\\ k_{2}=1\!\cdot\!((a_{11} \wedge b_{12})\!\cdot\!() = evA();\!\cdot\!count = count + 1;)*\!\cdot\!\neg a_{11}\end{array}\right.$.
.6 AComplete.
.7 $\left\{\begin{array}{l}Axioms: \{V=1, U=1, D=1, T=1\}\\ k_{1}=(a_{5}\!\cdot\!() = evA();\!\cdot\!count = count + 1;)*\!\cdot\!\neg a_{5}\\ k_{2}=1\!\cdot\!1\!\cdot\!((a_{11} \wedge b_{12})\!\cdot\!() = evA();\!\cdot\!count = count + 1;)*\!\cdot\!\neg a_{11}\end{array}\right.$.
.5 $\left\{\begin{array}{l}Cond: \neg a_{11} \\ k_{1}=(a_{5}\!\cdot\!() = evA();\!\cdot\!count = count + 1;)*\!\cdot\!\neg a_{5}\\ k_{2}=1\!\cdot\!((a_{11} \wedge b_{12})\!\cdot\!() = evA();\!\cdot\!count = count + 1;)*\!\cdot\!\neg a_{11}\end{array}\right.$.
.6 AComplete.
.7 $\left\{\begin{array}{l}Axioms: \{V=1, U=1\}\\ k_{1}=(a_{5}\!\cdot\!() = evA();\!\cdot\!count = count + 1;)*\!\cdot\!\neg a_{5}\\ k_{2}=1\!\cdot\!1\!\cdot\!((a_{11} \wedge b_{12})\!\cdot\!0)*\!\cdot\!\neg a_{11}\end{array}\right.$.
}

\end{scriptsize}
\medskip
\emph{Remaining 125 solutions ommitted for brevity.}
\subsection{Example synthesized solution for benchmark \texttt{01linarith.c}}
\begin{scriptsize}
\dirtree{%
.1 solution.
.2 AComplete.
.3 $\left\{\begin{array}{l}Axioms: \{G=1, J=1, T=1, S=I\}\\ k_{1}=c = nondet();\!\cdot\!servers = 4;\!\cdot\!resp = 0;\!\cdot\!curr\_serv = servers;\!\cdot\!tmp = nondet();\!\cdot\!(a_{5}\!\cdot\!1\!\cdot\!c = c - 1;\!\cdot\!() = shutdown();\\ \;\;\;\!\cdot\!curr\_serv = curr\_serv - 1;\!\cdot\!resp = resp + 1;\!\cdot\!tmp = nondet();)*\!\cdot\!\neg a_{5}\\ k_{2}=c = nondet();\!\cdot\!servers = 4;\!\cdot\!resp = 0;\!\cdot\!curr\_serv = servers;\!\cdot\!tmp = nondet();\!\cdot\!(a_{17}\!\cdot\!1\!\cdot\!() = pingall();\\ \;\;\;\!\cdot\!curr\_serv = curr\_serv - 1;\!\cdot\!c = c - 1;\!\cdot\!resp = resp + 1;\!\cdot\!() = shutdown();\\ \;\;\;\!\cdot\!tmp = nondet();)*\!\cdot\!\neg a_{17}\end{array}\right.$.
}

\end{scriptsize}
\medskip
\emph{Remaining 10 solutions ommitted for brevity.}
\subsection{Example synthesized solution for benchmark \texttt{01loopevent.c}}
\begin{scriptsize}
\dirtree{%
.1 solution.
.2 (Partial), cond: $a > 5$.
.3 $\left\{\begin{array}{l}Cond: \neg c_{16} \\ k_{1}=((a_{5} \vee b_{6})\!\cdot\!() = eventA();\!\cdot\!i = i + 1;\!\cdot\!j = j + 1;)*\!\cdot\!(\neg a_{5} \wedge \neg b_{6})\\ k_{2}=((a_{13} \vee b_{14})\!\cdot\!a = i + j;\!\cdot\!() = eventA();\!\cdot\!i = i + 1;\!\cdot\!j = j + 1;\!\cdot\!(c_{16}\!\cdot\!() = eventA(); + \neg c_{16}\!\cdot\!1))*\!\cdot\!(\neg a_{13} \wedge \neg b_{14})\end{array}\right.$.
.4 AComplete.
.5 $\left\{\begin{array}{l}Axioms: \{D=1\}\\ k_{1}=((a_{5} \vee b_{6})\!\cdot\!() = eventA();\!\cdot\!i = i + 1;\!\cdot\!j = j + 1;)*\!\cdot\!(\neg a_{5} \wedge \neg b_{6})\\ k_{2}=((a_{13} \vee b_{14})\!\cdot\!a = i + j;\!\cdot\!() = eventA();\!\cdot\!i = i + 1;\!\cdot\!j = j + 1;\!\cdot\!1\!\cdot\!1)*\!\cdot\!(\neg a_{13} \wedge \neg b_{14})\end{array}\right.$.
}

\end{scriptsize}
\medskip
\emph{Remaining 65 solutions ommitted for brevity.}
\subsection{Example synthesized solution for benchmark \texttt{01loopprint.c}}
\begin{scriptsize}
\dirtree{%
.1 solution.
.2 (Complete), cond: $a > 5$.
.3 $\left\{\begin{array}{l}Cond: c_{17} \\ k_{1}=((a_{5} \vee b_{6})\!\cdot\!() = printf(i, j);\\ \;\;\;\!\cdot\!i = i + 1;\!\cdot\!j = j + 1;)*\!\cdot\!(\neg a_{5} \wedge \neg b_{6})\\ k_{2}=((a_{13} \vee b_{14})\!\cdot\!a = i + j;\!\cdot\!() = printf(i, j);\\ \;\;\;\!\cdot\!i = i + 1;\!\cdot\!j = j + 1;\!\cdot\!(c_{17}\!\cdot\!fv_{1} = 0;\!\cdot\!() = printf(a, fv_{1}); + \neg c_{17}\!\cdot\!1))*\!\cdot\!(\neg a_{13} \wedge \neg b_{14})\end{array}\right.$.
.4 AComplete.
.5 $\left\{\begin{array}{l}Axioms: \{D=1, E=1, G=1\}\\ k_{1}=((a_{5} \vee b_{6})\!\cdot\!() = printf(i, j);\\ \;\;\;\!\cdot\!i = i + 1;\!\cdot\!j = j + 1;)*\!\cdot\!(\neg a_{5} \wedge \neg b_{6})\\ k_{2}=((a_{13} \vee b_{14})\!\cdot\!a = i + j;\!\cdot\!() = printf(i, j);\\ \;\;\;\!\cdot\!i = i + 1;\!\cdot\!j = j + 1;\!\cdot\!1\!\cdot\!fv_{1} = 0;\!\cdot\!() = printf(a, fv_{1});)*\!\cdot\!(\neg a_{13} \wedge \neg b_{14})\end{array}\right.$.
.3 $\left\{\begin{array}{l}Cond: \neg c_{17} \\ k_{1}=((a_{5} \vee b_{6})\!\cdot\!() = printf(i, j);\\ \;\;\;\!\cdot\!i = i + 1;\!\cdot\!j = j + 1;)*\!\cdot\!(\neg a_{5} \wedge \neg b_{6})\\ k_{2}=((a_{13} \vee b_{14})\!\cdot\!a = i + j;\!\cdot\!() = printf(i, j);\\ \;\;\;\!\cdot\!i = i + 1;\!\cdot\!j = j + 1;\!\cdot\!(c_{17}\!\cdot\!fv_{1} = 0;\!\cdot\!() = printf(a, fv_{1}); + \neg c_{17}\!\cdot\!1))*\!\cdot\!(\neg a_{13} \wedge \neg b_{14})\end{array}\right.$.
.4 AComplete.
.5 $\left\{\begin{array}{l}Axioms: \{D=1\}\\ k_{1}=((a_{5} \vee b_{6})\!\cdot\!() = printf(i, j);\\ \;\;\;\!\cdot\!i = i + 1;\!\cdot\!j = j + 1;)*\!\cdot\!(\neg a_{5} \wedge \neg b_{6})\\ k_{2}=((a_{13} \vee b_{14})\!\cdot\!a = i + j;\!\cdot\!() = printf(i, j);\\ \;\;\;\!\cdot\!i = i + 1;\!\cdot\!j = j + 1;\!\cdot\!1\!\cdot\!1)*\!\cdot\!(\neg a_{13} \wedge \neg b_{14})\end{array}\right.$.
}

\end{scriptsize}
\medskip
\emph{Remaining 10 solutions ommitted for brevity.}
\subsection{Example synthesized solution for benchmark \texttt{01sendrecv.c}}
\begin{scriptsize}
\dirtree{%
.1 solution.
.2 (Partial), cond: $c > 0$.
.3 $\left\{\begin{array}{l}Cond: d_{30} \\ k_{1}=(a_{8}\!\cdot\!b = recv();\!\cdot\!(c_{15}\!\cdot\!auth = check(b);\!\cdot\!(b_{12}\!\cdot\!n = constructReply();\!\cdot\!() = send(n);\\ \;\;\; + \neg b_{12}\!\cdot\!1) + \neg c_{15}\!\cdot\!() = log(b);)\!\cdot\!x = x - 1;)*\!\cdot\!\neg a_{8}\\ k_{2}=(a_{22}\!\cdot\!b = recv();\!\cdot\!(d_{30}\!\cdot\!() = log(b); + \neg d_{30}\!\cdot\!1)\!\cdot\!(c_{28}\!\cdot\!n = constructReply();\!\cdot\!() = send(n);\\ \;\;\;\!\cdot\!(d_{25}\!\cdot\!() = log(n); + \neg d_{25}\!\cdot\!1) + \neg c_{28}\!\cdot\!1)\!\cdot\!x = x - 1;)*\!\cdot\!\neg a_{22}\end{array}\right.$.
.4 (Partial), cond: $b > 0$.
.5 $\left\{\begin{array}{l}Cond: c_{28} \\ k_{1}=(a_{8}\!\cdot\!b = recv();\!\cdot\!(c_{15}\!\cdot\!auth = check(b);\!\cdot\!(b_{12}\!\cdot\!n = constructReply();\!\cdot\!() = send(n);\\ \;\;\; + \neg b_{12}\!\cdot\!1) + \neg c_{15}\!\cdot\!() = log(b);)\!\cdot\!x = x - 1;)*\!\cdot\!\neg a_{8}\\ k_{2}=(a_{22}\!\cdot\!b = recv();\!\cdot\!1\!\cdot\!() = log(b);\!\cdot\!(c_{28}\!\cdot\!n = constructReply();\!\cdot\!() = send(n);\\ \;\;\;\!\cdot\!() = log(n); + \neg c_{28}\!\cdot\!1)\!\cdot\!x = x - 1;)*\!\cdot\!\neg a_{22}\end{array}\right.$.
.6 (Partial), cond: $b > 0$.
.7 $\left\{\begin{array}{l}Cond: c_{15} \\ k_{1}=(a_{8}\!\cdot\!b = recv();\!\cdot\!(c_{15}\!\cdot\!auth = check(b);\!\cdot\!(b_{12}\!\cdot\!n = constructReply();\!\cdot\!() = send(n);\\ \;\;\; + \neg b_{12}\!\cdot\!1) + \neg c_{15}\!\cdot\!() = log(b);)\!\cdot\!x = x - 1;)*\!\cdot\!\neg a_{8}\\ k_{2}=(a_{22}\!\cdot\!b = recv();\!\cdot\!1\!\cdot\!() = log(b);\!\cdot\!1\!\cdot\!n = constructReply();\!\cdot\!() = send(n);\\ \;\;\;\!\cdot\!() = log(n);\!\cdot\!x = x - 1;)*\!\cdot\!\neg a_{22}\end{array}\right.$.
.8 (Partial), cond: $auth > 0$.
.9 $\left\{\begin{array}{l}Cond: b_{12} \\ k_{1}=(a_{8}\!\cdot\!b = recv();\!\cdot\!1\!\cdot\!auth = check(b);\!\cdot\!(b_{12}\!\cdot\!n = constructReply();\!\cdot\!() = send(n);\\ \;\;\; + \neg b_{12}\!\cdot\!1)\!\cdot\!x = x - 1;)*\!\cdot\!\neg a_{8}\\ k_{2}=(a_{22}\!\cdot\!b = recv();\!\cdot\!1\!\cdot\!() = log(b);\!\cdot\!1\!\cdot\!n = constructReply();\!\cdot\!() = send(n);\\ \;\;\;\!\cdot\!() = log(n);\!\cdot\!x = x - 1;)*\!\cdot\!\neg a_{22}\end{array}\right.$.
.10 AComplete.
.11 $\left\{\begin{array}{l}Axioms: \{I=1, J=1, K=1, M=1\}\\ k_{1}=(a_{8}\!\cdot\!b = recv();\!\cdot\!1\!\cdot\!auth = check(b);\!\cdot\!1\!\cdot\!n = constructReply();\!\cdot\!() = send(n);\\ \;\;\;\!\cdot\!x = x - 1;)*\!\cdot\!\neg a_{8}\\ k_{2}=(a_{22}\!\cdot\!b = recv();\!\cdot\!1\!\cdot\!() = log(b);\!\cdot\!1\!\cdot\!n = constructReply();\!\cdot\!() = send(n);\\ \;\;\;\!\cdot\!() = log(n);\!\cdot\!x = x - 1;)*\!\cdot\!\neg a_{22}\end{array}\right.$.
.5 $\left\{\begin{array}{l}Cond: \neg c_{28} \\ k_{1}=(a_{8}\!\cdot\!b = recv();\!\cdot\!(c_{15}\!\cdot\!auth = check(b);\!\cdot\!(b_{12}\!\cdot\!n = constructReply();\!\cdot\!() = send(n);\\ \;\;\; + \neg b_{12}\!\cdot\!1) + \neg c_{15}\!\cdot\!() = log(b);)\!\cdot\!x = x - 1;)*\!\cdot\!\neg a_{8}\\ k_{2}=(a_{22}\!\cdot\!b = recv();\!\cdot\!1\!\cdot\!() = log(b);\!\cdot\!(c_{28}\!\cdot\!n = constructReply();\!\cdot\!() = send(n);\\ \;\;\;\!\cdot\!() = log(n); + \neg c_{28}\!\cdot\!1)\!\cdot\!x = x - 1;)*\!\cdot\!\neg a_{22}\end{array}\right.$.
.6 (Partial), cond: $b > 0$.
.7 $\left\{\begin{array}{l}Cond: c_{15} \\ k_{1}=(a_{8}\!\cdot\!b = recv();\!\cdot\!(c_{15}\!\cdot\!auth = check(b);\!\cdot\!(b_{12}\!\cdot\!n = constructReply();\!\cdot\!() = send(n);\\ \;\;\; + \neg b_{12}\!\cdot\!1) + \neg c_{15}\!\cdot\!() = log(b);)\!\cdot\!x = x - 1;)*\!\cdot\!\neg a_{8}\\ k_{2}=(a_{22}\!\cdot\!b = recv();\!\cdot\!1\!\cdot\!() = log(b);\!\cdot\!1\!\cdot\!1\!\cdot\!x = x - 1;)*\!\cdot\!\neg a_{22}\end{array}\right.$.
.8 (Partial), cond: $auth > 0$.
.9 $\left\{\begin{array}{l}Cond: \neg b_{12} \\ k_{1}=(a_{8}\!\cdot\!b = recv();\!\cdot\!1\!\cdot\!auth = check(b);\!\cdot\!(b_{12}\!\cdot\!n = constructReply();\!\cdot\!() = send(n);\\ \;\;\; + \neg b_{12}\!\cdot\!1)\!\cdot\!x = x - 1;)*\!\cdot\!\neg a_{8}\\ k_{2}=(a_{22}\!\cdot\!b = recv();\!\cdot\!1\!\cdot\!() = log(b);\!\cdot\!1\!\cdot\!1\!\cdot\!x = x - 1;)*\!\cdot\!\neg a_{22}\end{array}\right.$.
.10 AComplete.
.11 $\left\{\begin{array}{l}Axioms: \{I=1, J=1, K=1\}\\ k_{1}=(a_{8}\!\cdot\!b = recv();\!\cdot\!1\!\cdot\!auth = check(b);\!\cdot\!1\!\cdot\!1\!\cdot\!x = x - 1;)*\!\cdot\!\neg a_{8}\\ k_{2}=(a_{22}\!\cdot\!b = recv();\!\cdot\!1\!\cdot\!() = log(b);\!\cdot\!1\!\cdot\!1\!\cdot\!x = x - 1;)*\!\cdot\!\neg a_{22}\end{array}\right.$.
.7 $\left\{\begin{array}{l}Cond: \neg c_{15} \\ k_{1}=(a_{8}\!\cdot\!b = recv();\!\cdot\!(c_{15}\!\cdot\!auth = check(b);\!\cdot\!(b_{12}\!\cdot\!n = constructReply();\!\cdot\!() = send(n);\\ \;\;\; + \neg b_{12}\!\cdot\!1) + \neg c_{15}\!\cdot\!() = log(b);)\!\cdot\!x = x - 1;)*\!\cdot\!\neg a_{8}\\ k_{2}=(a_{22}\!\cdot\!b = recv();\!\cdot\!1\!\cdot\!() = log(b);\!\cdot\!1\!\cdot\!1\!\cdot\!x = x - 1;)*\!\cdot\!\neg a_{22}\end{array}\right.$.
.8 AComplete.
.9 $\left\{\begin{array}{l}Axioms: \{I=1, J=1, K=1\}\\ k_{1}=(a_{8}\!\cdot\!b = recv();\!\cdot\!1\!\cdot\!() = log(b);\!\cdot\!x = x - 1;)*\!\cdot\!\neg a_{8}\\ k_{2}=(a_{22}\!\cdot\!b = recv();\!\cdot\!1\!\cdot\!() = log(b);\!\cdot\!1\!\cdot\!1\!\cdot\!x = x - 1;)*\!\cdot\!\neg a_{22}\end{array}\right.$.
.3 $\left\{\begin{array}{l}Cond: \neg d_{30} \\ k_{1}=(a_{8}\!\cdot\!b = recv();\!\cdot\!(c_{15}\!\cdot\!auth = check(b);\!\cdot\!(b_{12}\!\cdot\!n = constructReply();\!\cdot\!() = send(n);\\ \;\;\; + \neg b_{12}\!\cdot\!1) + \neg c_{15}\!\cdot\!() = log(b);)\!\cdot\!x = x - 1;)*\!\cdot\!\neg a_{8}\\ k_{2}=(a_{22}\!\cdot\!b = recv();\!\cdot\!(d_{30}\!\cdot\!() = log(b); + \neg d_{30}\!\cdot\!1)\!\cdot\!(c_{28}\!\cdot\!n = constructReply();\!\cdot\!() = send(n);\\ \;\;\;\!\cdot\!(d_{25}\!\cdot\!() = log(n); + \neg d_{25}\!\cdot\!1) + \neg c_{28}\!\cdot\!1)\!\cdot\!x = x - 1;)*\!\cdot\!\neg a_{22}\end{array}\right.$.
.4 (Partial), cond: $b > 0$.
.5 $\left\{\begin{array}{l}Cond: \neg c_{28} \\ k_{1}=(a_{8}\!\cdot\!b = recv();\!\cdot\!(c_{15}\!\cdot\!auth = check(b);\!\cdot\!(b_{12}\!\cdot\!n = constructReply();\!\cdot\!() = send(n);\\ \;\;\; + \neg b_{12}\!\cdot\!1) + \neg c_{15}\!\cdot\!() = log(b);)\!\cdot\!x = x - 1;)*\!\cdot\!\neg a_{8}\\ k_{2}=(a_{22}\!\cdot\!b = recv();\!\cdot\!1\!\cdot\!1\!\cdot\!(c_{28}\!\cdot\!n = constructReply();\!\cdot\!() = send(n);\\ \;\;\;\!\cdot\!1 + \neg c_{28}\!\cdot\!1)\!\cdot\!x = x - 1;)*\!\cdot\!\neg a_{22}\end{array}\right.$.
.6 (Partial), cond: $b > 0$.
.7 $\left\{\begin{array}{l}Cond: c_{15} \\ k_{1}=(a_{8}\!\cdot\!b = recv();\!\cdot\!(c_{15}\!\cdot\!auth = check(b);\!\cdot\!(b_{12}\!\cdot\!n = constructReply();\!\cdot\!() = send(n);\\ \;\;\; + \neg b_{12}\!\cdot\!1) + \neg c_{15}\!\cdot\!() = log(b);)\!\cdot\!x = x - 1;)*\!\cdot\!\neg a_{8}\\ k_{2}=(a_{22}\!\cdot\!b = recv();\!\cdot\!1\!\cdot\!1\!\cdot\!1\!\cdot\!1\!\cdot\!x = x - 1;)*\!\cdot\!\neg a_{22}\end{array}\right.$.
.8 (Partial), cond: $auth > 0$.
.9 $\left\{\begin{array}{l}Cond: \neg b_{12} \\ k_{1}=(a_{8}\!\cdot\!b = recv();\!\cdot\!1\!\cdot\!auth = check(b);\!\cdot\!(b_{12}\!\cdot\!n = constructReply();\!\cdot\!() = send(n);\\ \;\;\; + \neg b_{12}\!\cdot\!1)\!\cdot\!x = x - 1;)*\!\cdot\!\neg a_{8}\\ k_{2}=(a_{22}\!\cdot\!b = recv();\!\cdot\!1\!\cdot\!1\!\cdot\!1\!\cdot\!1\!\cdot\!x = x - 1;)*\!\cdot\!\neg a_{22}\end{array}\right.$.
.10 AComplete.
.11 $\left\{\begin{array}{l}Axioms: \{I=1, K=1\}\\ k_{1}=(a_{8}\!\cdot\!b = recv();\!\cdot\!1\!\cdot\!auth = check(b);\!\cdot\!1\!\cdot\!1\!\cdot\!x = x - 1;)*\!\cdot\!\neg a_{8}\\ k_{2}=(a_{22}\!\cdot\!b = recv();\!\cdot\!1\!\cdot\!1\!\cdot\!1\!\cdot\!1\!\cdot\!x = x - 1;)*\!\cdot\!\neg a_{22}\end{array}\right.$.
.7 $\left\{\begin{array}{l}Cond: \neg c_{15} \\ k_{1}=(a_{8}\!\cdot\!b = recv();\!\cdot\!(c_{15}\!\cdot\!auth = check(b);\!\cdot\!(b_{12}\!\cdot\!n = constructReply();\!\cdot\!() = send(n);\\ \;\;\; + \neg b_{12}\!\cdot\!1) + \neg c_{15}\!\cdot\!() = log(b);)\!\cdot\!x = x - 1;)*\!\cdot\!\neg a_{8}\\ k_{2}=(a_{22}\!\cdot\!b = recv();\!\cdot\!1\!\cdot\!1\!\cdot\!1\!\cdot\!1\!\cdot\!x = x - 1;)*\!\cdot\!\neg a_{22}\end{array}\right.$.
.8 AComplete.
.9 $\left\{\begin{array}{l}Axioms: \{I=1, K=1\}\\ k_{1}=(a_{8}\!\cdot\!b = recv();\!\cdot\!1\!\cdot\!() = log(b);\!\cdot\!x = x - 1;)*\!\cdot\!\neg a_{8}\\ k_{2}=(a_{22}\!\cdot\!b = recv();\!\cdot\!1\!\cdot\!1\!\cdot\!1\!\cdot\!1\!\cdot\!x = x - 1;)*\!\cdot\!\neg a_{22}\end{array}\right.$.
}

\end{scriptsize}
\medskip
\emph{Remaining 73 solutions ommitted for brevity.}
\subsection{Example synthesized solution for benchmark \texttt{01toggle.c}}
\begin{scriptsize}
\dirtree{%
.1 solution.
.2 AComplete.
.3 $\left\{\begin{array}{l}Axioms: \{b=!c\}\\ k_{1}=i = 0;\!\cdot\!(a_{5}\!\cdot\!n = nondet();\!\cdot\!(b_{9}\!\cdot\!t = 1; + \neg b_{9}\!\cdot\!t = 0;)\!\cdot\!i = i + 1;)*\!\cdot\!\neg a_{5}\\ k_{2}=i = 0;\!\cdot\!(a_{14}\!\cdot\!n = nondet();\!\cdot\!(c_{18}\!\cdot\!t = 0; + \neg c_{18}\!\cdot\!t = 1;)\!\cdot\!i = i + 1;)*\!\cdot\!\neg a_{14}\end{array}\right.$.
}

\end{scriptsize}
\subsection{Example synthesized solution for benchmark \texttt{02altern.c}}
\begin{scriptsize}
\dirtree{%
.1 solution.
.2 AComplete.
.3 $\left\{\begin{array}{l}Axioms: \{T=1, B=1\}\\ k_{1}=(a_{5}\!\cdot\!x = x - 1;\!\cdot\!() = eventA();\!\cdot\!() = eventB();)*\!\cdot\!\neg a_{5}\\ k_{2}=(a_{11}\!\cdot\!x = x - 1;\!\cdot\!() = eventB();\!\cdot\!() = eventA();)*\!\cdot\!\neg a_{11}\end{array}\right.$.
}

\end{scriptsize}
\subsection{Example synthesized solution for benchmark \texttt{02cdown.c}}
\begin{scriptsize}
\dirtree{%
.1 solution.
.2 AComplete.
.3 $\left\{\begin{array}{l}Axioms: \{X=A\}\\ k_{1}=() = eventA();\!\cdot\!(a_{7}\!\cdot\!x = x - 1;)*\!\cdot\!\neg a_{7}\!\cdot\!() = eventB();\!\cdot\!() = eventA();\\ k_{2}=() = eventA();\!\cdot\!(a_{14}\!\cdot\!x = x - 2;)*\!\cdot\!\neg a_{14}\!\cdot\!() = eventB();\!\cdot\!() = eventA();\end{array}\right.$.
}

\end{scriptsize}
\subsection{Example synthesized solution for benchmark \texttt{02foil.c}}
\begin{scriptsize}
\dirtree{%
.1 solution.
.2 AComplete.
.3 $\left\{\begin{array}{l}Axioms: \{E=V\}\\ k_{1}=() = eventA();\!\cdot\!() = eventB();\!\cdot\!() = eventB();\\ k_{2}=() = eventA();\!\cdot\!() = eventA();\!\cdot\!() = eventB();\end{array}\right.$.
}

\end{scriptsize}
\subsection{Example synthesized solution for benchmark \texttt{03buffer.c}}
\begin{scriptsize}
\dirtree{%
.1 solution.
.2 (Partial), cond: $brk < 1$.
.3 $\left\{\begin{array}{l}Cond: a_{6} \\ k_{1}=fv_{1} = 1024;\!\cdot\!buffer = array\_alloc(fv_{1});\!\cdot\!i = 0;\!\cdot\!brk = 0;\!\cdot\!(a_{6}\!\cdot\!c = getchar();\\ \;\;\;\!\cdot\!(c_{12}\!\cdot\!brk = 1; + \neg c_{12}\!\cdot\!(b_{11}\!\cdot\!brk = 1; + \neg b_{11}\!\cdot\!() = array\_write(buffer, i, c);\!\cdot\!i = i + 1;)))*\!\cdot\!\neg a_{6}\\ k_{2}=fv_{2} = 1024;\!\cdot\!buffer = array\_alloc(fv_{2});\!\cdot\!i = 0;\!\cdot\!brk = 0;\!\cdot\!(a_{19}\!\cdot\!c = getchar();\\ \;\;\;\!\cdot\!(c_{24}\!\cdot\!brk = 1; + \neg c_{24}\!\cdot\!(b_{23}\!\cdot\!brk = 1; + \neg b_{23}\!\cdot\!brk = 1;)))*\!\cdot\!\neg a_{19}\end{array}\right.$.
.4 AComplete.
.5 $\left\{\begin{array}{l}Axioms: \{M=1, N=1, Y=L, W=1, O=1, P=1\}\\ k_{1}=fv_{1} = 1024;\!\cdot\!buffer = array\_alloc(fv_{1});\!\cdot\!i = 0;\!\cdot\!brk = 0;\!\cdot\!1\!\cdot\!(a_{6}\!\cdot\!c = getchar();\\ \;\;\;\!\cdot\!(c_{12}\!\cdot\!brk = 1; + \neg c_{12}\!\cdot\!(b_{11}\!\cdot\!brk = 1; + \neg b_{11}\!\cdot\!() = array\_write(buffer, i, c);\!\cdot\!i = i + 1;)))*\!\cdot\!\neg a_{6}\\ k_{2}=fv_{2} = 1024;\!\cdot\!buffer = array\_alloc(fv_{2});\!\cdot\!i = 0;\!\cdot\!brk = 0;\!\cdot\!(a_{19}\!\cdot\!c = getchar();\\ \;\;\;\!\cdot\!(c_{24}\!\cdot\!brk = 1; + \neg c_{24}\!\cdot\!(b_{23}\!\cdot\!brk = 1; + \neg b_{23}\!\cdot\!brk = 1;)))*\!\cdot\!\neg a_{19}\end{array}\right.$.
.3 $\left\{\begin{array}{l}Cond: \neg a_{6} \\ k_{1}=fv_{1} = 1024;\!\cdot\!buffer = array\_alloc(fv_{1});\!\cdot\!i = 0;\!\cdot\!brk = 0;\!\cdot\!(a_{6}\!\cdot\!c = getchar();\\ \;\;\;\!\cdot\!(c_{12}\!\cdot\!brk = 1; + \neg c_{12}\!\cdot\!(b_{11}\!\cdot\!brk = 1; + \neg b_{11}\!\cdot\!() = array\_write(buffer, i, c);\!\cdot\!i = i + 1;)))*\!\cdot\!\neg a_{6}\\ k_{2}=fv_{2} = 1024;\!\cdot\!buffer = array\_alloc(fv_{2});\!\cdot\!i = 0;\!\cdot\!brk = 0;\!\cdot\!(a_{19}\!\cdot\!c = getchar();\\ \;\;\;\!\cdot\!(c_{24}\!\cdot\!brk = 1; + \neg c_{24}\!\cdot\!(b_{23}\!\cdot\!brk = 1; + \neg b_{23}\!\cdot\!brk = 1;)))*\!\cdot\!\neg a_{19}\end{array}\right.$.
.4 (Partial), cond: $brk < 1$.
.5 $\left\{\begin{array}{l}Cond: \neg a_{19} \\ k_{1}=fv_{1} = 1024;\!\cdot\!buffer = array\_alloc(fv_{1});\!\cdot\!i = 0;\!\cdot\!brk = 0;\!\cdot\!0\!\cdot\!0\\ k_{2}=fv_{2} = 1024;\!\cdot\!buffer = array\_alloc(fv_{2});\!\cdot\!i = 0;\!\cdot\!brk = 0;\!\cdot\!(a_{19}\!\cdot\!c = getchar();\\ \;\;\;\!\cdot\!(c_{24}\!\cdot\!brk = 1; + \neg c_{24}\!\cdot\!(b_{23}\!\cdot\!brk = 1; + \neg b_{23}\!\cdot\!brk = 1;)))*\!\cdot\!\neg a_{19}\end{array}\right.$.
.6 AComplete.
.7 $\left\{\begin{array}{l}Axioms: \{M=1, N=1, Y=L, W=1, O=1\}\\ k_{1}=fv_{1} = 1024;\!\cdot\!buffer = array\_alloc(fv_{1});\!\cdot\!i = 0;\!\cdot\!brk = 0;\!\cdot\!0\!\cdot\!0\\ k_{2}=fv_{2} = 1024;\!\cdot\!buffer = array\_alloc(fv_{2});\!\cdot\!i = 0;\!\cdot\!brk = 0;\!\cdot\!0\!\cdot\!0\end{array}\right.$.
}

\end{scriptsize}
\medskip
\emph{Remaining 190 solutions ommitted for brevity.}
\subsection{Example synthesized solution for benchmark \texttt{03syscalls.c}}
\begin{scriptsize}
\dirtree{%
.1 solution.
.2 (Complete), cond: $x == 0$.
.3 $\left\{\begin{array}{l}Cond: b_{28} \\ k_{1}=fv_{1} = 1000;\!\cdot\!c = array\_alloc(fv_{1});\!\cdot\!(b_{15}\!\cdot\!fv_{2} = 1;\!\cdot\!() = show(fv_{2}); + \neg b_{15}\!\cdot\!1)\!\cdot\!b = getchar();\\ \;\;\;\!\cdot\!(a_{9}\!\cdot\!() = array\_write(c, b);\!\cdot\!b = b - 1;)*\!\cdot\!\neg a_{9}\!\cdot\!r = array\_read(c);\\ k_{2}=a = nondet();\!\cdot\!c = array\_alloc(a);\!\cdot\!(b_{28}\!\cdot\!fv_{3} = 2;\!\cdot\!() = show(fv_{3}); + \neg b_{28}\!\cdot\!1)\!\cdot\!b = getchar();\\ \;\;\;\!\cdot\!((a_{21} \wedge c_{22})\!\cdot\!() = array\_write(c, b);\!\cdot\!b = b - 1;)*\!\cdot\!(\neg a_{21} \vee \neg c_{22})\!\cdot\!r = array\_read(c);\end{array}\right.$.
.4 AComplete.
.5 $\left\{\begin{array}{l}Axioms: \{K=1, L=T, M=1, P=1, W=1, !a=!c, Q=1, U=1\}\\ k_{1}=fv_{1} = 1000;\!\cdot\!c = array\_alloc(fv_{1});\!\cdot\!(b_{15}\!\cdot\!fv_{2} = 1;\!\cdot\!() = show(fv_{2}); + \neg b_{15}\!\cdot\!1)\!\cdot\!b = getchar();\\ \;\;\;\!\cdot\!(a_{9}\!\cdot\!() = array\_write(c, b);\!\cdot\!b = b - 1;)*\!\cdot\!\neg a_{9}\!\cdot\!r = array\_read(c);\\ k_{2}=a = nondet();\!\cdot\!c = array\_alloc(a);\!\cdot\!1\!\cdot\!(b_{28}\!\cdot\!fv_{3} = 2;\!\cdot\!() = show(fv_{3}); + \neg b_{28}\!\cdot\!1)\!\cdot\!b = getchar();\\ \;\;\;\!\cdot\!((a_{21} \wedge c_{22})\!\cdot\!() = array\_write(c, b);\!\cdot\!b = b - 1;)*\!\cdot\!(\neg a_{21} \vee \neg c_{22})\!\cdot\!r = array\_read(c);\end{array}\right.$.
.3 $\left\{\begin{array}{l}Cond: \neg b_{28} \\ k_{1}=fv_{1} = 1000;\!\cdot\!c = array\_alloc(fv_{1});\!\cdot\!(b_{15}\!\cdot\!fv_{2} = 1;\!\cdot\!() = show(fv_{2}); + \neg b_{15}\!\cdot\!1)\!\cdot\!b = getchar();\\ \;\;\;\!\cdot\!(a_{9}\!\cdot\!() = array\_write(c, b);\!\cdot\!b = b - 1;)*\!\cdot\!\neg a_{9}\!\cdot\!r = array\_read(c);\\ k_{2}=a = nondet();\!\cdot\!c = array\_alloc(a);\!\cdot\!(b_{28}\!\cdot\!fv_{3} = 2;\!\cdot\!() = show(fv_{3}); + \neg b_{28}\!\cdot\!1)\!\cdot\!b = getchar();\\ \;\;\;\!\cdot\!((a_{21} \wedge c_{22})\!\cdot\!() = array\_write(c, b);\!\cdot\!b = b - 1;)*\!\cdot\!(\neg a_{21} \vee \neg c_{22})\!\cdot\!r = array\_read(c);\end{array}\right.$.
.4 AComplete.
.5 $\left\{\begin{array}{l}Axioms: \{K=1, L=T, M=1, P=1, W=1, !a=!c, Q=1, U=1\}\\ k_{1}=fv_{1} = 1000;\!\cdot\!c = array\_alloc(fv_{1});\!\cdot\!(b_{15}\!\cdot\!fv_{2} = 1;\!\cdot\!() = show(fv_{2}); + \neg b_{15}\!\cdot\!1)\!\cdot\!b = getchar();\\ \;\;\;\!\cdot\!(a_{9}\!\cdot\!() = array\_write(c, b);\!\cdot\!b = b - 1;)*\!\cdot\!\neg a_{9}\!\cdot\!r = array\_read(c);\\ k_{2}=a = nondet();\!\cdot\!c = array\_alloc(a);\!\cdot\!1\!\cdot\!(b_{28}\!\cdot\!fv_{3} = 2;\!\cdot\!() = show(fv_{3}); + \neg b_{28}\!\cdot\!1)\!\cdot\!b = getchar();\\ \;\;\;\!\cdot\!((a_{21} \wedge c_{22})\!\cdot\!() = array\_write(c, b);\!\cdot\!b = b - 1;)*\!\cdot\!(\neg a_{21} \vee \neg c_{22})\!\cdot\!r = array\_read(c);\end{array}\right.$.
}

\end{scriptsize}
\medskip
\emph{Remaining 154 solutions ommitted for brevity.}
\subsection{Example synthesized solution for benchmark \texttt{04ident.c}}
\begin{scriptsize}
\dirtree{%
.1 solution.
.2 AComplete.
.3 $\left\{\begin{array}{l}Axioms: \{J=1, K=1, E=B\}\\ k_{1}=err = copyin(uap\_alen, len);\!\cdot\!(b_{17}\!\cdot\!() = fdrop(fp, p); + \neg b_{17}\!\cdot\!1)\!\cdot\!(c_{15}\!\cdot\!len = sa\_len; + \neg c_{15}\!\cdot\!1)\!\cdot\!\\ \;\;\; err = copyout(sa, uap\_asa, len);\!\cdot\!(b_{12}\!\cdot\!(a_{11}\!\cdot\!fv_{1} = 42;\!\cdot\!() = free(sa, fv_{1}); + \neg a_{11}\!\cdot\!1)\!\cdot\!() = fdrop(fp, p); + \neg b_{12}\!\cdot\!1)\!\cdot\!\\ \;\;\; err = copyout(len, uap\_alen, sizeof\_len);\\ k_{2}=err = copyin(uap\_alen, len);\!\cdot\!(b_{32}\!\cdot\!() = fdrop(fp, p); + \neg b_{32}\!\cdot\!1)\!\cdot\!(c_{30}\!\cdot\!len = sa\_len; + \neg c_{30}\!\cdot\!1)\!\cdot\!\\ \;\;\; err = copyout(sa, uap\_asa, len);\!\cdot\!(b_{27}\!\cdot\!(a_{26}\!\cdot\!fv_{2} = 42;\!\cdot\!() = free(sa, fv_{2}); + \neg a_{26}\!\cdot\!1)\!\cdot\!() = fdrop(fp, p); + \neg b_{27}\!\cdot\!1)\!\cdot\!\\ \;\;\; err = copyout(len, uap\_alen, sizeof\_len);\end{array}\right.$.
}

\end{scriptsize}
\medskip
\emph{Remaining 4 solutions ommitted for brevity.}
\subsection{Example synthesized solution for benchmark \texttt{05thttpdEr.c}}
\begin{scriptsize}
\dirtree{%
.1 solution.
.2 (Complete), cond: $keepalive <= 0$.
.3 $\left\{\begin{array}{l}Cond: a_{6} \\ k_{1}=(b_{8}\!\cdot\!(a_{6}\!\cdot\!() = shutdown();\\ \;\;\; + \neg a_{6}\!\cdot\!1) + \neg b_{8}\!\cdot\!() = update\_stats();)\\ k_{2}=(b_{18}\!\cdot\!() = clear\_connection();\!\cdot\!() = shutdown();\\ \;\;\; + \neg b_{18}\!\cdot\!() = update\_stats();)\end{array}\right.$.
.4 AComplete.
.5 $\left\{\begin{array}{l}Axioms: \{I=1\}\\ k_{1}=(b_{8}\!\cdot\!1\!\cdot\!() = shutdown();\\ \;\;\; + \neg b_{8}\!\cdot\!() = update\_stats();)\\ k_{2}=(b_{18}\!\cdot\!() = clear\_connection();\!\cdot\!() = shutdown();\\ \;\;\; + \neg b_{18}\!\cdot\!() = update\_stats();)\end{array}\right.$.
.3 $\left\{\begin{array}{l}Cond: \neg a_{6} \\ k_{1}=(b_{8}\!\cdot\!(a_{6}\!\cdot\!() = shutdown();\\ \;\;\; + \neg a_{6}\!\cdot\!1) + \neg b_{8}\!\cdot\!() = update\_stats();)\\ k_{2}=(b_{18}\!\cdot\!() = clear\_connection();\!\cdot\!() = shutdown();\\ \;\;\; + \neg b_{18}\!\cdot\!() = update\_stats();)\end{array}\right.$.
.4 (Complete), cond: $err > 0$.
.5 $\left\{\begin{array}{l}Cond: \neg b_{18} \\ k_{1}=(b_{8}\!\cdot\!1\!\cdot\!1 + \neg b_{8}\!\cdot\!() = update\_stats();)\\ k_{2}=(b_{18}\!\cdot\!() = clear\_connection();\!\cdot\!() = shutdown();\\ \;\;\; + \neg b_{18}\!\cdot\!() = update\_stats();)\end{array}\right.$.
.6 (Complete), cond: $err > 0$.
.7 $\left\{\begin{array}{l}Cond: \neg b_{8} \\ k_{1}=(b_{8}\!\cdot\!1\!\cdot\!1 + \neg b_{8}\!\cdot\!() = update\_stats();)\\ k_{2}=1\!\cdot\!() = update\_stats();\end{array}\right.$.
.8 AComplete.
.9 $\left\{\begin{array}{l}Axioms: \{I=1\}\\ k_{1}=1\!\cdot\!() = update\_stats();\\ k_{2}=1\!\cdot\!() = update\_stats();\end{array}\right.$.
}

\end{scriptsize}
\medskip
\emph{Remaining 3 solutions ommitted for brevity.}
\subsection{Example synthesized solution for benchmark \texttt{05thttpdWr.c}}
\begin{scriptsize}
\dirtree{%
.1 solution.
.2 (Partial), cond: $compress > 0$.
.3 $\left\{\begin{array}{l}Cond: a_{12} \\ k_{1}=(a_{12}\!\cdot\!() = compress();\!\cdot\!() = write\_headers();\\ \;\;\;\!\cdot\!() = httpd\_ssl\_write(); + \neg a_{12}\!\cdot\!() = write\_headers();\\ \;\;\;\!\cdot\!() = httpd\_ssl\_write();)\!\cdot\!err = write(out);\\ k_{2}=t = nondet();\!\cdot\!(b_{21}\!\cdot\!() = write\_headers();\\ \;\;\; + \neg b_{21}\!\cdot\!1)\!\cdot\!err = write(out);\end{array}\right.$.
.4 (Partial), cond: $t > 0$.
.5 $\left\{\begin{array}{l}Cond: b_{21} \\ k_{1}=1\!\cdot\!() = compress();\!\cdot\!() = write\_headers();\\ \;\;\;\!\cdot\!() = httpd\_ssl\_write();\!\cdot\!err = write(out);\\ k_{2}=t = nondet();\!\cdot\!(b_{21}\!\cdot\!() = write\_headers();\\ \;\;\; + \neg b_{21}\!\cdot\!1)\!\cdot\!err = write(out);\end{array}\right.$.
.6 AComplete.
.7 $\left\{\begin{array}{l}Axioms: \{K=1, M=1, P=1\}\\ k_{1}=1\!\cdot\!() = compress();\!\cdot\!() = write\_headers();\\ \;\;\;\!\cdot\!() = httpd\_ssl\_write();\!\cdot\!err = write(out);\\ k_{2}=t = nondet();\!\cdot\!1\!\cdot\!() = write\_headers();\\ \;\;\;\!\cdot\!err = write(out);\end{array}\right.$.
.3 $\left\{\begin{array}{l}Cond: \neg a_{12} \\ k_{1}=(a_{12}\!\cdot\!() = compress();\!\cdot\!() = write\_headers();\\ \;\;\;\!\cdot\!() = httpd\_ssl\_write(); + \neg a_{12}\!\cdot\!() = write\_headers();\\ \;\;\;\!\cdot\!() = httpd\_ssl\_write();)\!\cdot\!err = write(out);\\ k_{2}=t = nondet();\!\cdot\!(b_{21}\!\cdot\!() = write\_headers();\\ \;\;\; + \neg b_{21}\!\cdot\!1)\!\cdot\!err = write(out);\end{array}\right.$.
.4 (Partial), cond: $t > 0$.
.5 $\left\{\begin{array}{l}Cond: b_{21} \\ k_{1}=1\!\cdot\!() = write\_headers();\\ \;\;\;\!\cdot\!() = httpd\_ssl\_write();\!\cdot\!err = write(out);\\ k_{2}=t = nondet();\!\cdot\!(b_{21}\!\cdot\!() = write\_headers();\\ \;\;\; + \neg b_{21}\!\cdot\!1)\!\cdot\!err = write(out);\end{array}\right.$.
.6 AComplete.
.7 $\left\{\begin{array}{l}Axioms: \{K=1, P=1\}\\ k_{1}=1\!\cdot\!() = write\_headers();\\ \;\;\;\!\cdot\!() = httpd\_ssl\_write();\!\cdot\!err = write(out);\\ k_{2}=t = nondet();\!\cdot\!1\!\cdot\!() = write\_headers();\\ \;\;\;\!\cdot\!err = write(out);\end{array}\right.$.
}

\end{scriptsize}
\medskip
\emph{Remaining 60 solutions ommitted for brevity.}

\end{document}